\documentclass[11pt,letter,reqno]{amsart}

\usepackage{paralist, graphicx, verbatim}
\usepackage[margin=1in]{geometry}
\usepackage{tikz-cd}
\usepackage{iftex}
\usepackage{microtype}
\usepackage[mathscr]{euscript}
\usepackage[shortlabels]{enumitem}
\usepackage{background}

\renewcommand{\subset}{\subseteq}
\newcommand{\eqdef}{\stackrel{\mathrm{def}}{=}}
\newcommand{\CAT}{\mathrm{CAT}}

\ifpdftex
   \usepackage{fourier,amssymb}
   \PassOptionsToPackage{bibencoding=utf8, texencoding=ascii}{biblatex}
\else
\usepackage{fourier-otf}
\fi

\makeatletter
\def\resetMathstrut@{%
  \setbox\z@\hbox{%
    \mathchardef\@tempa\mathcode`\(\relax
    \def\@tempb##1"##2##3{\the\textfont"##3\char"}%
    \expandafter\@tempb\meaning\@tempa \relax
  }%
  \ht\Mathstrutbox@1.2\ht\z@ \dp\Mathstrutbox@1.2\dp\z@
}
\makeatother

\usepackage[url=false, isbn=false, related=false, style = alphabetic, sorting = nyt, backref = true, backend =biber, maxnames = 10,minnames = 2] {biblatex}
\usepackage[hidelinks]{hyperref} %
\usepackage{aliascnt}

\newtheorem{theorem}{Theorem}[section]
\newaliascnt{corollary}{theorem}
\newtheorem{corollary}[corollary]{Corollary}
\aliascntresetthe{corollary}

\newaliascnt{proposition}{theorem}
\newtheorem{proposition}[proposition]{Proposition}
\aliascntresetthe{proposition}

\newaliascnt{claim}{theorem}
\newtheorem{claim}[claim]{Claim}
\aliascntresetthe{claim}

\newaliascnt{observation}{theorem}
\newtheorem{observation}[observation]{Observation}
\aliascntresetthe{observation}

\newaliascnt{lemma}{theorem}
\newtheorem{lemma}[lemma]{Lemma}
\aliascntresetthe{lemma}

\newaliascnt{conjecture}{theorem}

\aliascntresetthe{conjecture}

\theoremstyle{definition}

\newaliascnt{question}{theorem}
\newtheorem{question}[question]{Question}
\aliascntresetthe{question}

\newaliascnt{definition}{theorem}
\newtheorem{definition}[definition]{Definition}
\aliascntresetthe{definition}

\newaliascnt{remark}{theorem}
\newtheorem{remark}[remark]{Remark}
\aliascntresetthe{remark}

\renewcommand*{\eqref}[1]{\hyperref[#1]{(\ref*{#1})}}

\DeclareMathOperator{\Log}{Log}
\DeclareMathOperator{\cone}{\mathsf{Cone}}
\DeclareMathOperator{\diam}{diam}
\DeclareMathOperator{\girth}{girth}
\DeclareMathOperator{\supp}{supp}
\DeclareMathOperator{\estimate}{\mathsf{Estimate}}
\DeclareMathOperator{\pair}{\mathsf{Pair}}
\DeclareMathOperator{\query}{Query}
\newcommand{\R}{\mathbb R}
\DeclareMathOperator{\E}{\mathbb{E}}
\newcommand{\dd}{\mathsf{d}}
\DeclareMathOperator{\sgn}{sgn}
\newcommand{\N}{\mathbb{N}}
\newcommand{\n}{\{1,\ldots,n\}}
\newcommand{\Z}{\mathbb{Z}}
\newcommand{\e}{\varepsilon}
\newcommand{\gr}{\mathsf{g}}
\newcommand{\F}{{\mathcal{F}}}
\newcommand{\X}{{\mathcal{X}}}

\newcommand{\RRG}{{\mathbb{X}}}
\newcommand{\sub}{\mathscr{C}}
\newcommand{\sfG}{\mathsf{G}}
\newcommand{\sfE}{\mathsf{E}}
\newcommand{\sfV}{\mathsf{V}}

\DeclareMathOperator{\avg}{avg}
\newcommand{\ud}{\underline{d}}
\newcommand{\od}{\overline{d}}
\newcommand{\DA}{\mathsf{DA}}
\newcommand{\SE}{\smash{\sqrt{\mathcal{E}}}}
\newcommand{\embed}{\text{\scalebox{0.52}[1]{$\hookrightarrow$}}}
\newcommand{\zigzag}{{\text{\textcircled z}}}
\newcommand{\oz}{\zigzag}
\newcommand{\A}{\mathcal{A}}
\newcommand{\circr}{{\text{\textcircled r}}}
\newcommand{\K}{\mathcal{K}}

\renewcommand{\le}{\leqslant}
\renewcommand{\ge}{\geqslant}
\renewcommand{\leq}{\leqslant}
\renewcommand{\geq}{\geqslant}
\newcommand{\MM}{\mathcal{M}}
\newcommand{\NN}{\mathcal{N}}
\newcommand{\cF}{\mathcal{F}}
\newcommand{\f}{\varphi}
\newcommand{\cc}{\mathsf{c}}
\renewcommand{\setminus}{\smallsetminus}

\newcommand{\mnote}[1]{$\ll$\textsf{\textcolor{red}{#1}} ---MM$\gg$}
\newcommand{\anote}[1]{$\ll$\textsf{\textcolor{red}{#1}} ---AE$\gg$}
\renewcommand{\mnote}[1]{}
\renewcommand{\anote}[1]{}
\usepackage{dutchcal}

\title{An optimal algorithm  for average distance in typical regular graphs}

\begin{document}

\setcounter{page}{0}

\author{Alexandros Eskenazis}
\address{CNRS, Inst.\  de Math.\  de Jussieu, Sorbonne U., France}
\email{alexandros.eskenazis@imj-prg.fr}
\urladdr{https://www.alexandroseskenazis.com}

\author{Manor Mendel}
\address{Dep.\ of Math.\ and Comp.\ Sci., The Open U.\  of Israel, Israel}
\email{manorme@openu.ac.il}
\urladdr{https://sites.google.com/site/mendelma}

\author{Assaf Naor}
\address{Math. Dep., Princeton U., Princeton, NJ, USA}
\email{naor@math.princeton.edu}
\urladdr{https://web.math.princeton.edu/~naor/}

\backgroundsetup{contents={}}

\subjclass[2020]{%
30L05, 
30L15, 
46B85, 
53C23, 
05C82, 
68R10, 
68R12  
}

\keywords{%
CAT(0) spaces, metric transforms, bi-Lipschitz embeddings, Euclidean cones,  
expanders, Poincaré inequalities, average-distance approximation, distortion growth}

\begin{abstract}
We design a deterministic algorithm that, given $n$ points in a \emph{typical} constant degree regular~graph, queries  $O(n)$
distances to output a constant factor approximation to the average distance among those points, thus answering a question posed in~\cite{MN14}.
Our algorithm uses the method of~\cite{MN14} to construct a sequence of constant degree graphs that are expanders with respect to certain nonpositively curved metric spaces, together with a new rigidity theorem for metric transforms of nonpositively curved metric spaces. The fact that our algorithm works for typical (uniformly random) constant degree regular  graphs rather than for all constant degree graphs is unavoidable, thanks to the following impossibility result that we obtain: For every fixed $k\in \N$,  the approximation factor of  any algorithm for average distance that works for all constant degree graphs and queries $o(n^{1+1/k})$ distances must necessarily be  at least $2(k+1)$.
This matches the upper bound attained by the  algorithm that was designed for general finite metric spaces in~\cite{BGS}. Thus, any  algorithm for average distance in constant degree graphs whose approximation guarantee is less than $4$ must query $\Omega(n^2)$ distances, any such algorithm whose approximation guarantee is less than $6$ must query $\Omega(n^{3/2})$ distances, any such algorithm whose approximation guarantee less than $8$ must query $\Omega(n^{4/3})$ distances, and so forth, and furthermore there exist algorithms achieving those parameters.
\end{abstract}

\maketitle

\setcounter{tocdepth}{1} 

\tableofcontents

\vfill

\autoref{part:EA} is an extended abstract that will appear in the proceedings of the 37th ACM-SIAM Symposium on Discrete Algorithms (SODA 2026). 
\autoref{part:FV} is the draft submitted to SODA 2026 of a longer paper that contains, in particular, the proofs
of all the new results appearing in \autoref{part:EA}.

\thispagestyle{empty}

\newpage

\part{Extended Abstract}
\label{part:EA}

\section{Introduction}
\label{EA:sec:Intro}

 Indyk systematically investigated~\cite{Indyk-sublinear} fast algorithms for geometric computations involving distances. One of the main problems that was featured in~\cite{Indyk-sublinear} is approximating the average distance\footnote{Throughout the ensuing text, we will use the common notations $[n]=\n$ and $\tbinom{[n]}{2}=\{\{a,b\}\subset [n]:\ a\neq b\}$ for $n\in \N$.}
\begin{equation}\label{EA:eq:average dist}
\frac{1}{n^2}\sum_{(i,j)\in [n]^2} d_\MM(x_i,x_j)
\end{equation}
among $n$ points $x_1,\ldots,x_n\in \MM$ in a given metric space $(\MM,d_\MM)$.
As there are $n^2$ summands in~\eqref{EA:eq:average dist}, such an algorithm is considered sublinear if it performs $o(n^2)$ distance queries. Obtaining a finite approximation factor entails making at least $n-1$ distance queries since otherwise the graph whose edges are the pairs of points that were queried would be disconnected,  so by varying the distances between its connected components one can obtain two markedly different metric spaces that the algorithm cannot distinguish.

Here, we will study deterministic algorithms for the question stated above.
In combination with the mathematics that we will develop for that purpose (belonging to metric embeddings and Alexandrov geometry), this investigation will lead to the resolution of a range of questions that were left open in the literature. In order to explain these matters, it is beneficial to first ground the discussion (and set notation/terminology) by recalling the following formal definition of the computational model that will be treated herein.
Its generality entails that the lower bound that we will obtain is quite a strong statement, while the algorithm that we will design belongs to the least expressive aspect of this framework, namely, it will be non-adaptive and, in fact, it will be a ``universal approximator'' per terminology from~\cite{BGS} that we will soon recall.

\begin{definition}
\label{EA:def:adapt-approximator} Let $\mathcal{F}$ be a family of metric spaces. For $n,m\in \N$ and $\alpha\geq 1$, a deterministic $\alpha$-approximation algorithm for the average distance in $\mathcal{F}$ that makes $m$ distance queries on inputs of size $n$ consists of functions
\begin{equation}\label{EA:eq:functions for the game}
\Big\{\mathsf{Pair}_i:\Big(\tbinom{[n]}{2}\times [0,\infty)\Big)^{i-1}  \to
\tbinom{[n]}{2}\Big\}_{i=1}^m\qquad\mathrm{and}\qquad \estimate: \Big(\tbinom{[n]}{2}\times [0,\infty)\Big)^m\to [0,\infty),
\end{equation}
where $\pair_1$ is understood to be a fixed element $\pair_1=\{a_1,b_1\}\in \tbinom{[n]}{2}$ that satisfies the following property.

Given a metric space $(\MM,d_\MM)\in \cF$ and $x_1,\ldots,x_n\in \MM$, define  $\{a_1,b_1\},\{a_2,b_2\},\ldots,\{a_m,b_m\}\in \tbinom{[n]}{2}$ inductively (recalling that $\{a_1,b_1\}$ has already been fixed above) by setting, for every $i\in [m-1]$,
\begin{equation}\label{EA:eq:approximation game}
\{a_{i},b_{i}\}\eqdef \pair_{i} \Big(\big(\{a_1,b_1\},d_\MM(x_{a_1},x_{b_1})\big), \big(\{a_2,b_2\},d_\MM(x_{a_2},x_{b_2})\big),\ldots, \big(\{a_{i-1},b_{i-1}\},d_\MM(x_{a_{i-1}},x_{b_{i-1}})\big)\Big).
\end{equation}
We then require that%
\footnote{If the denominator in~\eqref{EA:eq:alpha approx def} vanishes (i.e., when $x_1=\ldots=x_n$),  then we require that the numerator in~\eqref{EA:eq:alpha approx def} also vanishes.}
\begin{equation}\label{EA:eq:alpha approx def}
1\le \frac{\estimate \Big(\big(\{a_1,b_1\},d_\MM(x_{a_1},x_{b_1})\big),\ldots, \big(\{a_{m},b_{m}\},d_\MM(x_{a_{m}},x_{b_{m}})\big)\Big)}{\frac{1}{n^2}\sum_{(i,j)\in [n]^2} d_\MM(x_i,x_j)} \le \alpha.
\end{equation}

A \emph{non-adaptive}  deterministic $\alpha$-approximation algorithm for average distance in $\mathcal{F}$ which makes $m$ distance queries on size $n$ inputs, is the restrictive case of the above setup in which $\pair_1,\ldots,\pair_m$  are constant functions, so their images are, respectively, $m$ pairs of indices $\{a_1,b_1\},\ldots,\{a_m,b_m\}\in \tbinom{[n]}{2}$. The output is then a function of the $d_\MM$-distances between the pairs $\{x_{a_1},x_{b_1}\},\ldots, \{x_{a_m},x_{b_m}\}\subset \{x_1,\ldots,x_n\}$ of points whose indices  were decided upfront and not adapted to the specific input $x_1,\ldots,x_n\in \MM$.

When we discuss $\alpha$-approximation algorithms for average distance that make $m$ distance queries on inputs of size $n$   without specifying the underlying family $\cF$, we will tacitly mean that $\cF$ consists of all metric spaces; as any finite metric space is isometric to a subset of $\ell_\infty$, we may take, in this case, $\cF=\{\ell_\infty\}$.
\end{definition}
The standard interpretation of~\eqref{EA:eq:approximation game} is that the algorithm performs the following ``approximation game'' in $m$ rounds.
Starting by querying the $d_\MM$-distance between $x_{a_1}$ and $x_{b_1}$, based only on the answer it receives, the algorithm computes a new pair of points $x_{a_2},x_{b_2}$ and queries the $d_\MM$-distance between them.
In round $i+1$, the algorithm selects a new pair of points $x_{a_{i+1}},x_{b_{i+1}}$ and queries their $d_\MM$ distance, where that pair is a function of only the (ordered) sequence of pairs queried by the algorithm and their $d_\MM$ distances (the responses to the queries) in the preceding $i$ rounds.
After $m$ rounds, the algorithm outputs an estimate for the average distance, which is required to be a function of only the (ordered) sequence of the pairs it queried and the responses to those queries throughout this procedure.

Although the new algorithmic results obtained herein only concern deterministic algorithms, we wish to mention randomized algorithms when referring to previous results in the literature.
Thus, in the context of \autoref{EA:def:adapt-approximator}, a randomized $\alpha$-approximation algorithm for the average distance in $\mathcal{F}$ that makes $m$ distance queries on inputs of size $n$ is the same as in \autoref{EA:def:adapt-approximator}, except that the functions in~\eqref{EA:eq:functions for the game} have an additional variable $\omega$ from some probability space $(\Omega,\mathbb{P})$, and it is required that the approximation guarantee~\eqref{EA:eq:alpha approx def} holds with probability at least, say,  $2/3$. In the non-adaptive special case, $\pair_1,\ldots,\pair_m$ depend only on $\omega$.
In this context, $\cF$ is understood to be all metric spaces if it is not specified.

Indyk considered the simple randomized non-adaptive algorithm that is obtained by averaging over a uniformly random sampling of $O(n/\e^{7/2})$ pairs for some $0<\e<1$. He proved that this straightforward procedure yields a $1+\e$ factor approximation with constant probability.
Barhum, Goldreich, and Shraibman~\cite{BGS} improved Indyk's analysis to
$O(n/\e^2)$ queries.
Goldreich and Ron~\cite{GR06,GR-avg} studied a restricted version of the above problem in which one wishes to approximate
the average distance among \emph{all} vertices of a given \emph{unweighted connected} graph on $n$ vertices, showing that averaging over a uniformly random sampling of $O(\sqrt{n}/\e^{2})$ pairs of vertices yields a $(1+\e)$-approximation algorithm with constant probability.%
\footnote{This algorithm depends crucially on estimating the average over all pairs of vertices and the graph being unweighted, whereas in our model there is $\Omega(n)$ lower bound even for adaptive randomized algorithms.}

Deterministic algorithms for average distance estimation were broached in~\cite{BGS}, where an especially simple example of non-adaptive algorithms, called
\emph{universal approximators}, was studied. Per Definition~6 in~\cite{BGS}, given $\alpha\ge 1$, an $\alpha$-universal approximator in a family $\cF$ of metric spaces  of size $m$ for  inputs of size $n$ consists of  $m$ unordered pairs of indices  $\{a_1,b_1\},\ldots,\{a_m,b_m\}\in \tbinom{[n]}{2}$ and a scaling factor  $\sigma>0$ such that for every metric space $(\MM,d_\MM)\in \cF$ and every $x_1,\ldots,x_n\in \MM$  the output of the algorithm is
$$
\estimate_{\mathrm{BGS}}\big(d_\MM(x_{a_1},x_{b_1}),\ldots, d_\MM(x_{a_m},x_{b_m})\big)\eqdef \sigma\sum_{\ell=1}^m d_\MM(x_{a_\ell},x_{b_\ell}),
$$
 and~\eqref{EA:eq:alpha approx def}
holds.
That is, $n^{-2}\sum_{(i,j)\in [n]^2}d_\MM(x_i,x_j)\le \sigma\sum_{\ell=1}^md_\MM(x_{a_\ell},x_{b_\ell})\le \alpha n^{-2}\sum_{(i,j)\in [n]^2}d_\MM(x_i,x_j)$ for every metric space $(\MM,d_\MM)\in\cF$ and every $x_1,\ldots,x_n\in \MM$. As before, if $\cF$ is not mentioned, then it will be assumed tacitly to be all the possible metric spaces.
The following theorem was proved in~\cite{BGS}:
\begin{theorem}
\label{EA:thm:BGS-approximator} For $k,n\in \N$ with $k\ge 2$ there is an $n$-vertex $(2k)$-universal approximator of  size $O(kn^{1+1/k})$.
\end{theorem}
In particular, \autoref{EA:thm:BGS-approximator} shows that for every fixed integer $k\ge 2$ there is a (nonadaptive) deterministic $(2k)$-approximation algorithm for the average distance that makes $O(n^{1+1/k})$ distance queries on inputs of size $n$. \autoref{EA:thm:impossibility-DA} below, which is one of the main results of the present work, provides a matching impossibility statement.
It should be noted that, in addition to proving~\autoref{EA:thm:BGS-approximator},
a (non-matching) lower bound on the size of universal approximators was proved in~\cite{BGS}.
Ruling out any algorithm whatsoever (in the aforementioned oracle model) rather than only universal approximators is conceptually an entirely different matter.
In fact, the concrete adversarial metric / widget that is used in~\cite{BGS} fails to fool even some non-adaptive algorithms, while our proof of \autoref{EA:thm:impossibility-DA} builds an implicit adversarial metric (algorithm-dependent) as a solution to an auxiliary optimization problem.
An overview of the proof of \autoref{EA:thm:impossibility-DA}  appears in \autoref{sec:hardness overview} below, and its complete details are presented in~\cite{EMN-hadamard}.

\begin{theorem} \label{EA:thm:impossibility-DA}
Fix $k\in\N$ and $\alpha\ge 1$. If for all $n\in\N$, there is a deterministic $\alpha$-approximation algorithm for the average distance that makes $o(n^{1+1/k})$ distance queries on inputs of size $n$, then necessarily
$\alpha\geq 2(k+1)$.
\end{theorem}

By combining \autoref{EA:thm:BGS-approximator}  and \autoref{EA:thm:impossibility-DA}, we obtain the following description of the approximation landscape for the average distance.
Any algorithm for the average distance with an approximation guarantee less than $4$ must query $\Omega(n^2)$ distances;
any such algorithm with an approximation guarantee less than $6$ must query $\Omega(n^{3/2})$ distances; any such algorithm whose approximation guarantee is less than $8$ must query $\Omega(n^{4/3})$ distances, etc. Furthermore, there are algorithms that attain these parameters. See \autoref{EA:fig-ub-lb-chart} for a graphical representation of the lower bound versus the upper bound.

\begin{figure}[htb]
 \begin{center}
\includegraphics[width=0.5\textwidth]{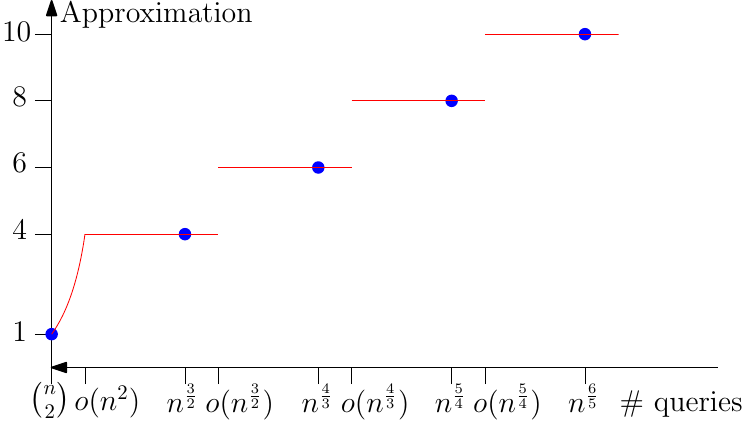}
 \end{center}
 \caption{The blue dots correspond to the available algorithms per \autoref{EA:thm:BGS-approximator}.
 The red lines  correspond to the impossibility statement of  \autoref{EA:thm:impossibility-DA}. 
 Notice that the graph's horizontal axis
 decreases as one moves from left to right, and it introduces a macroscopic gap between $O(n^{1+1/k})$ queries and $o(n^{1+1/k})$ queries.
 }
\label{EA:fig-ub-lb-chart}
\end{figure}

\begin{remark}
Similar-looking (but not identical!) tight approximation-to-size trade-offs are postulated as famous conjectures on related problems, such as spanners~\cite{AGDJS93} and approximate distance oracles~\cite{TZ05}. However, the tight trade-off that we obtain here is an unconditional theorem because, in contrast to the aforementioned examples, our lower bound does not assume the validity of the long-standing Erd\H{o}s girth conjecture~\cite{Erdos-girth-conj}. In our opinion, this aspect of \autoref{EA:thm:impossibility-DA} should be further explored.
\end{remark}

Thanks to \autoref{EA:thm:impossibility-DA}, we now know that there does not exist a deterministic constant approximation algorithm for the average distance that makes $O(n)$ distance queries on inputs of size $n$.
By the simple fact contained in \autoref{EA:prop:reg-universal} below, this impossibility result implies the same statement for constant approximation algorithms for the average distance in $\mathsf{Reg}(d)$, which makes $O(n)$ distance queries on inputs of size $n$.
Here for each $d\in \{3,4,\ldots,\}$ we let $\mathsf{Reg}(d)$ denote the family of metric spaces that consists of all finite connected $d$ regular graphs, equipped with their shortest-path metric.
A detailed (straightforward) proof of \autoref{EA:prop:reg-universal}  appears in~\cite{EMN-hadamard}.
\begin{proposition}\label{EA:prop:reg-universal}
For every $\e>0$ and every $d\in\{3,4,\ldots\}$, any finite metric space embeds with distortion $1+\e$ into some finite connected $d$-regular graph, equipped with its shortest-path metric.
\end{proposition}

The (standard; see, e.g.,~\cite[Chapter~15]{Mat02}) terminology that is used in \autoref{EA:prop:reg-universal} (as well as throughout the ensuing text) is that an embedding $f:\MM\to \NN$ of a metric space $(\MM,d_\MM)$ into a metric space $(\NN,d_\NN)$
has (bi-Lipschitz) distortion at most $D\ge 1$
if there exists (a scaling factor) $s>0$ such that
\begin{equation}\label{eq:def distortion}
 \forall x,y\in \MM,\qquad d_\MM(x,y)\le sd_\NN\big(f(x),f(y)\big) \le Dd_\MM(x,y).
\end{equation}
The $\NN$-distortion of $\MM$ is a numerical invariant that is commonly denoted $\cc_\NN(\MM)$, following~\cite{LLR}, and is defined as the infimum over those $D\ge 1$ such that there exists an embedding $f:\MM\to \NN$ whose distortion is at most $D$. Sometimes, one also needs to use the notations $\cc_{(\NN,d_\NN)}(\MM,d_\MM)$ or $\cc_{\NN}(\MM,d_\MM)$ instead of $\cc_\NN(\MM)$ when the corresponding underlying metrics are not clear from the context.
Such situations in which multiple metrics on the same space must be considered, and will indeed occur below.
When $\MM$ is finite and $p\ge 1$, it is common to use the shorthand  $\cc_p(\MM)=\cc_{\ell_{p}}(\MM)$.

It is possible to bypass the aforementioned hardness by considering deterministic approximation algorithms for the average distance that work when the input is allowed to be an arbitrary $n$-tuple of vertices in a {\em typical} finite connected regular graph rather than in any such graph. The precise formulation of this phenomenon appears in the following theorem, which is another main result of the present work:

\begin{theorem}\label{EA:thm:main positive} For  $N\in \N$ and $d\in \{3,4,\ldots,\}$, let $\mathcal{G}_{N,d}$ be the set of all the (isomorphism classes of) $N$-vertex $d$-regular graphs. There exists $\cF_{N,d}\subset \mathcal{G}_{N,d}$ that consists of connected graphs (equipped with their shortest path metric) that satisfy $\lim_{N\to \infty} |\cF_{N,d}|/ |\mathcal{G}_{N,d}|=1$, such  that the family $\cF_d=\bigcup_{N=1}^\infty \cF_{N,d}$ of metric spaces has the following property. There exists a deterministic algorithm that takes as input $n\in \N$ and outputs in $O(n)$ time $m$ pairs $\{a_1,b_1\},\ldots, \{a_m,b_m\}\subset [n]$, where  $m=O(n)$, such that\footnote{We will use the following (standard) conventions for asymptotic notation, in addition to the usual $O(\cdot),\Omega(\cdot),\Theta(\cdot)$ notation. Given $a,b>0$, by writing
$a\lesssim b$ or $b\gtrsim a$ we mean $a\le \kappa b$ for a
universal constant $\kappa>0$, and $a\asymp b$
stands for $(a\lesssim b) \wedge  (b\lesssim a)$.  }
\begin{equation}\label{eq:super expander random}
\forall \mathsf{G}=(\mathsf{V}_{\!\mathsf{G}},\mathsf{E}_{\mathsf{G}})\in \cF_d,\ \forall x_1,\ldots,x_n\in \mathsf{V}_{\!\mathsf{G}},\qquad \frac{1}{m}\sum_{\ell=1}^m d_{\mathsf{G}}(x_{a_\ell},x_{b_\ell})\asymp \frac{1}{n^2}\sum_{(i,j)\in [n]^2} d_{\mathsf{G}}(x_i,x_j),
\end{equation}
where in~\eqref{eq:super expander random}, and throughout what follows, $d_\mathsf{G}$ denotes the shortest-path metric on a connected graph $\mathsf{G}$. In particular, there exists a  non-adaptive  deterministic $O(1)$-approximation algorithm for the average distance in $\mathcal{F}_d$  which makes $O(n)$ distance queries on inputs of size $n$ (in contrast to \autoref{EA:thm:impossibility-DA}).
\end{theorem}

The precursor to \autoref{EA:thm:main positive} is~\cite[Theorem~1.2]{MN14} or~\cite[Theorem~2.1]{MN-expanders2}, providing the same statement as \autoref{EA:thm:main positive} except that the distances appearing in~\eqref{eq:super expander random} are {\em squared}.
The obvious question if \autoref{EA:thm:main positive} holds as stated above for approximating the actual average distance (in accordance with the prior literature on this topic) rather than approximating the average quadratic distance was a central problem that was left open in~\cite{MN14,MN-expanders2}.
Now, it is resolved (positively) by \autoref{EA:thm:main positive}.

\begin{remark}\label{rem:bf} Even though the impossibility of a constant approximation with  $m=O(n)$ queries that \autoref{EA:thm:impossibility-DA}  establishes  was not known at the time of writing of~\cite{MN14,MN-expanders2}, in hindsight we see that~\cite{MN14,MN-expanders2} proposed a different way to cope with the nonexistence of deterministic algorithms. The aforementioned sublinear  average distance approximation problem falls into the category of multiple known algorithmic tasks (in an oracle model) for which constant factor approximation is provably impossible for deterministic algorithms (with suitable limitations  on the number of oracle queries), yet it is achievable using randomized algorithms. A notable instance of this type of situation is the problem of approximating the volume of a convex body in $\R^n$ that is given by a weak membership oracle (see~\cite{GLS93} for the relevant terminology), for which of B\'ar\'any and F\"uredi established~\cite{BF87} the impossibility of constant-factor approximation in oracle-polynomial time, yet Dyer, Frieze and Kannan  proved~\cite{DFK91} that such a randomized algorithm does exist (a curiosity of this analogy is that the chronology is reversed for average distance approximation, namely, in that setting a randomized algorithm was known to exist~\cite{Indyk-sublinear} before hardness for deterministic algorithms was established). The idea of~\cite{MN14,MN-expanders2}, which is taken up by \autoref{EA:thm:main positive}, is that when it is known that there is no deterministic algorithm, instead of asking for a randomized algorithm that succeeds (with high probability) on all possible inputs, one could still demand that the algorithm will be deterministic but require it to efficiently provide  a  good approximation  on {\em every} sub-structure of a typical input.   In the present setting ``typical'' is understood to be uniformly random and ``sub-structure'' is an  arbitrary subset, but one could conceive of other versions of  such phenomena  to investigate (e.g.~typical inputs could arise from various probabilistic  models that are thought to mimic real-world instances). For example, in the above setting of volume computation, one could  consider for some integers $1\le n\le N$ a random convex body  $K\subset \R^N$ (a well-studied example of this type is Gluskin bodies~\cite{Glu81} and variants thereof~\cite{STJ04}, which are convex hulls of randomly selected points in $\R^N$, but there are many more such possibilities in the literature), and then demand that a deterministic algorithm yields a good approximation for the $n$-dimensional  volume of {\em every} $n$-dimensional section of $K$. It could be a worthwhile research direction to study in multiple settings such deterministic algorithmic tasks that are required to  succeed on every sub-structure of a random input.
\end{remark}

Passing from squared distances to the distances themselves is not merely a technical matter that was left unresolved in the literature; indeed, it was speculated in~\cite{MN14,MN-expanders2} (specifically, see Remark~2.3 there) that achieving this would require a substantially new idea beyond those of~\cite{MN14,MN-expanders2}, because the proof in~\cite{MN14,MN-expanders2} of the aforementioned quadratic variant makes heavy use of a metric space-valued martingale inequality\footnote{A  version due to~\cite{MN-barycentric} of Pisier's classical martingale cotype inequality~\cite{Pis75} for (suitably defined) {\em nonlinear} martingales.} which fails to hold (even when the underlying metric space is an interval in the real line)  if  the distances are raised to  power $1$ (or any power $0< p<2$) rather than   power $2$.

The above prediction materializes herein through the geometric rigidity result in \autoref{EA:thm:cat(0)-transform-closure-intro} below.
It provides an ingredient that makes it possible to apply the methodology of~\cite{MN14,MN-expanders2} to obtain \autoref{EA:thm:main positive}.
Although the present extended abstract focuses on the algorithmic context, it should be noted that \autoref{EA:thm:cat(0)-transform-closure-intro}  is of geometric interest in its own right and implies solutions to long-standing questions in metric geometry.
Those will be briefly mentioned in \autoref{rem:math applications} below and thoroughly treated in~\cite{EMN-hadamard}.
Before stating \autoref{EA:thm:cat(0)-transform-closure-intro}, we will next recall a modicum of (standard) background.

We will use throughout the common terminology (see, e.g., the monographs~\cite{Blu53,WW75,DL10,BB12}) that a function $\varphi:[0,\infty)\to [0,\infty)$ is called a \emph{metric transform} if $\varphi$ is nondecreasing, concave and satisfies $\varphi(0)=0$. This ensures that $(\MM,\f\circ d_\MM)$ is a metric space for every metric space $(\MM,d_\MM)$.

A metric space $(\MM,d_\MM)$ is said to be nonpositively curved in the sense of Alexandrov if it is a geodesic metric space, i.e., for any two points  $x$ and $y$ in $\MM$ there exists a constant speed geodesic $\gamma:[0,1]\to \MM$ joining $x$ to $y$ (namely, $d_\MM(\gamma(t),x)=td_\MM(x,y)$ and $d_\MM(\gamma(t),y)=(1-t)d_\MM(x,y)$
 for every $0\le t\le 1$), and   the distances between points along geodesic triangles in $\MM$ are bounded from above by the corresponding distances in a triangle with the same edge lengths that resides in the Euclidean plane $\R^2$. This is equivalent to requiring that for every $x,y,z\in \MM$, if $\gamma:[0,1]\to \MM$ is a geodesic of constant-speed
 that joins $x$ to $y$, then
\begin{equation}\label{eq:def CTA0}
\forall t\in [0,1],\qquad d_\MM(\gamma(t),z)^2\le td_\MM(x,z)^2+(1-t)d_\MM(y,z)^2-t(1-t)d_\MM(x,y)^2.
\end{equation}
Thorough treatments of this important, useful and extensively studied notion appears in, e.g.,~\cite{Jos97,BH99,BBI01,Sturm-NPC,Gro07,Bac14}. Following~\cite{Gro87}, metric spaces that are nonpositively curved in the sense of Alexandrov are also called $\CAT(0)$ spaces (shorthand for ``Cartan, Alexandrov and Toponogov''), and complete $\CAT(0)$ spaces are commonly called Hadamard spaces.

We can now state \autoref{EA:thm:cat(0)-transform-closure-intro},  which is the main geometric contribution of the present work; an outline of its proof is contained in \autoref{sec:1.7} below, and all details appear in~\cite{EMN-hadamard}.
It should be clarified at the outset that even though our proof of \autoref{EA:thm:main positive} relies on \autoref{EA:thm:cat(0)-transform-closure-intro}, the relevance of \autoref{EA:thm:cat(0)-transform-closure-intro} to \autoref{EA:thm:main positive} is roundabout, and readers who are exposed for the first time to this material should not expect it to be obvious how \autoref{EA:thm:cat(0)-transform-closure-intro} could lead to \autoref{EA:thm:main positive}.
The complete details of this deduction appear in~\cite{EMN-hadamard}, and we will indicate some of its ingredients below, but the entirety of the proof is beyond the scope of the present extended abstract.

\begin{theorem} \label{EA:thm:cat(0)-transform-closure-intro}
Let $\varphi:[0,\infty)\to [0,\infty)$ be a metric transform, and let $(\MM,d_\MM)$ be a $\CAT(0)$ metric space. Then, the metric space $(\MM,\varphi\circ d_\MM)$ is embedded with constant distortion in some (other) $\CAT(0)$ space.
\end{theorem}

Our proof of \autoref{EA:thm:main positive}  follows the approach developed in~\cite{MN14,MN-expanders2} to obtain its aforementioned quadratic variant.
\autoref{EA:thm:cat(0)-transform-closure-intro} provides the previously missing ingredient that enables us to implement the ideas of~\cite{MN14,MN-expanders2} for this purpose, ideas that, on their own, failed to produce \autoref{EA:thm:main positive}.
We do not see how to use for this purpose the results of~\cite{MN14,MN-expanders2} as a ``black box,'' so we need to adjust multiple details within the {\em proofs of}~\cite{MN14,MN-expanders2} in order to be able to combine them with the new geometric input provided by \autoref{EA:thm:cat(0)-transform-closure-intro} so as to derive \autoref{EA:thm:main positive}.

The starting point of the reasoning is the important realization of~\cite{BGS} that universal approximators are related to \emph{expander graphs}; indeed, using this connection, \cite{BGS} obtained for any $0<\e<1$ a deterministic non-adaptive  $(1+\e)$-approximation algorithm for the average distance in the family $\cF$ of subsets of Hilbert space that makes $O(n/\e^2)$ queries on inputs of size $n$. The  link between expanders and universal approximators is the point of view from which~\cite{MN14,MN-expanders2} approached the construction of a constant approximation algorithm of size $O(n)$ for the average squared distance in subsets of random regular graphs, with the ``twist'' that now one encounters the need to consider this question from the perspective of the theory of  nonlinear spectral gaps and expanders with respect to metric spaces; this is an extensively studied topic that is recalled in~\cite{EMN-hadamard} (information on this rich mathematical research direction can be found in~\cite{Mat-extrapolation,Wan98,Ozawa,IN05,Lafforgue-type-expanders,
Pisier-interpolation,NS11,MN-barycentric,MN-superexpanders,Naor-comparison,MN-expanders2,Mim15,Nao18,EMN,Laat-Salle,Naor-avg-john,Esk22,Nao24-mixing}, and applications of it to algorithm design can be found in~\cite{MN14,ANNRW18}).
The special expanders that were used in the construction of~\cite{MN14,MN-expanders2} are variants of the zigzag construction of Reingold, Vadhan and Wigderson~\cite{RVW}, relying on its adaptation to the metric space setting in~\cite{MN-superexpanders,MN-barycentric}.

The method of~\cite{MN14,MN-expanders2} relies on tracking how nonlinear spectral gaps with respect to certain metric spaces evolve under a suitable iteration of the zigzag product; this utilized inequalities for nonlinear martingales that are available by~\cite{MN-barycentric} on $\CAT(0)$ spaces (among others).
For that reason, even though the ultimate goal is to obtain a nonlinear spectral gap with respect to (discrete) random regular graphs, in~\cite{MN14,MN-expanders2}  those graphs were considered as subsets of a larger (continuous)  $\CAT(0)$ space on which the aforementioned martingale analysis can be performed.
\autoref{EA:thm:cat(0)-transform-closure-intro} comes to our aid at this point by taking the metric transform $\varphi(t)=\sqrt{t}$ of the aforementioned larger auxiliary metric space, which \autoref{EA:thm:cat(0)-transform-closure-intro} embeds into an even larger $\CAT(0)$ on which one could hope to apply the martingale tools developed in~\cite{MN-superexpanders,MN-barycentric,MN14,MN-expanders2}.
However, this observation itself is insufficient
and a better understanding of the connection between metric transforms of random regular graphs and $\CAT(0)$ spaces is needed; the complete details of this appear in~\cite{EMN-hadamard}.

\begin{question}\label{EA:ques:Kleinberg}
Our proof of \autoref{EA:thm:main positive} uses methods from analysis and geometry (in addition to graph theory and probability), even though its statement is solely about elementary combinatorics. It remains an intriguing challenge to find a combinatorial proof of \autoref{EA:thm:main positive} that does not make such a circuitous excursion to mathematical areas that are (seemingly) distant from the setting of \autoref{EA:thm:main positive}. In addition to its intrinsic interest, this goal is likely to be relevant to related questions that remain open; an example of such a question (posed by J.~Kleinberg) is discussed in~\cite[
§~2]{MN-expanders2}.
\end{question}

\begin{remark}\label{rem:math applications} Even though the present extended abstract is devoted to algorithmic matters, it would be remiss of us not to mention that in~\cite{EMN-hadamard} we derive other applications of \autoref{EA:thm:cat(0)-transform-closure-intro} to pure mathematics. For example, by combining \autoref{EA:thm:cat(0)-transform-closure-intro} with~\cite{Gromov-rand-groups,Kondo,ANN}, we prove that there exists a metric space $(\MM,d_\MM)$ that embeds with constant distortion into a metric space $(\mathcal{X},d_\mathcal{X})$ of nonpositive Alexandrov  curvature, and also $(\MM,d_\MM)$ embeds with constant distortion into a metric space $(\mathcal{Y},d_\mathcal{Y})$ of nonnegative Alexandrov curvature%
\footnote{Nonnegative Alexandrov curvature for a geodesic metric space corresponds to reversing the inequality in~\eqref{eq:def CTA0}.},
yet $(\MM,d_\MM)$ does not admit a coarse embedding into a Hilbert space (see~\cite{Roe03,Gro07,NY12} for background on the well-studied notion of coarse embeddings, though in~\cite{EMN-hadamard} we obtain an even stronger impossibility result).
This is in contrast to a classical result of Wilson~\cite{Wil32} and Blumenthal~\cite{Blu35} (see also~\cite[pages~122--128]{Blu53} or~\cite[Section~7]{Lurie-Hadamard}) which implies that a geodesic metric space that embeds isometrically into both a metric space of nonpositive Alexandrov curvature and a metric space of nonnegative Alexandrov curvature must be isometric to a subset of a Hilbert space. As another example of an application of \autoref{EA:thm:cat(0)-transform-closure-intro},  in~\cite{EMN-hadamard} we resolve (negatively) a question that Ding, Lee and Peres posed in~\cite[Question~1.13]{DLP} by proving that  there exists a  metric space which has Markov type\/~$2$ (per~\cite{Ball}), \emph{sharp} metric cotype\/~$2$ (per~\cite{MN-cotype}), and is Markov\/ $2$-convex (per~\cite{LNP-markov-convexity}), yet, it does not embed coarsely into a Hilbert space.
This rules out (in a strong form) an approach (in accordance with the Ribe program~\cite{Bourgain-superreflexivity,Kal08,Naor-Ribe,Ball-Ribe,Ostrovskii-book,God17,Nao18}) toward a possible metric version of Kwapie\'n's theorem~\cite{Kwapien} that was investigated in~\cite{DLP}.
As a third example of an application of \autoref{EA:thm:cat(0)-transform-closure-intro},  in~\cite{EMN-hadamard} we address an old folklore question in the theory of metric embeddings about the possible $\ell_1$ distortion growth rates, by proving that for any metric transform $\varphi$ there exists a Hadamard space $\MM_\varphi$ for which $\sup\{\cc_1(\sub): \sub\subset \MM_\f\ \wedge |\sub|\le n\}$ is bounded from above and from below by positive universal constant multiples of $\varphi(\log (n+1))/\varphi(1)$. For embeddings into a Hilbert space, we are currently able to obtain this result under the additional assumption that $\varphi(t)\gtrsim \varphi(1)\sqrt{\log t}$.
In particular, we prove that for every $0\le \theta\le 1$ there exists a Hadamard space $\MM_\theta$ for which $\sup\{\cc_2(\sub): \sub\subset \MM_\theta\ \wedge |\sub|\le n\}$ is bounded from above and from below by positive universal constant multiples of $(\log (n+1))^\theta$; this statement was not previously known when $\theta\notin \{0,\frac12,1\}$, even if one relaxes it by allowing $\MM_\theta$ to be any geodesic metric space.
\end{remark}

\section{Overview of the proof of \texorpdfstring{\autoref{EA:thm:impossibility-DA}}{Theorem \ref{EA:thm:impossibility-DA}}}\label{sec:hardness overview}

Due to their length, the complete details of the proof of \autoref{EA:thm:impossibility-DA}---our (optimal) lower bounds for deterministic approximation algorithms for average distance that work for all metric spaces---are deferred to~\cite{EMN-hadamard}.
This section presents an outline of this proof.

The following theorem implies \autoref{EA:thm:impossibility-DA} in a straightforward manner (see below):

\begin{theorem} \label{EA:thm:det-adaptive-lb}
Fix $k\in\N$. Continuing with the notation of \autoref{EA:def:adapt-approximator}, suppose that there is a deterministic $\alpha$-approximation algorithm for the average distance which makes $m=o(n^{1+1/k})$ distance queries on inputs of size $n$. Thus, we are given functions as in~\eqref{EA:eq:functions for the game} which satisfy the properties described in \autoref{EA:def:adapt-approximator} when $\cF$ is the family of all possible metric spaces. Then, there exist
\begin{equation}\label{eq:two adverseries}
\od,\ud:[n]^2\to [0,k+1]\qquad\mathrm{and}\qquad  \{a_1,b_1\},\ldots,\{a_m,b_m\}\in \tbinom{[n]}{2},
\end{equation}
with the following properties:
\begin{enumerate}[label=(\Roman*)]
\item Both $\od$ and $\ud$ are metrics on $[n]$ (of diameter at most $k+1$); \label{EA:it:(1)}
    \item For every $i\in [m]$ we have $\od(a_i,b_i)=\ud(a_i,b_i)\eqdef w(a_i,b_i)$; \label{EA:it:(2)}
    \item \label{EA:it:(3)} The recursive relation~\eqref{EA:eq:approximation game} holds for $(x_1,\ldots,x_n)=(1,\ldots,n)$, i.e.,
    \begin{equation}\label{eq:recursion with w}
\forall i\in [m-1],\qquad \{a_{i+1},b_{i+1}\}=\pair_{i+1} \Big(\big(\{a_1,b_1\},w({a_1},{b_1})\big),\ldots, \big(\{a_{i},b_{i}\},w({a_{i}},{b_{i}})\big)\Big);
\end{equation}
    \item There exists $X\subset [n]^2$ with $|X|=o(n^2)$ such that $\od(x,y)=k+1$ for every $(x,y)\in [n]^2\setminus X$;\label{EA:it:(4)}
    \item There exists $Y\subset [n]^2$ with $|Y|=o(n^2)$ such that $\ud(x,y)=\tfrac12$ for every $(x,y)\in [n]^2 \setminus Y$ with $x\neq y$. \label{EA:it:(5)}
\end{enumerate}
\end{theorem}

We check that \autoref{EA:thm:det-adaptive-lb} indeed implies \autoref{EA:thm:impossibility-DA}, by contrasting { \autoref{EA:it:(1)}}, { \autoref{EA:it:(2)}} and { \autoref{EA:it:(3)}} in the conclusion of \autoref{EA:thm:det-adaptive-lb} with \autoref{EA:def:adapt-approximator} we see that
\begin{align}\label{eq:output same for bith metrics}
\begin{split}
\frac{1}{n^2}\sum_{(i,j)\in [n]^2} \od(i,j)&\stackrel{\eqref{EA:eq:alpha approx def}}{\le} \estimate \Big(\big(\{a_1,b_1\},\od(x_{a_1},x_{b_1})\big),\ldots, \big(\{a_{m},b_{m}\},\od({a_{m}},{b_{m}})\big)\Big)\\
&= \estimate \Big(\big(\{a_1,b_1\},\ud(x_{a_1},x_{b_1})\big),\ldots, \big(\{a_{m},b_{m}\},\ud({a_{m}},{b_{m}})\big)\Big) \stackrel{\eqref{EA:eq:alpha approx def}}{\le} \frac{\alpha}{n^2}\sum_{(i,j)\in [n]^2} \ud(i,j).
\end{split}
\end{align}
At the same time, by { \autoref{EA:it:(4)}} we have
\begin{equation}\label{eq:upper average}
\frac{1}{n^2}\sum_{(i,j)\in [n]^2} \od(i,j)\ge \frac{n^2-|X|}{n^2}(k+1)=\big(1-o(1)\big)(k+1),
\end{equation}
and by {\autoref{EA:it:(4)}} we have (using the fact that the diameter of the metric $\ud$ is at most $k+1$),
\begin{equation}\label{eq:lower average}
\frac{1}{n^2}\sum_{(i,j)\in [n]^2} \ud(i,j)\le \frac{n^2-|Y|}{n^2}\cdot\frac12 +\frac{|Y|}{n^2}\diam\big([n],\ud\big)\le \frac{1-o(1)}{2}+o(1)(k+1).
\end{equation}
By contrasting~\eqref{eq:output same for bith metrics} with~\eqref{eq:upper average} and~\eqref{eq:lower average} and letting $n\to \infty$, we conclude that $\alpha\ge 2(k+1)$, as required.

\smallskip

The remainder of this section will sketch the proof of \autoref{EA:thm:det-adaptive-lb}.
For that and throughout what follows, given an edge-weighted graph $\sfG=\big(\mathsf{V},\mathsf{E},w:\sfE\to [0,\infty)\big)$ we will denote by $d_{\sfG}:\sfV^2\to [0,\infty]$ its associated shortest-path (extended) metric, with the convention that $d_\sfG(x,y)=\infty$  if $x,y\in \sfV$ belong to distinct connected components of $\sfG$. The degree in $\sfG$ of a vertex $x\in \sfV$ will be denoted $\deg_\sfG(x)\in \N\cup\{0\}$.

We will define $\{a_1,b_1\},\ldots,\{a_i,b_i\}\in \tbinom{[n]}{2}$ and a metric $d_i$ on $[n]$ by induction on $i\in [m]$.
Per \autoref{EA:def:adapt-approximator}, the initial pair $\{a_1,b_1\}$ is fixed as part of the algorithm's parameters.
Set $d_1(a_1,b_1)=1$. Assume inductively that $\{a_1,b_1\},\ldots,\{a_{i},b_{i}\}$  and $d_i$ have already been defined for $i\in [m-1]$.

To go from $i$ to $i+1$ in our recursive definition, let the next pair $\{a_{i+1},b_{i+1}\}\in \tbinom{[n]}{2}$ be given by
\begin{equation}\label{eq:recursion with di}
\{a_{i+1},b_{i+1}\}\eqdef \pair_{i+1} \Big(\big(\{a_1,b_1\},d_i({a_1},{b_1})\big),\ldots, \big(\{a_{i},b_{i}\},d_i({a_{i}},{b_{i}})\big)\Big).
\end{equation}
Write
$
\sfE_{i}\eqdef\big\{\{a_1,b_1\},\ldots,\{a_{i},b_{i}\}\big\}
$.
Let $\sfG_{i}\eqdef \big([n],\sfE_{i}, d_{i}\big)$ be the edge-weighted graph whose vertex set is $[n]$, whose edge set is $\sfE_{i}$, and in which the weight of each edge $\{a,b\}\in \sfE_i$ equals $d_i(a,b)$.
Next, define
\begin{align}\label{eq:def hi}
\begin{split}
\forall x\in [n],\qquad h_i(x)\eqdef \min\bigg\{\bigg\lfloor\frac{k}{\log n}&\log\Big(1+\frac{1}{\sqrt{m}}n^{\frac{k+1}{2k}}\deg_{\sfG_i}(x)\Big)\bigg\rfloor,k-1\bigg\}\in \{0,\ldots,k-1\}.
\end{split}
\end{align}
We now define $d_{i+1}(a,b)$ as follows for every $a,b\in [n]$ such that $\{a,b\}\in \big\{\{a_1,b_1\},\ldots,\{a_{i+1},b_{i+1}\}\big\}\eqdef \sfE_{i+1}$:
\begin{equation}\label{eq:def di+1}
d_{i+1}(a,b)\eqdef \max \Big\{\min\big\{\max\{h_{i}(a),h_{i}(b)\}+1,d_{\sfG_{i}}(a,b)\big\}, \max_{\{u,v\}\in \sfE_{i}} \big( d_{i}(u,v)-d_{\sfG_{i}}(u,a)-d_{\sfG_{i}}(v,b)\big)\Big\}.
\end{equation}

Observe that $d_{i+1}(a,b)=d_{i+1}(b,a)$, so we may consider the weighted graph $\sfG_{i+1}\eqdef \big([n],\sfE_{i+1},d_{i+1}\big)$. An elementary induction that is carried out in~\cite{EMN-hadamard}
shows that
\begin{equation}\label{consistent and less than k}
\forall j\in [i+1],\qquad d_{i+1}(a_j,b_j)=d_{\sfG_{i+1}}(a_j,b_j)=d_{\sfG_j}(a_j,b_j)\le k.
\end{equation}
Although $d_{i+1}(a,b)$ was initially defined in~\eqref{eq:def di+1} only for those $(a,b)\in [n]^2$ such that $\{a,b\}\in \sfE_{i+1}$, by~\eqref{consistent and less than k} the following definition provides an extension of $d_{i+1}$ to all of $[n]^2$ which is also a metric on $[n]$:
\begin{equation}\label{eq:definition of di+1 on all of n}
\forall (a,b)\in [n]^2,\qquad d_{i+1}(a,b)\eqdef \min \big\{d_{\sfG_{i+1}}(a,b),k+1\big\}.
\end{equation}
This completes the inductive step from $i$ to $i+1$.

Continuing with the proof of \autoref{EA:thm:det-adaptive-lb}, we will take the metric $\od$ in~\eqref{eq:two adverseries} to be $\od\eqdef d_m$.
That is, it is obtained as the result of the inductive construction
described above in step $i=m$.
It indeed takes values in $[0,k+1]$, as required in~\eqref{eq:two adverseries}, thanks to~\eqref{eq:definition of di+1 on all of n}. The pairs $\{a_1,b_1\},\ldots,\{a_m,b_m\}$ in~\eqref{eq:two adverseries} are also given in step $i=m$.

Writing $w(a_i,b_i)\eqdef\od(a_i,b_i)$ for $i\in [m]$ according to  {\autoref{EA:it:(2)}} of \autoref{EA:thm:det-adaptive-lb}, we see from~\eqref{consistent and less than k} that $d_i(a_i,b_i)=d_{i+1}(a_i,b_i)=\ldots=d_m(a_i,b_i)=w(a_i,b_i)$ for every $i\in [m]$.
Consequently, the desired equality~\eqref{eq:recursion with w} in  {\autoref{EA:it:(3)}} of \autoref{EA:thm:det-adaptive-lb} coincides with~\eqref{eq:recursion with di}.

For the ensuing discussion, it will be convenient to write $\sfE\eqdef \sfE_m$ and $\sfG\eqdef \sfG_m$.
Thus, we have
\begin{equation}\label{eq:rewrite od G}
\forall x,y\in [n],\qquad \od(x,y)=\min\big\{d_\sfG(x,y),k+1\big\}\in \{0,\ldots,k+1\}.
\end{equation}
We will also use the notation $h\eqdef h_m$, and let $B_\sfG(x,r)\eqdef \{y\in [n]:\ d_\sfG(x,y)\le r\}$ be the closed ball with respect to the shortest-path metric $d_\sfG$ on the weighted graph $\sfG$ centered at $x\in [n]$ of radius $r\ge 0$.
The following lemma on the geometry of $\sfG$ is proved in~\cite{EMN-hadamard}:

\begin{lemma}\label{lem:ball growth} The growth rate of the size of $\sfG$-balls  of radius at most $k-1$ satisfies the following:
\begin{equation}\label{eq:growth rate of balls}
\forall (x,r)\in [n]\times [k-1],\qquad |B_\sfG(x,r)|\lesssim \frac{\sqrt{m}}{n^{\frac{k-2r+1}{2k}}}.
\end{equation}
Also, the following relation holds between the size of any $\sfG$-ball of radius $k$ and the $\sfG$-degree of its center:
\begin{equation}\label{eq:k ball small if degree small}
\forall x\in [n],\qquad \deg_\sfG(x)< \frac{\sqrt{m}}{n^{\frac{k+1}{2k}}}(n-1)\implies |B_\sfG(x,k)|\lesssim n^{\frac12-\frac{1}{2k}}\sqrt{m}.
\end{equation}
\end{lemma}

Item {\ref{EA:it:(4)}} of \autoref{EA:thm:det-adaptive-lb} can be quickly deduced from the second part~\eqref{eq:k ball small if degree small} of \autoref{lem:ball growth} for the following choice of $X\subset [n]^2$:
\begin{equation}\label{eq:choose our X}
X\eqdef \big\{(x,y)\in [n]^2:\ \od(x,y)\le k\big\}\stackrel{\eqref{eq:rewrite od G}}{=}\big\{(x,y)\in [n]^2:\ y\in B_\sfG(x,k)\big\}.
\end{equation}
As $\od$ is integer-valued, the first equality in~\eqref{eq:choose our X} gives that if $(x,y)\in [n]^2\setminus X$, then $\od(x,y)\ge k+1$.
Next, set
\begin{equation}\label{eq:def H and small ball on H}
 H\eqdef \bigg\{x\in [n]:\ \deg_\sfG(x)\ge \frac{\sqrt{m}}{n^{\frac{k+1}{2k}}}(n-1)\bigg\}\stackrel{\eqref{eq:k ball small if degree small}}{\implies} \forall x\in [n]\setminus H,\qquad |B_\sfG(x,k)|\lesssim n^{\frac12-\frac{1}{2k}}\sqrt{m}.
\end{equation}
Assuming that $n\ge 2$, we then have
\begin{equation}\label{eq:H small}
m=|\sfE|=\frac12\sum_{x\in [n]} \deg_\sfG(x)\ge  \frac{\sqrt{m}}{2n^{\frac{k+1}{2k}}}(n-1)|H|\ge \frac{n^{\frac{k-1}{2k}}\sqrt{m}}{4}|H|\implies |H|\le \frac{4\sqrt{m}}{n^{\frac{k-1}{2k}}}.
\end{equation}
\begin{equation}\label{eq:X small}
|X|\stackrel{\eqref{eq:choose our X}}{\le}\sum_{x\in [n]}|B_\sfG(x,k)|\le |H|n+n\max_{x\in [n]\setminus H}|B_\sfG(x,k)| \stackrel{\eqref{eq:def H and small ball on H} \wedge \eqref{eq:H small}}{\lesssim} n^{\frac12+\frac{1}{2k}}\sqrt{m} +n^{\frac32-\frac{1}{2k}}\sqrt{m} \asymp n^{\frac32-\frac{1}{2k}}\sqrt{m},
\end{equation}
where the second step of~\eqref{eq:X small}  uses the trivial estimate $|B_\sfG(x,k)| \leq n$ for $x\in H$ and the last step of~\eqref{eq:X small} holds since $k\ge 1$.
Due to the assumption $m=o(n^{1+1/k})$ of \autoref{EA:thm:det-adaptive-lb}, it follows from~\eqref{eq:X small} that $|X|=o(n^2)$. This concludes the verification of {\autoref{EA:it:(4)}} of \autoref{EA:thm:det-adaptive-lb}.

The second metric $\ud$ in~\eqref{eq:two adverseries} is implicitly defined as any minimizer of the (linear) objective function
\begin{equation}\label{eq:ud objective}
\left(d\in[0,\infty)^{[n]^2}\right)\mapsto  \sum_{(x,y)\in [n]^2} d(x,y),
\end{equation}
subject to the following system of (linear) constraints:
\begin{equation}\label{eq: constraints}
\left\{ \begin{array}{ll} \forall x\in [n],\qquad &d(x,x)=0,\\ \forall x,y\in [n],\qquad &d(x,y)=d(y,x),\\ \forall x,y,z\in [n],\qquad &d(x,z)\le d(x,y)+d(y,z), \\ \forall \{x,y\}\in \sfE,\qquad  &d(x,y)= w(x,y),    \\
 \forall x,y \in [n], \qquad &\min\Big\{ \od(x,y), \max\{h(x),h(y)\}+\tfrac12\Big\}\le d(x,y)\le \od(x,y).
\end{array}\right.
\end{equation}
The set of all those $d:[n]^2\to [0,\infty)$ that satisfy~\eqref{eq: constraints} is nonempty since $\od$ belongs to it.
Due to the last constraint in~\eqref{eq: constraints}, this set is compact, so there is
a minimizer $\ud$ of~\eqref{eq:ud objective}.
The first three constraints in~\eqref{eq: constraints} ensure that $\ud$ is a metric on $[n]$.
Since by~\eqref{eq:rewrite od G} we know that $\od$ takes values in $[0,k+1]$, thanks to the last constraint in~\eqref{eq: constraints} we also know that $\ud$ takes values in $[0,k+1]$, as required in~\eqref{eq:two adverseries}.
The fourth constraint in~\eqref{eq: constraints} ensures that  {\autoref{EA:it:(2)}} of \autoref{EA:thm:det-adaptive-lb} holds.
Hence, the proof \autoref{EA:thm:det-adaptive-lb} will be complete if we prove that its conclusion { \autoref{EA:it:(5)}} is true.

\begin{remark} A na\"ive way  to obtain $\ud$ would be to consider a minimizer of~\eqref{eq:ud objective} subject to the first four constraints in~\eqref{eq: constraints},  as this would produce a (pseudo)metric%
\footnote{The remaining requirement $d(x,y)\ge 0$ that is needed for $d$ to be a pseudometric follows by taking $x=z$ in the third constraint in~\eqref{eq: constraints} while applying the first constraint in~\eqref{eq: constraints}.}
for which {\autoref{EA:it:(2)}} of \autoref{EA:thm:det-adaptive-lb} holds,  and it has the smallest possible average distance among all such (pseudo)metrics.
However, it turns out that adding the fifth constraint in~\eqref{eq: constraints} yields further information that facilitates the subsequent analysis.
\end{remark}

Below, closed balls  that are induced by the metrics $\od$ and $\ud$ are denoted
$$
\forall (x,r)\in [n]\times \R,\qquad B_{\od}(x,r)\eqdef \big\{y\in [n]:\ \od(x,y)\le r\big\}\qquad\mathrm{and}\qquad B_{\ud}(x,r)\eqdef \big\{y\in [n]:\ \ud(x,y)\le r\big\}.
$$
It turns out that the following subset of $[n]^2$ satisfies Conclusion~{\em \eqref{EA:it:(5)}} of \autoref{EA:thm:det-adaptive-lb}:
\begin{equation}\label{eq:def Y}
Y\eqdef \big(U\times [n]\big)\cup \big([n]\times U\big)\cup W\cup\ \widetilde{W}\cup \big\{(x,y)\in [n]^2:\ \{x,y\}\in\sfE\big\}\cup  \big\{(x,x):\ x\in [n]\big\},
\end{equation}
where $U\subset [n]$ and $W,\widetilde{W}\subset [n]^2$  in~\eqref{eq:def Y} are defined as follows:
\begin{align}\label{def UW}
\begin{split}
U&\eqdef  \bigcup_{x\in [n]}B_{\ud}\big(x,h(x)-1\big),\\
W&\eqdef \bigcup_{u\in [n]} \Big(B_{\ud}\big(u,h(u)\big)\times \big\{v\in [n]:\ \{u,v\}\in \sfE\ \  \mathrm{and}\ \ w(u,v)=h(u)+1\big\}\Big),\\
\widetilde{W} &\eqdef \bigcup_{u\in [n]} \Big(\big\{v\in [n]:\ \{u,v\}\in \sfE\ \  \mathrm{and}\ \ w(u,v)=h(u)+1\big\}\times B_{\ud}\big(u,h(u)\big)\Big)=\big\{(x,y)\in [n]^2:\ (y,x)\in W\big\}.
\end{split}
\end{align}
To prove that $Y$ satisfies the requirements of {\autoref{EA:it:(5)}} of \autoref{EA:thm:det-adaptive-lb}, one must show that $|Y|=o(n^2)$ and that $\ud(x,y)=\tfrac12$  for distinct $x,y\in [n]$ such that $(x,y)\in [n]^2\setminus Y$.
The latter is the content of \autoref{lem:variational} below,  whose detailed proof appears in~\cite{EMN-hadamard}.
This is a key step in which the choice of $\ud$ as a minimizer of~\eqref{eq:ud objective} is used extensively through variational reasoning, as discussed below.
\begin{lemma}\label{lem:variational} If $x,y\in [n]^2$ and $(x,y)\notin Y$, then $\ud(x,y)\in \bigl\{0,\tfrac12\bigr\}$.
\end{lemma}
The proof of \autoref{lem:variational} is by contradiction.
It starts by taking the distinct $(x,y)\in [n]^2\setminus Y$ with $\ud(x,y)\neq 1/2$, and furthermore $\ud(x,y)$ is {\em maximal} among all those $x,y\in [n]$ with this property.
Observe that the first inequality in the fifth constraint in~\eqref{eq: constraints} ensures that $\ud(u,v)\ge 1/2$ for every distinct $u,v\in [n]$, so we necessarily have $\ud(x,y)> 1/2$.
Recalling the definition~\eqref{eq:def Y} of $Y$, the fact that $(x,y)\notin Y$ implies $x,y\notin U$.
By the definition of $U$ in~\eqref{def UW}, this implies $h(x)=h(y)=0$, since otherwise if, say, $h(x)\ge 1$ (recall that $h\ge 0$ is integer valued), then we have $x\in B_{\ud}(x,h(x)-1)\subset U$.
Hence, the first inequality in the fifth constraint in~\eqref{eq: constraints}  is strict.
Furthermore, since $Y\supseteq \{(a,b)\in [n]^2:\ \{a,b\}\in \sfE\}$ by~\eqref{eq:def Y}, we know that $\{x,y\}\notin \sfE$, so the fourth constraint in~\eqref{eq: constraints} does not apply to $\ud(x,y)$.
Because $\ud$ is a minimizer of~\eqref{eq:ud objective} subject to the system of constraints~\eqref{eq: constraints}, the aforementioned considerations imply that there must exist $z\in [n]\setminus \{x,y\}$ such that $\ud(x,z)=\ud (x,y)+\ud(y,z)$, because otherwise it would have been possible to  reduce $\ud(x,y)$ without changing the rest of the values of $\ud$ while ensuring that all of the constraints in~\eqref{eq: constraints} are still satisfied (this is where  the strictness of the lower bound on $\ud(x,y)$ that appears in the fifth constraint in~\eqref{eq: constraints} is used), which contradicts  the fact that $\ud$ is a minimizer of~\eqref{eq:ud objective}  subject to those constraints.

The above process is iterated in~\cite{EMN-hadamard} to obtain $(u,v)\in[n]^2$ which are endpoints of a discrete geodesic with respect to $\ud$ that contains $x,y$ (thus $\ud(u,x)+\ud(x,y)+\ud(y,v)=\ud(u,v)$), and is maximal with respect to inclusion. As before, one of the lower bounds on $\ud$ that appear in the system of constraints~\eqref{eq: constraints} must hold as equality, but this time, due to the maximality of the aforementioned geodesic whose endpoints are $u$ and $v$,
it cannot be the third constraint in~\eqref{eq: constraints}, i.e., the triangle inequality. Hence, either the first inequality in the fifth constraint in~\eqref{eq: constraints} holds as equality, or $\{u,v\}\in \sfE$, i.e., the fourth constraint in~\eqref{eq: constraints} applies. From here, a case analysis that is performed in~\cite{EMN-hadamard} leads to the desired contradiction.

To conclude the sketch of the proof of\/ \autoref{EA:thm:det-adaptive-lb}, it remains to explain why $|Y|=o(n^2)$.  Observe that
\begin{equation}\label{eq:less than 1/2}
\forall x,y\in [n],\qquad \ud(x,y)< \max\{h(x),h(y)\}+\tfrac12\implies \ud(x,y)=\od(x,y).
\end{equation}
Indeed, \eqref{eq:less than 1/2} is implied by the fifth constraint in~\eqref{eq: constraints}. Now,
\begin{equation}\label{eq:balls of ran h the same}
\forall x\in [n], \forall r\in [h(x)],\qquad B_{\ud}(x,r)\stackrel{\eqref{eq:less than 1/2}}{=}B_{\od}(x,r)\stackrel{\eqref{eq:def hi}\wedge \eqref{eq:rewrite od G}}{=}B_{\sfG}(x,r),
\end{equation}
where the last step of~\eqref{eq:balls of ran h the same} holds since $h(x)\le k-1$ for every $x\in [n]$ by~\eqref{eq:def hi}  and $\od$ is given by~\eqref{eq:rewrite od G}.
The a priori upper bound $h(x)\le k-1$ also allows us to apply the first part~\eqref{eq:growth rate of balls} of \autoref{lem:ball growth}  to deduce using~\eqref{eq:balls of ran h the same} that
\begin{equation}\label{eq:upper bound on ud balls}
\forall x\in [n], \forall r\in [h(x)],\qquad \left|B_{\ud}(x,r)\right|\lesssim \frac{\sqrt{m}}{n^{\frac{k-2r+1}{2k}}}\le n^{\frac12-\frac{3}{2k}}\sqrt{m}.
\end{equation}
Consequently,
\begin{align*}
|Y|&\stackrel{\eqref{eq:def Y}}{\le} 2n|U|+2|W|+2|\sfE|+n\\&\stackrel{\eqref{def UW}}{\lesssim} n\sum_{x\in [n]} \left|B_{\ud}\big(x,h(x)-1\big)\right|+\sum_{u\in [n]} \left|B_{\ud}\big(u,h(u)\big)\right|\deg_\sfG(u)+m\\
&\stackrel{\eqref{eq:upper bound on ud balls}}{\lesssim} n\sum_{x\in [n]} \frac{\sqrt{m}}{n^{\frac{k-2(h(x)-1)+1}{2k}}} +n^{\frac12-\frac{3}{2k}}\sqrt{m}\sum_{u\in [n]} \deg_\sfG(u)+m \displaybreak[0]\\
&\ \asymp n^{\frac12-\frac{3}{2k}}\sqrt{m}\sum_{x\in [n]} n^{\frac{h(x)}{k}} +n^{\frac12-\frac{3}{2k}}m^{\frac32}\\
&\stackrel{\eqref{eq:def hi}}{\le} n^{\frac12-\frac{3}{2k}}\sqrt{m}\sum_{x\in [n]}\Big(1+\frac{1}{\sqrt{m}}n^{\frac{k+1}{2k}}\deg_{\sfG}(x)\Big)+n^{\frac12-\frac{3}{2k}}m^{\frac32}\displaybreak[0]\\
&\ \asymp n^{\frac32 -\frac{3}{2k}}\sqrt{m}+n^{1-\frac{1}{k}}m+n^{\frac12-\frac{3}{2k}}m^{\frac32}\\&\ =o(n^2),
\end{align*}
where we used (twice) the fact that $\sum_{x\in [n]}\deg_\sfG(x)=2m$ and the last step is valid because $m=o(n^{1+1/k})$.

\section{Overview of the proof of
\texorpdfstring{\autoref{EA:thm:cat(0)-transform-closure-intro}}{Theorem \ref{EA:thm:cat(0)-transform-closure-intro}}}
\label{sec:1.7}

The Euclidean cone over a metric space $(\MM,d_\MM)$, denoted $\cone(\MM,d_\MM)$ or $\cone(\MM)$ when the metric is clear from the context, is
defined~\cite{Ber-Cone} as the completion of $(0,\infty)\times \MM$ under the following
metric:
\begin{equation}\label{EA:eq:def cone}
\forall (s,x),(t,y)\in (0,\infty)\times \MM,\qquad d_{\cone(\MM,d_\MM)}\big((s,x),(t,y)\big)\eqdef\sqrt{ {s^2+t^2-2st\cos\big(\min\{\pi,d_X(x,y)\}\big)}}.
\end{equation}
See~\cite{ABN86} and~\cite[Chapter~I.5]{BH99} for a treatment of this useful concept (including an explanation of the nomenclature); in particular, \cite[Proposition~I.5.9]{BH99} proves that~\eqref{EA:eq:def cone} indeed defines a metric.

The following general proposition is a key component of the proof of \autoref{EA:thm:cat(0)-transform-closure-intro}:

\begin{proposition}
\label{EA:prop:transform-in-producr-cones}
For every metric transform $\varphi:[0,\infty)\to [0,\infty)$ there exist coefficients $\alpha_0=\alpha_0^\f\ge 0$ and $\{\alpha_n=\alpha_n^\f\}_{n=1}^\infty ,\{\beta_n=\beta_n^\f\}_{n=1}^\infty\subset (0,\infty)$
with the following property. Suppose that $(\MM,d_\MM)$ is a metric space. There exist $\{z_k\}_{k=0}^\infty\subset \MM$ and $\{r_k\}_{k=1}^\infty\subset (0,\infty)$ such that if we consider the Pythagorean product  $(\mathcal{P},\rho)$, where
$$
\mathcal{P}\eqdef \left\{\big(x_0,(s_1,x_1),(s_2,x_2),\ldots,\big)\in \MM\times \big((0,\infty)\times \MM\big)^{\N}:\ \sum_{k=1}^\infty \beta_{k} d_{\cone\left(\MM,\pi\sqrt{\frac{\alpha_k}{\beta_k}}d_{\MM}\right)}\big((s_k,x_k),(r_k,z_k)\big)^2<\infty
\right\},
$$
and $\rho$ is the metric on $\mathcal{P}\subset \MM\times \big((0,\infty)\times \MM\big)^{\N}$ that is given by setting
$$
\rho(\chi,\upsilon)\eqdef \sqrt{\alpha_0d_\MM(x_0,y_0)^2+\frac{1}{\pi^2}\sum_{k=1}^\infty \beta_{k} d_{\cone\left(\MM,\pi\sqrt{\frac{\alpha_k}{\beta_k}}d_{\MM}\right)}\big((s_k,x_k),(t_k,y_k)\big)^2},
$$
for each $\chi=\big(x_0,(s_1,x_1),(s_2,x_2),\ldots,\big),\upsilon=\big(y_0,(t_1,y_1),(t_2,y_2),\ldots,\big)\in \mathcal{P}$. Then
$
\cc_{(\mathcal{P},\rho)}(\MM,\f\circ d_\MM)\lesssim 1.
$
\end{proposition}

The full details of the proof of \autoref{EA:prop:transform-in-producr-cones} appear in~\cite{EMN-hadamard}; we will next sketch its steps.
The following lemma shows that for any metric space $(\MM,d_\MM)$, the truncated metric space $(\MM,\min\{d_\MM,\pi\})$ embeds into the Euclidean cone over $(\MM,d_\MM)$ with bi-Lipschitz distortion $O(1)$.  The corresponding embedding is simply $(x\in \MM)\mapsto (1,x)\in (0,\infty)\times \MM$, for which the stated distortion bound is proved in a straightforward way using elementary calculus; the details appear in~\cite{EMN-hadamard}.
\begin{lemma} \label{EA:prop:cone-truncation} Every metric space $(\MM,d_\MM)$ satisfies $\cc_{\cone(\MM,d_\MM)}\big(\MM,\min\{d_\MM,\pi\}\big)\le \frac{\pi}{2}$.
\end{lemma}

The connection between \autoref{EA:prop:cone-truncation} and \autoref{EA:prop:transform-in-producr-cones} is obtained using the following representation by Brudny\u{\i} and Krugljak~\cite{BK} of any metric transform $\omega:[0,\infty)\to [0,\infty)$ as an affine combination of truncations. By~\cite[Proposition~3.2.6]{BK}, there are $\alpha_0=\alpha_0^\omega\ge 0$ and $\{\alpha_n=\alpha_n^\omega\}_{n=1}^\infty ,\{\beta_n=\beta_n^\omega\}_{n=1}^\infty\subset (0,\infty)$ such that
$$
\forall t\ge 0,\qquad \omega(t)\asymp \alpha_0 t+ \sum_{k=1}^\infty \min\{\alpha_k t , \beta_k\}.
$$
Furthermore, $\lim_{t\to \infty} \omega(t)/t>0$ if and only if $\alpha_0>0$.

It is elementary to verify (as is done in, e.g.,~\cite[Remark~5.4]{MN-quotients}) that if $\f:[0,\infty)\to [0,\infty)$ is a metric transform,
then the mapping $\omega=\omega^\f:[0,\infty)\to [0,\infty)$, given by setting $$\forall t\ge 0,\qquad \omega(t)\eqdef  \f\big(\sqrt{t}\big)^2$$ is also a metric transform. Using the aforementioned general representation of this specific $\omega$, we see that
\begin{equation}
\label{EA:eq:l2-transform-approximation-intro}
\forall t\ge 0,\qquad     \f(t)\asymp \sqrt{\alpha_0t^2+\sum_{k=1}^\infty \min\{\alpha_k t^2,\beta_k\}}.
\end{equation}
\autoref{EA:prop:transform-in-producr-cones} now follows by combining~\eqref{EA:eq:l2-transform-approximation-intro} with  \autoref{EA:prop:cone-truncation}  through a sequence of elementary estimates that are performed in~\cite{EMN-hadamard}.

With \autoref{EA:prop:transform-in-producr-cones}  at hand, we can now outline how \autoref{EA:thm:cat(0)-transform-closure-intro} is proved; the complete details of the subsequent deduction appear in~\cite{EMN-hadamard}.

A metric space $(\NN,d_\NN)$ is said to be a $\CAT(1)$ space if it satisfies the following conditions. Firstly, we require that for every $x,y\in \NN$ with $d_\NN(x,y)<\pi$ there exists a constant speed geodesic in $\NN$ joining $x$ to $y$.
Next, suppose that $x,y,z\in \NN$ satisfies
 $ d_\NN(x,y)+d_\NN(y,z)+d_\NN(z,x)<2\pi,$
and that $\gamma_{x,y},\gamma_{y,z},\gamma_{z,x}:[0,1]\to \NN$ are geodesics
that join $x$ to $y$, $y$ to $z$, and $z$ to $x$, respectively.
Let $\mathbb{S}^2$ be the unit Euclidean sphere in $\R^3$, and let $d_{\mathbb S^2}$
denote the geodesic metric on $\mathbb{S}^2$ (thus the diameter of $\mathbb S^2$
equals $\pi$ under this metric).
As explained in~\cite{BH99},
there exist $a,b,c\in \mathbb{S}^2$ such that $d_{\mathbb{S}^2}(a,b)=d_\NN(x,y)$,
$d_{\mathbb{S}^2}(b,c)=d_\NN(y,z)$, and $d_{\mathbb{S}^2}(c,a)=d_\NN(z,x)$. Let
$\phi_{a,b},\phi_{b,c},\phi_{c,a}:[0,1]\to \mathbb{S}^2$ be $d_{\mathbb S^2}$-geodesics that join $a$
to $b$, $b$ to $c$, and $c$ to $a$, respectively.
The second
requirement in the definition of a $\CAT(1)$ space is that $d_\NN(\phi_{x,y}(s),\phi_{y,z}(t))\le
d_{\mathbb{S}^2}(\phi_{a,b}(s),\phi_{b,c}(t))$ for every
$s,t\in [0,1]$.  See~\cite{BGP92,Per95,BH99,Stu99,BBI01,Sturm-NPC,Gro07} for more on this fundamental notion.

A straightforward consequence
of the relation between distances in $\mathbb R^2$ and
$\mathbb S^2$ shows that every $\CAT(0)$ space is also a $\CAT(1)$ space; see~\cite[Theorem~II.1.12]{BH99}.
An important theorem of Berestovski{\u\i}~\cite{Ber-Cone} (see also~\cite{ABN86}
and~\cite[Theorem~II.3.14]{BH99}) states that a metric space $(\NN,d_\NN)$ is a $\CAT(1)$ space if and only if $\cone(\NN,d_\NN)$ is a $\CAT(0)$ space. By combining these two facts (and the obvious fact that if one multiplies the metric of a $\CAT(0)$ space by a positive constant, then the resulting metric space is also a $\CAT(0)$ space), if $(\MM,d_\MM)$ is a $\CAT(0)$ metric space, then $\cone(\MM,sd_\MM)$ is also a $\CAT(0)$ space for any $s>0$.

Returning to the setting of \autoref{EA:thm:cat(0)-transform-closure-intro}, we are given a $\CAT(0)$ metric space $(\MM,d_\MM)$ and a metric transform $\f:[0,\infty)\to [0,\infty)$. We will next apply \autoref{EA:prop:transform-in-producr-cones}, and proceed using the notation in its statement. As explained above, $\cone(\MM,\pi\sqrt{\alpha_k/\beta_k}d_{\MM})$ is a $\CAT(0)$ space for each $k\in \N$. Consequently, the Pythagorean product $(\mathcal{P}, \rho)$  of \autoref{EA:prop:transform-in-producr-cones} is also a $\CAT(0)$ space, since the $\CAT(0)$ condition~\eqref{eq:def CTA0} is a quadratic inequality that evidently passes to Pythagorean products (by applying it coordinate-wise).
\autoref{EA:prop:transform-in-producr-cones} asserts that the metric space $(\MM,\f\circ d_\MM)$ admits an embedding
into the $\CAT(0)$ space $(\mathcal{P}, \rho)$
whose bi-Lipschitz distortion is $O(1)$, so the conclusion of \autoref{EA:thm:cat(0)-transform-closure-intro} holds.

\subsection*{Added in proof} The concurrent work~\cite{ADTT} made substantial (exciting) progress towards Question~\ref{EA:ques:Kleinberg} by proving that, with high probability, a random regular graph is a universal $O(1)$-approximator with respect to an independently chosen random regular graph. This resolves the question of Kleinberg that \autoref{EA:ques:Kleinberg}  mentions and provides a randomized algorithm that performs as described in~\autoref{EA:thm:main positive}.
Derandomizing that construction remains a major challenge; see~\cite{ADTT} for more details.

\newpage
\thispagestyle{empty}

\backgroundsetup{
  opacity=0.1,
  contents={\textsf{\textup{DRAFT}}}
}

\part{\MakeUppercase{Hadamard subspaces are closed under metric transforms (draft)}}
\markright{\MakeUppercase{Hadamard subspaces are closed under metric transforms (Draft)}}
\label{part:FV}

\begin{abstract}
We prove that every metric transform of a CAT(0) space admits a bi-Lipschitz embedding into some CAT(0) space. 
Using this rigidity result, we prove that there is a metric space $\MM$ that admits bi-Lipschitz embeddings into a space of nonpositive Alexandrov curvature and also into a space of nonnegative Alexandrov curvature, yet $\MM$ does not embed coarsely into a Hilbert space. 
This result refutes a bi-Lipschitz analog of a  theorem of Wilson (1932) and Blumenthal (1935), and as a consequence, it resolves a question of Ding, Lee, and Peres (2013) by establishing the impossibility of a  certain weak metric interpretation of a theorem of Kwapie\'n  (1972).
Furthermore, we use the above embedding result to construct an expander with respect to random regular graphs that satisfies the $p$-Poincar\'e inequalities for all $p \in (0, \infty)$, thus answering a question of Mendel and Naor (2014).  This in turn yields a deterministic approximation algorithm to compute average distances in $n$-point subsets of typical graphs with constant degree using $O(n)$ queries.  We complement it by proving super-linear lower bound on the query complexity for arbitrary such graphs, matching the approximation algorithm of Barhum, Goldreich, and Shraibman (2007) for general finite metric spaces. Finally, we construct  $CAT(0)$ spaces for which the growth rates of the distortion of embeddings of their finite subsets into a Hilbert space or into $L_1$ grows to $\infty$ at an essentially arbitrary rate.
\end{abstract}

\ifdefined \withtitle
\maketitle
\fi

\setcounter{tocdepth}{1} 
\tableofcontents

\section{Introduction} \label{sec:intro}

\subsection{Metric transforms and rigidity}
\label{sec:intro:transform}

The main result of this paper is the following theorem:.
\begin{theorem}\label{thm:cat(0)-transform-closure}
Let $\varphi:[0,\infty)\to [0,\infty)$ be a metric transform, and $(X,d_X)$ be a CAT$(0)$ space. Then $(X,\varphi\circ d_X)$ embeds in some CAT$(0)$ space with distortion at most $\pi \sqrt{3/2}<4$.
\end{theorem}

We recall the concepts used in \autoref{thm:cat(0)-transform-closure}.
The CAT(0) property is a metric notion for (global) nonpositive curvature that was introduced by Aleksandrov; see \autoref{def:CAT(0)} (the nomenclature here is due to Gromov, who studied this notion extensively). 
Distortion is the multiplicative factor that is needed to approximate one metric using the other. Formally:
\begin{definition}[Bi-Lipschitz embedding]
\label{def:biLipschitz-embedding}
An embedding $f:(X,d_X)\to(Y,d_Y)$ of a metric space $X$ in a metric space $Y$
is said to have distortion at most $D$
if there exists $s>0$ such that
\[ d_X(x,y)\le s\cdot d_Y(f(x),f(y)) \le D\cdot d_X(x,y),\quad
\forall x,y\in X.
\]
The distortion of embedding $X$ in $Y$ is denoted $c_Y(X)$ and defined as
\begin{equation} \label{eq:def:c_Y(X)} c_Y(X)=\inf \bigl \{D\ge 1:\ \exists f:X\to Y \text{ with distortion at most }D\bigr\}.\end{equation}
$X$ admits a bi-Lipschitz embedding in $Y$ if $c_Y(X)<\infty$. We use the shorthand $c_p(X)=c_{L_p}(X)$.
\end{definition}


\begin{definition}[Metric transform] \label{def:metric-transform}
A real function $\varphi:[0,\infty)\to [0,\infty)$ is called a \emph{metric transform} if $\varphi$ is nondecreasing, concave, and $\varphi(0)=0$.
\end{definition}

It follows immediately that if $(X,d_X)$ is a metric space and $\varphi$ is a metric transform, then $(X,\varphi\circ d_X)$ is also a metric space.
Important special cases of metric transforms are \emph{truncations} $\varphi(t)=\min\{\tau,t\}$ for some $\tau>0$ and \emph{snowflakes} $\varphi(t)={t}^\theta$, for some $\theta\in(0,1]$.

The demonstrable property of \autoref{thm:cat(0)-transform-closure} is called rigidity (under metric transforms) of the class of CAT(0) spaces.
Previously, it was proved that
Hilbert spaces are rigid in this sense~\cite[Remark~5.4]{MN-quotients}, and so is $L_1$ \cite[Remark~5.5]{MN-expanders2} (see also \cite[Remark~123]{Naor-extension}).
On the other hand, for any $p\in(2,\infty)$, $L_p$  is not rigid~\cite[Remark~5.12]{MN-quotients}.
The rigidity of
$L_p$, when $p\in(1,2)$ is an open question, although some progress is reported in~\cite[Lemma~5.9]{MN-quotients}.
The rigidity of $L_1$ and $L_2$ under metric transforms has been found to be useful in metric geometry.
Examples include embeddability of quotients of subspaces of the Hamming cube in Hilbert spaces~\cite{MN-quotients}, construction of expander graphs with strong metric property~\cite{MN-expanders2}, extension of Lipschitz functions
of planar graphs~\cite{DLP}, and more.
In the rest of this section, we present applications of \autoref{thm:cat(0)-transform-closure} that we derive herein.

\subsection{Counterexample to nonlinear Kwapie\'n formulations}
\label{sec:intro-kwapien}

We establish the existence of a metric space that admits a
bi-Lipschitz embedding into both a CAT(0) space (\autoref{def:CAT(0)}) and an Aleksandrov space of nonnegative curvature (\autoref{def:NNC}),
but fails to embed in a Hilbert space, even in the following weak sense.

\begin{definition}[Coarse embedding on the average]
\label{def:coarse-average}
A metric space $X$ is said to be coarsely embedded on the average in another metric space $Y$ if there exist
monotone
$\omega,\rho,\Omega:[0,\infty)\to [0,\infty)$ satisfying  $\omega(0)=\rho(0)=0$ and $\omega(t),\rho(t),\Omega(t)\to\infty$ as $t\to \infty$,
with the following properties:
For any Borel probability measure $\mu$ on $X$ there exists
an embedding $f_\mu:\supp(\mu)\to Y$ satisfying, on the one hand,
\begin{equation} \label{eq:coarse-ub}
d_Y(f_\mu(x),f_\mu(y)) \le \Omega(d_X(x,y)), \quad \forall x,y\in\supp(\mu),
\end{equation}
and on the other hand,
\begin{equation}\label{eq:coarse-avg-lb}
\iint_{X\times X} d_Y(f_\mu(x),f_\mu(y)) 
\dd\mu(x)\dd\mu(y)
\ge \rho\biggl( \iint_{X\times X} \omega(d_X(x,y))\dd\mu(x)\dd\mu(y)\biggr).
\end{equation}
\end{definition}

\begin{remark} \label{rem:embedding-notions}
Coarse embedding on the average is a weakening of several known notions of embedding.
The relevant ones here are bi-Lipschitz embedding (\autoref{def:biLipschitz-embedding}),
weak bi-Lipschitz embedding~\cite{NPSS,DLP},
{bi-H\"older embedding}~\cite{Assouad},
{quadratic average embedding of the snowflake}~\cite{Naor-comparison},
and {coarse embedding}~\cite{Gromov-coarse}.
See \autoref{fig:embedding-notions} for the relations between them when embedding into Hilbert space.%
\footnote{%
Coarse embedding~\cite{Gromov-coarse} is a coarse embedding on the average 
with 
$f_\mu=f_\nu$ on $\supp(\mu)\cap \supp(\nu)$ for any pair of Borel measures $\mu,\nu$. 
Bi-H\"older embedding is bi-Lipschitz embedding of a snowflake of the metric.
Quadratic average distortion $D$ embedding of the $\theta$-snowflake is a coarse embedding on the average with
$\Omega(t)=D\cdot t$, $\rho(t)=t^{1/2}$ and $\omega(t)=t^{2\theta}$.
$(\tau,L)$-weak bi-Lipschitz embedding is a Lipschitz mapping that maps distances at least $\tau$ to distances at least $\tau/L$. Weak bi-Lipschitz embedding is $(\tau,L)$-weak bi-Lipschitz embedding for some $L$ and all $\tau>0$. 
It is called \emph{threshold embedding} in~\cite{DLP}.}
\end{remark}

\begin{figure}[ht]
\begin{center}
\begin{tikzcd}[cramped, sep=small, cells={nodes={draw=black}}]
\text{Isometric} \ar[r] & \text{Bi-Lipschitz} \ar[r]
 & \text{Weak bi-Lipschitz}
 \ar[d] \ar[r]
 & \text{Quadratic average of the snowflake} \ar[dd]
 \\
&  & \text{Bi-H\"older} \ar[d]
 \\
  &  & \text{Coarse}
 \ar[r]
& \text{\fbox{Coarse on the average}}
\end{tikzcd}
\end{center}
\caption{Notions of embedding in \autoref{rem:embedding-notions} and their relative strength for embedding into Hilbert space.}
\label{fig:embedding-notions}
\end{figure}
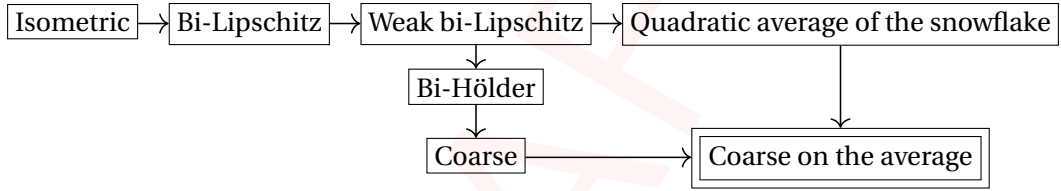

\begin{theorem} \label{thm:no-Kwapien}
There exists a metric space $\SE$ that admits a bi-Lipschitz embedding in some CAT$(0)$ space as well as in some Alexandrov spaces of nonnegative curvature.
Yet, $\SE$  is not coarsely embedded on average in Hilbert space.
\end{theorem}

\autoref{thm:no-Kwapien} is proved in \autoref{sec:no-kwapien}.
The space $\sqrt{\mathcal E}$ is the snowflake of a union of high-girth expander graphs. 
Those graphs are discussed in~\autoref{sec:intro-extrapolation-hadamard} and \autoref{sec:no-kwapien}.
\autoref{thm:no-Kwapien} stands in contrast to the following classical characterization of convex subsets of Hilbert space by Blumenthal:

\begin{theorem}[{\cite{Blu35}, see also \cite[pp.~122-128]{Blu53}}]
\label{thm:blumenthal}
A geodesic space is a convex subset of a Hilbert space if and only if
it satisfies both the CAT(0) condition and Aleksandrov's nonnegative curvature condition.
\end{theorem}

\begin{remark}
\autoref{thm:blumenthal} is part of the classical characterizations of Hilbert space via $4$-point properties.
It was preceded by Wilson~\cite{Wil32}.
Analogous results for normed spaces were obtained by
Fr\'echet~\cite{Fre35}, Jordan--von-Neumann~\cite{JvN35} and Day~\cite{Day47}.
From a different perspective,
Blumenthal's characterization was extended to a
regularity theory of Aleksandrov spaces with two-sided curvature bounds.
Aleksandrov first explored this in~\cite{Ale57} (see also~\cite[Theorem~11.4]{ABN86}) under additional topological assumptions,
and ultimately by Nikolaev~\cite{Nik83}
(see also~\cite[Theorem~10.10.13]{BBI01}).
We refer to Lurie's notes~\cite[\S~7]{Lurie-Hadamard}
for a simple modern proof of \autoref{thm:blumenthal}.
\end{remark}

In light of the above, we ask:

\begin{question} \label{ques:isometric-embedding}
Does every separable metric space that embeds \emph{isometrically} in some Aleksandrov space of nonnegative curvature and in some CAT(0) space must coarsely embed on average in Hilbert space?
\end{question}

\begin{question}
Does every \emph{geodesic} metric space that embeds \emph{bi-Lipschitzly} in some Alexandrov space of nonnegative curvature and in some CAT(0) space must coarsely embed on average in Hilbert space?
\end{question}

\medskip

\autoref{thm:no-Kwapien} also refutes some non-linear formulations of
 Kwapie\'n theorem, which is stated as follows.

\begin{theorem}[\cite{Kwapien}]
A Banach space $X$ has type-$2$ and cotype-$2$  if and only if it is linearly isomorphic to Hilbert space.
\end{theorem}

Type-$2$ and cotype-$2$ are linear isomorphic properties of Banach spaces. Together, they constitute the following high-dimensional isomorphic generalization of the parallelogram identity:
\begin{equation}
\label{eq:isomorphic-parallelgoram}    
\forall n\in\mathbb N,\ \forall x_1,\ldots,x_n\in X,\qquad 
2^{-n}\!\!\!\!\!\!\!\!\!\sum_{\e_1,\ldots\e_n\in\{\pm 1\}}
\Bigl\|\sum_{i=1}^n \e_i x_i\Bigr\|_X^2 \asymp
\sum_{i=1}^n \|x_i\|_X^2.
\end{equation}
More specifically, type~2 property is when the left-hand side of~\eqref{eq:isomorphic-parallelgoram}
is at most the right-hand side up to some constant factor, while cotype~2 is the reverse inequality.
See, e.g.,~\cite{Maurey} for their importance.
The Ribe program (see~\cite{Ball-Ribe,Naor-Ribe}) is an attempt to formulate nonlinear analogs to linear concepts and theorems.
As part of this program, concepts of nonlinear type such as \emph{BMW type}~\cite{BMW}, \emph{Enflo type}~\cite{Enflo-type1,Enflo-type2,Pisier-type,IvHV}, scaled Enflo type~\cite{MN-scaled-enflo,GN-sacled-enflo}
and \emph{Markov type}~\cite{Ball,NPSS} were developed.
Similarly, \emph{metric Markov cotype}~\cite{Ball,MN-barycentric},  \emph{metric cotype}~\cite{MN-cotype,GMN-cotype,EMN}
and \emph{sharp metric cotype}~\cite{MN-cotype,EMN}
were developed as nonlinear counterparts for cotype.
See the cited literature for definitions.

The possibility of a nonlinear version of Kwapie\'n theorem
was first raised in~\cite{MN-cotype} and explored in~\cite{DLP}.
The straightforward approach --- employing metric cotype~$2$, Enflo type~$2$, and bi-Lipschitz embedding into Hilbert space --- is false, as noted in~\cite[\S~1]{DLP}.
Ding et al.~\cite{DLP} proposed easing the notion of embedding into Hilbert space by using weak bi-Lipschitz embedding
(see \autoref{rem:embedding-notions}).
They asked whether a metric space that possesses both metric cotype~$2$ and Markov type~$2$ necessarily admits a weak bi-Lipschitz embedding into Hilbert space.
Interpreted literally, Ding et al.'s question is easy to disprove as follows.
\begin{inparaenum}[(a)]
\item Weak bi-Lipschitz embedding into Hilbert space implies a coarse embedding into Hilbert space; see
\autoref{rem:embedding-notions}.
\item Many metric spaces $\mathcal N$ are known to be coarsely non-embeddable in Hilbert space; see~\cite[\S~2]{EMN}.
\item It follows from the definitions of metric cotype~$2$ and Markov type~$2$ that the square root of any metric has both.
\item Coarse embedding is invariant for taking the square roots of the distances.
\item Therefore,
the square roots of these spaces $\mathcal N$ have metric cotype~$2$ and Markov type~$2$, but they do not embed coarsely in Hilbert space.
\end{inparaenum}

However, the spirit of Ding et.~al.'s question
is more open-ended: Are there type-like and cotype-like conditions
that, together, force some nontrivial embedding in Hilbert space?
In particular, the counterexample above does not cover \emph{sharp} metric cotype~$2$, since the square root of an
arbitrary metric space does not necessarily have \emph{sharp} metric cotype~$2$~\cite{ANN}.
However, CAT(0) spaces have sharp metric cotype~$2$~\cite{EMN},
and Alexandrov spaces of nonnegative curvature have Markov type~$2$~\cite{Ohta} and Markov $2$-convexity~\cite{AN-markov-conv}.
We therefore have the following corollary of \autoref{thm:no-Kwapien}.

\begin{corollary}
The metric space $\sqrt{\smash[a]{\mathcal E}}$ from \autoref{thm:no-Kwapien} has Markov type\/~$2$, \emph{sharp} metric cotype\/~$2$, and Markov\/ $2$-convexity.
However, it is not coarsely embedded on average in Hilbert space.
\end{corollary}

$\SE$ is not a geodesic space.
In~\cite{EMN-diamond-convexity}
the authors provide a geodesic counterexample for coarse Kwapia\'n-like formulations:
The Wasserstein space $W_2(\mathbb R^3)$ is a \emph{geodesic space} with Markov type~$2$, metric cotype~$2$ (but not \emph{sharp} metric cotype) and Markov $2$-convexity.
Nevertheless, it is not coarsely embeddable on average in Hilbert space, since it admits bi-Lipschitz embedding of $\SE$.
This naturally leads to the following question raised in~\cite{EMN-diamond-convexity}:
\begin{question}
Does there exist a \emph{geodesic} metric space with both
\emph{sharp} metric cotype~$2$ and Enflo type~$2$ that does not embed coarsely on average in Hilbert space?
\end{question}

\subsection{Matou\v{s}ek extrapolation for Hadamard spaces}
\label{sec:intro-extrapolation-hadamard}

Expander graphs are families of graphs with some expansion properties that are useful in branches of mathematics such as combinatorics, geometric group theory, probability, metric geometry, computer science, and more.
Good introductions can be found in~\cite{HLW,Kowalski}.
From a metric geometry perspective,
they are defined by the Poincar\'e inequalities for expanders:
\begin{definition}[Poincar\'e inequality for expanders w.r.t. $X$]
\label{def:X-PI}
Fix $p\in(0,\infty)$, and a metric space $(X,d_X)$.
An undirected regular graph $G=(V,E)$ is said to satisfy the $p$-Poincar\'e inequality (for expanders) with respect to $(X,d_X)$ and constant $\gamma>0$ if for every $f: V\to X$,
\begin{equation} \label{eq:p-poincare-metric}
\frac{1}{|V|^2}\sum_{u,v\in V} d_X(f(u),f(v))^p \le \frac{\gamma}{|E|} \sum_{(u,v)\in E} d_X(f(u),f(v))^p.
\end{equation}
For a family $\mathcal G$ of undirected regular graphs and a family $\mathcal X$ of metric spaces, we define by $\gamma_p(\mathcal G, \mathcal X)\in[0,\infty]$ the infimal $\gamma$
satisfying~\eqref{eq:p-poincare-metric} for every $G\in\mathcal G$, and $X\in\mathcal X$.
An infinite family of constant degree graphs $\mathcal E$ satisfying $\gamma_p(\mathcal E,\mathcal X)<\infty$
is called
\emph{$\mathcal X$-expander with power $p$}.
When $p=2$, we sometimes drop the mention of the power.
\end{definition}

Classical expanders are usually defined by a positive constant lower bound on edge expansion or a positive constant lower bound on the spectral gap $1-\lambda_2(G)$,
where $\lambda_2(G)$ is the second eigenvalue of the (symmetric) Markov operator associated with $G$.
It is folklore and straightforward that $1/\gamma_2(G,\mathbb R)$ is the spectral gap of $G$, and $1/\gamma_p(G,\{0,1\})$ is the normalized edge expansion%
\footnote{More precisely, for $G=(V,E)$,
$1/\gamma_p(G,\{0,1\})=
\min_{\emptyset\ne S\subset V} \frac{n\cdot |E(S,V\setminus S)|}{d \cdot|S|\cdot |V\setminus S|}$.
This quantity is sometimes called \emph{the conductance of\/ $G$}. It is equivalent to the normalized edge expansion up to a factor of $2$.}
of $G$.
A variant of the Cheeger inequality (see~\cite[§~1.4]{HLW} or~\eqref{eq:mat-extrapolation} below)
is equivalent to
\begin{equation*} 
    \gamma_2(G,\mathbb R)\geq \gamma_p(G,\{0,1\}) \geq \sqrt{\gamma_2(G,\mathbb R)/8}.
\end{equation*}
Thus, using the notation of \autoref{def:X-PI}, classical expanders are both $\mathbb R$-expanders and $\{0,1\}$-expanders.
It is natural to ask how the geometry $\mathcal X$ and the power $p$ affect the notion of expansion.
This is an active area of research with no adequate survey as far as we know, but pointers to the literature can be found
in~\cite[\S 2.3 -- \S 2.6]{EMN}.
A straightforward observation is that
Hilbert expanders are equivalent to classical expanders.
Furthermore,
$\gamma_2(G,\mathbb R)=\gamma_2(G,\mathcal H)$ for any regular graph $G$ and Hilbert space $\mathcal H$.
Moreover, if $X$ is a metric space with more than one point, then $\gamma_p(G,X)\geq \gamma_p(G,\{0,1\})$, and thus any $X$-expander with power $p>0$ is necessarily also a classical expander.
The reverse is not true: There exists a metric space $X$, and two classical expanders $\mathcal E_1$ and $\mathcal E_2$ such that $\mathcal E_1$ is an ${X}$-expander, while $\mathcal E_2$ is not~\cite{MN-expanders2}.
In the last two decades, this subject has been studied extensively~\cites{Ozawa,Pisier-interpolation,Lafforgue-type-expanders,MN-superexpanders,MN-expanders2,Naor-comparison}.

Matou\v{s}ek~\cite{Mat-extrapolation}, who was the first to consider the Poincar\'e inequalities for expanders,
proved that for any $p\in[1,\infty)$ and a family of regular graphs $\mathcal E$,
$\gamma_p(\mathcal E, \{0,1\})<\infty$ is equivalent to
$\gamma_p(\mathcal E, \mathbb R)<\infty$.
Quantitatively, his argument gives, for any $p,q\in [1,\infty)$ and a regular graph $G$,
\begin{equation}\label{eq:mat-extrapolation}
\gamma_p(G, \mathbb R)\lesssim_{p,q} \gamma_q(G,\mathbb R)^{\max\{1,p/q\}}.
\end{equation}
(Inequality~\eqref{eq:mat-extrapolation} for general $q\ge 1$ first appeared explicitly in~\cite[Lemma~5.5]{BLMN}.)
de Laat and de la Salle~\cite{Laat-Salle} extended~\eqref{eq:mat-extrapolation} to arbitrary Banach spaces:
For any Banach space $X$, any $p,q\in[1,\infty)$ and a regular graph $G$,
\begin{equation}\label{eq:Laat-Salle--extrapolation}
\gamma_p(G, X)\lesssim_{p,q} \gamma_q(G,X)^{\max\{1,p/q\}}.
\end{equation}

(For the best known implicit constant factors known in~\eqref{eq:Laat-Salle--extrapolation}, see \cite[Eq.~(131)]{Naor-avg-john}.)
In~\autoref{sec:mat-extrapolation} we adapt
de Laat and de la Salle's argument for Hadamard spaces.

\begin{theorem} \label{thm:extrapolation-cat0}
Let ``Hadamard'' be the class of all Hadamard spaces, and
fix $p,q\in(0,\infty)$. Then for every regular graph $G$,
\begin{equation} \label{eq:extrapolation-cat0}
\gamma_p(G,\mathrm{Hadamard}) \le \begin{cases}
4^q \gamma_q(G,\mathrm{Hadamard})  &0< p \le q ,\\
C_{p,q}\cdot  \gamma_q(G,\mathrm{Hadamard})^{p/q}
& 1\le q\le p.
\end{cases}
\end{equation}
\end{theorem}

An expander with respect to all Hadamard spaces is called \emph{Hadamard-expander}.
It is currently unknown whether Hadamard expanders exist to begin with.
But if they do, then the following corollary of \autoref{thm:extrapolation-cat0} is of interest.

\begin{corollary}\label{cor:p-hadamard expander}
If there exists a Hadamard expander, then it satisfies the $p$-Poincar\'e inequality with respect to any Hadamard space and any $p\in(0,\infty)$.
\end{corollary}

Gromov~\cite{Gromov-rand-groups} and Kondo~\cite{Kondo} precluded the possibility that every (classical) expander is a Hadamard expander by constructing a specific expander $\mathcal{E}$ and a
specific Hadamard space $\mathcal{X}$ such that $\mathcal{E}$ admits a bi-Lipschitz embedding in $\mathcal{X}$.
Nonetheless, their observation \emph{does not rule out} the existence of Hadamard expanders, which remains a compelling open problem.
Indeed, there exists a family of expander graphs that are expanders with respect to "most" Gromov-Kondo Hadamard spaces, in a concrete technical sense~\cite{MN-expanders2}.
Consequently, we reiterate a known open question (see~\cite{EMN}): 

\begin{question}
Does a Hadamard expander exist?
\end{question}

It would also be helpful to have a more granular, de Latt--de la Salle type extrapolation, especially if Hadamard expanders do not exist.
\begin{question}
Is it true that
for any CAT(0) space $X$, any $p,q\in[1,\infty)$ and a regular graph $G$,
\begin{math}
\gamma_p(G, X)\leq f_{p,q}( \gamma_q(G,X))
\end{math}
for some monotonous non-decreasing functions $f_{p,q}:[1,\infty)\to [1,\infty)$?
\end{question}

\subsection{Expanders with respect to random regular graphs}
\label{sec:intro-RRG}

The utilization of \autoref{thm:extrapolation-cat0} for Hadamard expanders (\autoref{cor:p-hadamard expander}) is currently hypothetical. 
However, analogous reasoning derives Poincar\'e inequalities for expanders with respect to the random regular graphs, as outlined below.
Define $\mathcal G_{N,d}$ as the uniform probability space over the collection of unlabeled simple $N$-vertex $d$-regular graphs.
The existence of expanders with respect to random regular graphs was established in~\cite{MN-expanders2}:
\begin{theorem}[{\cite[Theorem~1.2]{MN-expanders2}}]
\label{thm:MN-rrg-expander-2}
There exists an infinite family $\mathcal{E}$ of\/ $3$-regular graphs, and $\Gamma>0$ such that for any $d\ge 3$,
\(
\lim_{N\to \infty}\Pr_{\RRG\sim \mathcal G_{N,d}}
\left[\gamma_2(\mathcal{E},\RRG)\le \Gamma\right]=1.
\)
\end{theorem}

In the following theorem, proved in \autoref{sec:extrapolation-rrg}  we ``extrapolate" the power to the range $(0,\infty)$:

\begin{theorem}
\label{thm:MN-rrg-expander-p}
There exists an infinite family $\mathcal{E}$ of\/ $3$-regular graphs with the following property: For every $q\in\{2,3,4,\ldots\}$ there exists $\Gamma_q\in(0,\infty)$ such that
for any integer $d\ge 3$, 
\[
\lim_{N\to \infty}\Pr_{\RRG\sim \mathcal G_{N,d}}
\left[\sup_\varphi \gamma_q\bigl (\mathcal{E},\bigr (\RRG,\varphi\circ d_{\RRG}\bigr )\bigr)\le \Gamma_q\right]=1,
\]
where the supremum is over all metric transforms $\varphi:[0,\infty) \to [0,\infty)$.
\end{theorem}
By substituting $q=\max\{2,\lceil p\rceil\}$ and $\varphi(t)=t^{p/q}$ in \autoref{thm:MN-rrg-expander-p} for a given $p\in(0,\infty)$, 
we conclude:
\begin{corollary} \label{cor:rrg-expander-p}
There exists an infinite family $\mathcal{E}$ of\/ $3$-regular graphs with the following property: For every $p\in(0,\infty)$ there exists $\Gamma_p\in(0,\infty)$ such that
for any integer $d\ge 3$, 
\[
\lim_{N\to \infty}\Pr_{\RRG\sim \mathcal G_{N,d}}
\left [\gamma_p (\mathcal{E},\RRG)\le \Gamma_p\right]=1,
\]
\end{corollary}

The extension of \autoref{thm:MN-rrg-expander-2} to
\autoref{thm:MN-rrg-expander-p} is achieved using three different ingredients. The first is a generalization of the construction that proves that \autoref{thm:MN-rrg-expander-2} works with power greater than $2$. This generalization is relatively straightforward, as it involves generalizing all inequalities with power $2$ to power $p\in[1,\infty)$,
except for the Pisier martingale cotype $p$,
which is false in the range $p\in[1,2)$.
Generalizing it to $p\in(2,\infty)$ involves a generalization of a barycentric property of CAT(0) to power $p>2$ (\autoref{prop:p-barycentric-CAT0}) by adapting an argument from~\cite{MN-superexpanders} written in the context of uniformly convex Banach spaces and is similar to~\cite{Gietl25}. 
All are based on an argument from~\cite{Ball}.
The second crucial ingredient is generalizing the construction for $p\in(0,2)$. As the martingale cotype $p$ is false in this range, this is achieved by extending the construction to work not only with the metric of random regular graphs but also with its truncations. 
This task is helped by the fact that the
proof of \autoref{thm:MN-rrg-expander-2} uses
embeddings of random regular graphs in Euclidean cones over
spaces that are ``almost" CAT(1).
We will exploit this special structure to apply arguments similar to those used in the proof of \autoref{thm:extrapolation-cat0}, when $p<q$.
The third ingredient is weaving the constructions of different expanders for different powers into one expander, which is based on a similar task performed in~\cite{MN-superexpanders} in a somewhat different context.

Independently of our paper, Atschuler et.~al.~\cite{ADTT} proved
a theorem very similar to \autoref{cor:rrg-expander-p} (for $p\in[1,\infty)$), using a very different method: they solved a question of J.~Kleinberg (see~\cite{MN-expanders2}) who asked whether for two stochastically independent
random regular graphs $\mathbb X$ and $\mathbb Y$ 
(of possibly different size and degree), $\mathbb Y$ is an $\mathbb X$-expander asymptotically almost surely.

\subsection{Estimation of the average distance in random regular graphs}
\label{sec:intro:avg-distance-approximator}
As part of a program to develop fast algorithms for geometric computations involving distances, 
Indyk~\cite{Indyk-sublinear} systematically investigated  fast algorithms for approximating the average distance%
\footnote{
We use the common notations $[n]=\n$ and $\tbinom{[n]}{2}=\{\{a,b\}\subset [n]:\ a\neq b\}$ for each $n\in \N$.}  
\begin{equation}\label{average dist}
\frac{1}{n^2}\sum_{(i,j)\in [n]^2} d_{\MM}(x_i,x_j)
\end{equation} 
among $n$ points $x_1,\ldots,x_n\in \MM$ in a given metric space $(\MM,d_{\MM})$.  
As there are $n^2$ summands in~\eqref{average dist}, such an algorithm is considered sublinear if it performs $o(n^2)$ distance queries. 
Obtaining a finite approximation factor entails making at least $n-1$ distance queries since otherwise the graph whose edges are the pairs of points that were queried would be disconnected,  so by varying the distances between its connected components one can obtain two markedly different metric spaces that the algorithm cannot distinguish.

Here, we will study deterministic algorithms for the above question. 
It is beneficial to first ground the discussion (and set the notation/terminology) by recalling the following formal definition of the computational model that will be treated herein. 
Its generality entails that the lower bound that we will obtain is quite a strong statement, while the algorithm that we will design  belongs to the least expressive aspect of this framework, namely, it will be non-adaptive, and, in fact, it will be a ``universal approximator'' per terminology from~\cite{BGS} that we will soon recall.

\begin{definition} \label{def:adapt-approximator} 
Let $\mathcal{F}$ be a family of metric spaces. For $n,m\in \N$ and $\alpha\ge 1$, a deterministic $\alpha$-approximation algorithm for average distance in $\mathcal{F}$  which makes $m$ distance queries on size $n$ inputs consists of functions
\begin{equation}\label{eq:functions for the game}
\Big\{\mathsf{Pair}_i:\Big(\tbinom{[n]}{2}\times [0,\infty)\Big)^{i-1}  \to
\tbinom{[n]}{2}\Big\}_{i=1}^m\qquad\mathrm{and}\qquad \estimate: \Big(\tbinom{[n]}{2}\times [0,\infty)\Big)^m\to [0,\infty),
\end{equation}
where $\pair_1$ is understood to be a fixed element $\pair_1=\{a_1,b_1\}\in \tbinom{[n]}{2}$.

Given a metric space $(\MM,d_{\MM})\in \cF$ and $x_1,\ldots,x_n\in \MM$, define  $\{a_1,b_1\},\{a_2,b_2\},\ldots,\{a_m,b_m\}\in \tbinom{[n]}{2}$ inductively (recalling that $\{a_1,b_1\}$ has already been fixed above) by setting for every $i\in [m-1]$:
\begin{equation}\label{eq:approximation game}
\{a_{i},b_{i}\}\eqdef \pair_{i} \Big(\big(\{a_1,b_1\},d_{\MM}(x_{a_1},x_{b_1})\big), \big(\{a_2,b_2\},d_{\MM}(x_{a_2},x_{b_2})\big),\ldots, \big(\{a_{i-1},b_{i-1}\},d_{\MM}(x_{a_{i-1}},x_{b_{i-1}})\big)\Big).
\end{equation}
We then require that%
\footnote{If the denominator in~\eqref{eq:alpha approx def} vanishes (i.e., when $x_1=\ldots=x_n$),  then we require that also the numerator in~\eqref{eq:alpha approx def} vanishes.}
\begin{equation}\label{eq:alpha approx def}
1\le \frac{\estimate \Big(\big(\{a_1,b_1\},d_{\MM}(x_{a_1},x_{b_1})\big),\ldots, \big(\{a_{m},b_{m}\},d_{\MM}(x_{a_{m}},x_{b_{m}})\big)\Big)}{\frac{1}{n^2}\sum_{(i,j)\in [n]^2} d_{\MM}(x_i,x_j)} \le \alpha.
\end{equation}

A \emph{non-adaptive}  deterministic $\alpha$-approximation algorithm for average distance in $\mathcal{F}$  which makes $m$ distance queries on size $n$ inputs is the restrictive case of the above setup in which $\pair_1,\ldots,\pair_m$  are constant functions, so their images are, respectively, $m$ pairs of indices $\{a_1,b_1\},\ldots,\{a_m,b_m\}\in \tbinom{[n]}{2}$. The output is then a function of the $d_{\MM}$-distances   between the pairs $\{x_{a_1},x_{b_1}\},\ldots, \{x_{a_m},x_{b_m}\}\subset \{x_1,\ldots,x_n\}$ of points whose indices  were decided upfront and not adapted to the specific input $x_1,\ldots,x_n\in \MM$.

When we discuss $\alpha$-approximation algorithms for average distance which make $m$ distance queries on inputs of size $n$ without specifying the underlying family $\cF$, we will tacitly mean that $\cF$ consists of all metric spaces; as any finite metric space is isometric to a subset of $\ell_\infty$, we may take in this case $\cF=\{\ell_\infty\}$. 
\end{definition}
The standard interpretation of~\eqref{eq:approximation game} is that the algorithm performs the following ``approximation game'' in $m$ rounds. 
Starting with querying the $d_{\MM}$-distance between $x_{a_1}$ and $x_{b_1}$, based only on the answer that it receives, the algorithm computes a new pair of points $x_{a_2},x_{b_2}$ and queries the $d_{\MM}$-distance between them. 
In round $i+1$,   the algorithm singles out a new pair of points $x_{a_{i+1}},x_{b_{i+1}}$ and queries their $d_{\MM}$-distance,  where that pair is a function of only  the (ordered) sequence of pairs that the algorithm queried and their $d_{\MM}$-distances (the responses to the queries) in the preceding $i$ rounds. 
After $m$ rounds, the algorithm outputs an estimate for the average distance, which is required to be   a function of only the (ordered) sequence  of the pairs that it queried   and the  responses to those queries in the entirety of this procedure. 

Even though the new algorithmic results that are obtained herein  concern only deterministic algorithms, we wish to mention randomized algorithms when referring to previous results in the literature. 
Thus, in the context of \autoref{def:adapt-approximator}, a randomized $\alpha$-approximation algorithm for average distance in $\mathcal{F}$  which makes $m$ distance queries on size $n$ inputs is the same as in \autoref{def:adapt-approximator} except that the functions in~\eqref{eq:functions for the game} have an additional variable $\omega$ from some probability space $(\Omega,\mathbb{P})$, and one requires that the approximation guarantee~\eqref{eq:alpha approx def} holds with probability at least, say,  $2/3$. In the non-adaptive special case, $\pair_1,\ldots,\pair_m$ depend only on $\omega$. 
Also in this context, $\cF$ is understood to be all metric spaces if it is not specified. 


Indyk considered the simple randomized non-adaptive algorithm that is obtained by averaging over a uniformly random sampling of $O(n/\e^{7/2})$ pairs for some $0<\e<1$. 
He proved that this straightforward procedure yields an $1+\e$ factor approximation.
Barhum, Goldreich, and Shraibman~\cite{BGS} improved Indyk's analysis to
$O(n/\e^2)$ queries. Goldreich and Ron~\cite{GR06,GR-avg} studied a restricted version of the above problem in which one wishes to approximate
the average distance among \emph{all} the vertices of a given \emph{unweighted connected} graph on $n$ vertices, showing that averaging over a uniformly random sampling of $O(\sqrt{n}/\e^{2})$ pairs of vertices yields a $(1+\e)$-approximation algorithm. 

Deterministic algorithms for average distance estimation were broached in~\cite{BGS}, where an especially simple subclass of non-adaptive algorithms, called
\emph{universal approximators}, was studied. Per Definition~6 in~\cite{BGS}, given $\alpha\ge 1$, an $\alpha$-universal approximator in a family $\cF$ of metric spaces of size $m$ for inputs of size $n$ consists of $m$ unordered pairs of indices $\{a_1,b_1\},\ldots,\{a_m,b_m\}\in \tbinom{[n]}{2}$ and a scaling factor $\sigma>0$ such that for every metric space $(\MM,d_{\MM})\in \cF$ and every $x_1,\ldots,x_n\in \MM$  the output of the algorithm is 
\begin{equation}
\label{eq:UA}
\estimate_{\mathrm{UA}}\big(d_{\MM}(x_{a_1},x_{b_1}),\ldots, d_{\MM}(x_{a_m},x_{b_m})\big)\eqdef \sigma\sum_{\ell=1}^m d_{\MM}(x_{a_\ell},x_{b_\ell}),
\end{equation}
 and~\eqref{eq:alpha approx def} holds, that is, $n^{-2}\sum_{(i,j)\in [n]^2}d_{\MM}(x_i,x_j)\le \sigma\sum_{\ell=1}^md_{\MM}(x_{a_\ell},x_{b_\ell})\le \alpha n^{-2}\sum_{(i,j)\in [n]^2}d_{\MM}(x_i,x_j)$ for every metric space $(\MM,d_{\MM})\in\cF$ and every $x_1,\ldots,x_n\in \MM$. As before, if $\cF$ is not mentioned, then it will be assumed tacitly to be all the possible metric spaces.   The following theorem was proved in~\cite{BGS}:

\begin{theorem}[{\cite{BGS}}]
\label{thm:BGS-approximator}
For any $k\in\{2,3,4,\ldots\}$, there exists a  
$(2k)$-universal approximator algorithm that uses at most
$kn^{(k+1)/k}/2$ queries.    
\end{theorem}

In particular, \autoref{thm:BGS-approximator} shows that for every fixed integer $k\ge 2$ there is a (nonadaptive) deterministic $(2k)$-approximation algorithm for the average distance that makes $O(n^{1+1/k})$ distance queries on inputs of size $n$. 
\autoref{thm:impossibility-DA} below, which is one of the main results of the present work, provides a matching impossibility statement.%
\footnote{Besides proving~\autoref{thm:BGS-approximator},  
a (non-matching) lower bound on the size of \emph{universal approximators} was proved in~\cite{BGS}. 
Ruling out any algorithm whatsoever (per \autoref{def:adapt-approximator}) rather than only universal approximators is conceptually an entirely different matter. 
Indeed, the concrete adversarial metric/gadget that is used in~\cite{BGS} fails to fool even some non-adaptive algorithms.}
 
\begin{theorem} \label{thm:impossibility-DA}
Assume $A$ is a deterministic adaptive $\alpha$-approximation algorithm of the average distance for all finite metric spaces. 
If\/ $A$ uses $o\bigl(n^{\frac{k+1}{k}}\bigr)$ queries on $n$-point sets for some fixed $k\in\mathbb N$, 
then $\alpha\geq 2(k+1)$.
\end{theorem}

By combining \autoref{thm:BGS-approximator}  and \autoref{thm:impossibility-DA}, we get the following description of the approximation landscape for the average distance. Any algorithm for the average distance whose approximation guarantee is less than $4$ must query $\Omega(n^2)$ distances; any such algorithm whose approximation guarantee is less than $6$ must query $\Omega(n^{3/2})$ distances; any such algorithm whose approximation guarantee is less than $8$ must query $\Omega(n^{4/3})$ distances, and so forth. Furthermore, there are algorithms attaining those parameters. 
This is depicted in the chart on \autoref{fig:UB-Lb-chart}.

\begin{figure}[htb]
 \begin{center}
\includegraphics[width=0.5\textwidth]{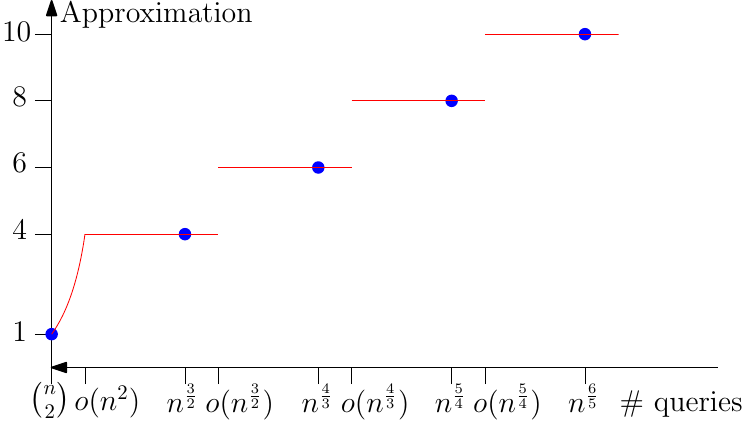}
 \end{center}
 \caption{The universal approximators of \autoref{thm:BGS-approximator} in blue dots vs. the lower bounds of \autoref{thm:impossibility-DA} and \autoref{prop:UA-lb14} in red lines.
 Note that 
 the chart's horizontal axis 
 decreases as one moves from left to right, and it introduces a macroscopic gap between $O(n^{1+1/k})$ queries and $o(n^{1+1/k})$ queries.
 }
\label{fig:UB-Lb-chart}
\end{figure}

In particular, deterministic $O(1)$-approximation algorithms for general $n$-point spaces that use $O(n)$ queries are impossible.
However, 
Barhum et.~al. proved the existence of a universal $O(1)$-approximator using $O(n)$ queries with respect to subsets of Hilbert space.%
\footnote{More precisely, $(1+\e)$-approximator using $O(n/\e^2)$ queries.
We are interested in some constant approximation and will not try to reach an approximation close to $1$.}
Their construction is based on constant-degree expanders.
In this paper, we study deterministic $O(1)$-approximation algorithms using $O(n)$ queries for subsets of the shortest path metric of constant degree graphs.
First, we observe that this is impossible for arbitrary $d$-regular graphs.
\begin{proposition}\label{prop:reg-universal}
For any $\e>0$, any $d\in\{3,4,\ldots\}$, and any finite metric space $X$, 
there exists a $d$-regular graph $G$ such that $c_G(X)\le 1+\e$.
\end{proposition}

In particular, it follows from \autoref{thm:impossibility-DA} that for any $d\in\{3,4,\ldots\}$, 
there are no deterministic $O(1)$-approximation algorithms
for the average distance in $n$-point subsets of $d$-regular graphs that use $O(n)$ queries.
Using the expander from \autoref{cor:rrg-expander-p}, we bypass this impossibility result by providing a universal $O(1)$-approximator for subsets of \emph{typical} constant-degree graphs, solving an open question from~\cite[Remark~2.3]{MN-expanders2}.
Formally, denote by $\DA(\mathcal E,X)$ the infimal $\alpha$ such
that $\mathcal E$ is a deterministic $\alpha$-approximation algorithm of the average distance for finite subsets of $X$. 
We prove the following.

\begin{theorem} \label{thm:u-approxer}
There exists a universal approximator $\mathcal{E}$ that uses $O(n)$ queries for $n$-point subsets
and $\alpha \in [1,\infty)$ such that
for every $d\in \{3,4,5,\ldots\}$,
\[
\lim_{N\to \infty}\;\Pr_{\RRG\sim \mathcal G_{N,d}}
\left[\DA(\mathcal E,\RRG)\leq \alpha)\right]=1.
\]
Here $\mathcal G_{N,d}$ is a uniform distribution on the $N$-vertex $d$-regular simple graphs.
\end{theorem}

\autoref{thm:u-approxer} is a straightforward corollary of
\autoref{cor:rrg-expander-p} with $p=1$.
Its proof is given in \autoref{sec:UA-proofs}, along with a proof of \autoref{thm:impossibility-DA} and \autoref{prop:reg-universal}.

The expander of~\cite{ADTT} mentioned in \autoref{sec:intro-RRG}
gives a different proof of \autoref{thm:u-approxer}. 
However, since the approximator of~\cite{ADTT} is constructed stochastically,
it can be considered a deterministic approximator only in a nonuniform model of computation and hence it is not a deterministic algorithm. 
In contrast, our approximator is based on the zigzag iteration, which can be implemented by a deterministic algorithm in a straightforward way.

\subsection{Distortion growth of embedding into 
\texorpdfstring{$\ell_2$}{l\_2} and \texorpdfstring{$\ell_1$}{l\_1}}
\label{sec:intro-growth-rate}

Motivated by applications of bi-Lipschitz embeddings of finite metric spaces to Computer Science, there is an interest in the growth rate of the distortion of embedding $k$-point subsets.
Let $c_Y(S)$ be the distortion of embedding $S$ in $Y$ as in~\eqref{eq:def:c_Y(X)}, and denote for infinite metric spaces $X$ and $Y$,
\begin{equation*}
D_k(X\embed Y)=\sup \bigl\{c_Y(S):\ {{S\subseteq X,\ |S|\le k}}
\bigr\}.
\end{equation*}
(In the literature, ``$D_k(X\embed Y)$'' is sometimes denoted ``$c_Y^k(X)$''.)

Fundamental results in the embedding theory of finite metric spaces can be described using this notation:
\begin{inparaenum}[(i)]
\item Bourgain's embedding theorem~\cite{Bourgain-embed} and its matching lower bound~\cite{LLR} amount to $D_k(\ell_\infty \embed \ell_1)\asymp D_k(\ell_\infty \embed \ell_2)\asymp \log k$.
\item Enflo's work~\cites{Enflo-type1}{Enflo2} together with~\cite{CNR} (improving on \cite{ALN}) gives
$  D_k(\ell_1\embed \ell_2)\asymp  \sqrt{\log k} $.
\item Bourgain's~\cite{Bourgain-superreflexivity} and Matou\v{s}ek's~\cite{Mat-trees} work on embedding of tree metrics gives $D_k(\text{Trees} \embed\ell_2)\asymp \sqrt{\log \log k}$.
\end{inparaenum}

Those results naturally lead to the study of the possible asymptotics of $D_k(X \embed Y)$ when either $X$ or $Y$ is a fixed ``classical" space, and the other is a general space.
In this paper, we address the possible asymptotic behaviors of $D_k(X \embed \ell_1)$ and $D_k(X\embed \ell_2)$ 
(the universality of $\ell_\infty$ means that $D_k(X\embed \ell_\infty)=1$.)
The behavior of $D_k(X\embed\ell_1)$ is an old folklore question.
The following theorem answers this question.
\begin{theorem} \label{thm:no-embedding-dichotomy-ell_1}
For any metric transform $\varphi:[0,\infty) \to [0,\infty)$ there exists a Hadamard space $X_\varphi$ for which $D_k(X_\varphi\embed \ell_1)\asymp \varphi(\log (k+1))/\varphi(1)$.
\end{theorem}

For $D_k(X\embed \ell_2)$ we have a similar but slightly weaker result.

\begin{theorem} \label{thm:no-embedding-dichotomy-ell_2}
For any metric transform $\varphi:[0,\infty) \to [0,\infty)$ satisfying $\varphi(t)\gtrsim \varphi(1) \sqrt{\log t}$ for any $t\geq 2$, there exists a Hadamard space $X_\varphi$ for which $D_k(X_\varphi\embed \ell_2)\asymp \varphi(\log (k+1))/\varphi(1)$.
\end{theorem}

To our knowledge, the only asymptotics of $D_k(X\embed \ell_2)$ previously known when $X$ is a geodesic space were $\log^\theta k$, for $\theta\in\{0,\tfrac12,1\}$.
The above theorem implies, in particular, that for every $\theta\in[0,1]$ there exists a Hadamard space $X_\theta$ for which $D_k(X_\theta,\ell_2)\asymp \log^\theta (k+1)$.
\autoref{thm:no-embedding-dichotomy-ell_2} is stated for embedding in $\ell_2$, but its proof is easily generalized to
\( D_k(X_\varphi\embed \ell_p)\asymp_p \varphi(\log (k+1))/\varphi(1)\) for any fixed $p\in[1,\infty)$, 
see \autoref{rem:D(X,ell_p)}.
\autoref{thm:no-embedding-dichotomy-ell_2} also raises natural questions about the asymptotics of $D_k(X\embed \ell_2)$ when $X$ is a Banach space.
These questions are open beyond $X\in \{ \ell_1, \ell_2, \ell_\infty\}$. For $p\in(2,\infty)$, it was recently proved
in~\cite{NR25} that $D_k(\ell_p\embed \ell_2)=O(\sqrt{\log n})$.
We ask:

\begin{question}
What is $D_k(\ell_p\embed \ell_2)$ for $p\in(1,2)\cup (2,\infty)$?
Is $D_k(\ell_p \embed \ell_2)\asymp (\log k)^{\frac1p-\frac12}$, when $p\in[1,2]$?
\end{question}

\begin{question}
What are the possible asymptotics of $D_k(X\embed \ell_2)$ when $X$ is a Banach space?
\end{question}

For the interested reader, we mention that complementary problems have received more attention in the literature:
\begin{inparaenum}[(i)]
\item $D_k(\ell_\infty\embed Y)$ was studied in \cites{Bourgain-embed}{LLR}{MN-cotype}.
\item $D_k(\ell_1\embed Y)$ was studied in \cites{BMW}{Enflo-type1}{Enflo2}{ALN}{CNR}.
\item $D_k(\mathrm{Trees}\embed Y)$ was studied in \cites{Bourgain-superreflexivity}{Mat-trees}{LNP-markov-convexity}{MN-markov-convexity}.
\item $D_k(\mathbb R\embed Y)$ was studied in~\cites{Mat-Ramsey}{MN-markov-convexity}{MN-note-dichotomies}.
\item When $Y$ is an infinite-dimensional normed space, it follows from the Dvoretzky theorem that $D_k(\ell_2 \embed Y)=1$.
\end{inparaenum}
For the current understanding of these quantities, the reader is directed to the cited literature
and the mini-survey~\cite{Mendel-dichotomies}.
However, much is still not understood. 
See, for example,~\cite[Question~5]{MN-note-dichotomies}.

\section{Euclidean cones, metric transforms, and embeddings}
\label{sec:transform-cone}

Recall \autoref{def:metric-transform} of the metric transforms. 
Truncations $\varphi(t)=\min\{\tau,t\}$ are a special type of metric transforms, but are sufficiently expressive to approximate general metric transforms:
It was proved in~\cite{BK} that every metric transform is equivalent to a conic combination of truncations and a linear function.
Here we state a special case of \cite[Proposition~3.2.6]{BK}.
\begin{proposition} \label{prop:transform-approx}
    Let $\varphi:[0,\infty) \to [0,\infty)$ be a metric transform.
    Then there exist
     $\alpha_n,\beta_n\ge0$ for $n\ge 0$
     such that
    the function
    \(
\hat \varphi(t)=\alpha_0 t+ \sum_{n=1}^\infty \min\{\alpha_n t , \beta_n\}
    \)
    satisfies $\varphi(t)/2\le \hat \varphi(t)\le 3\varphi(t)$, for every $t\in[0,\infty)$.
    Furthermore, $\lim_{t\to \infty} \frac{\varphi(t)}{t}>0$ if and only if\/ $\alpha_0>0$.
\end{proposition}

The case $q=2$ of the following proposition is proved, for example, in~\cite[Remark~5.4]{MN-quotients}.
The same proof extends to $q\geq 1$ in a straightforward way.

\begin{proposition} \label{prop:phi(sqrt(t))^2}
For any $q\in[1,\infty)$ and any metric transform $\varphi(t)$, $\omega(t)=\varphi\bigl({t}^{1/q}\,\bigr)^q$
is also a metric transform.
\end{proposition}

Combining the above, we deduce the following useful representation for metric transforms.

\begin{corollary}
\label{cor:l2-transform-approximation}
Let $\varphi$ be a metric transform and $q\in[1,\infty)$.
Then, there exist $\alpha_n,\beta_n\ge0$ for $n\ge 0$
such that we have
\begin{equation}
\label{eq:l2-transform-approximation}
\forall \ t\in[0,\infty),\qquad     \tfrac13 \Bigl(\alpha_0 t^q +\sum_{n\ge 1} \min\{\alpha_n t^q,\beta_n\} \Bigr)
\le
     \varphi(t)^q\
     \le
     2 \Bigl(\alpha_0 t^q +\sum_{n\ge 1} \min\{\alpha_n t^q,\beta_n\} \Bigr).
\end{equation}
\end{corollary}

\emph{Euclidean cones} is the basic operation in the proof of \autoref{thm:cat(0)-transform-closure}: 
It is closed for CAT(0) spaces (as we shall see in \autoref{sec:CAT(k)}) and approximates the truncation operation, as we shall see next.%
\footnote{Beyond truncation, the structure of Euclidean cones will be useful in the proofs of \autoref{thm:extrapolation-cat0}, \autoref{lem:gamma-2-union}, and the embedding of high-girth graphs in CAT(0) spaces~\cites{Gromov-rand-groups,Kondo}, which is used in~\autoref{sec:extrapolation-rrg} and in~\autoref{sec:no-dichotomy}.}

\begin{definition}[Euclidean cones]
\label{def:cone}
    Let $(X,d_X)$ be a metric space. The Euclidean cone over $X$, denoted $\cone(X,d_X)=\cone(X)$, is
defined to be the completion of $(0,\infty)\times X$ under the
metric
\begin{equation}\label{eq:def cone}
d_{\cone(X)}\big((s,x),(t,y)\big)^2= {s^2+t^2-2st\cos\left(\min\{\pi,d_X(x,y)\}\right)}.
\end{equation}
Observe that~\eqref{eq:def cone} can be rewritten as follows
\begin{equation}\label{eq:1-cos identity}
d_{\cone(X)}\big((s,x),(t,y)\big)^2 = {(s-t)^2+2st\left(1-\cos\left(\min\{\pi,d_X(x,y)\}\right)\right)}.
\end{equation}
For $y=(s,x)\in \cone(X)$, $s$ is called the \emph{absolute
value} of $y$ and is denoted $s=|y|$. $x$ is called the
\emph{argument} of $y$ and is denoted $x=\arg(y)$.
\end{definition}

This definition is due to~\cite{Ber-Cone}; see also the exposition
in~\cite{ABN86} and~\cite[Chapter~I.5]{BH99}.
In particular, the fact
that $d_{\cone(X)}$ satisfies the triangle inequality is proved
in~\cite[Proposition~I.5.9]{BH99}. When $X$ itself
is complete, the completion of $(0,\infty)\times X$ according to the metric
defined in~\eqref{eq:def cone} amounts to extending the
definition~\eqref{eq:def cone} to $[0,\infty)\times X$ and
taking the quotient of $\{0\}\times X$.
That is, adding an
additional point to $(0,\infty)\times X$ as the ``cusp" of the
cone.
The results discussed below pass immediately from a
space to its completion, so we will mostly ignore this one-point
completion in the ensuing arguments.

\begin{proposition} \label{prop:cone-truncation}
Let $(X,d_X)$ be an arbitrary metric space.
Then
\begin{equation} \label{eq:truncation in cone}
c_{\cone(X)}(X,\min\{\pi,d_X\})\le \pi/2.
\end{equation}
\end{proposition}
In other words, the Euclidean cone over a space $X$ contains bi-Lipschitzly a truncation of $X$.
\begin{proof}
Consider the embedding $f:X\to \cone(X)$ given by
$f(x)=(1,x)$.
Then, by~\eqref{eq:1-cos identity},
\[
d_{\cone(X)}(f(x),f(y))=\sqrt{2(1-\cos(\min\{\pi, d_X(x,y)\})} .
\]
By basic calculus, for $\theta\in[0,\pi]$ we have
\[
1- \frac{4}{\pi^2} \cdot \frac{\theta^2}{2} \ge \cos \theta \ge 1-\frac{\theta^2}{2}.
\]
So, for $\theta= \min\{\pi, d_X(x,y)\}$ we have
\begin{equation}\label{eq:cone-approx-homogene}
\frac{2}{\pi}\cdot \theta \le \sqrt{2(1-\cos \theta)}\le \theta.
\qedhere
\end{equation}
\end{proof}

\begin{lemma}[{\cite{MN-expanders2}}]\label{lem:cone(Lp) in Lp}
For any metric space $X$ and $p\in\{1,2\}$,
\begin{equation} \label{eq:cone(Lp) in Lp}
c_p(\cone(X))\lesssim c_p(X).
\end{equation}
\end{lemma}

The case $p=1$ is proved in
\cite[Corollary~3.8]{MN-expanders2}, while
the case $p=2$ is proved in~\cite[Proposition~5.6]{MN-expanders2}.

\begin{remark}
The statement of \autoref{lem:cone(Lp) in Lp} is false for $p\in(2,\infty)$, and its status is open for $p\in(1,2)$.
This follows from \cite[Lemma~5.9, and Remark~5.12]{MN-quotients} and an equivalence between cone embeddability and truncation embeddability, which we omit here. 
\end{remark}

\begin{corollary} \label{cor:blow-out-cone-in-cone}
Let $(X,d_X)$ be an arbitrary metric space, $p\in\{1,2\}$ and $\delta \ge 1$.
Then
\begin{equation} \label{eq:blow-out-cone-in-cone}
c_p(\cone(X,\delta d_X))\lesssim c_p(\cone(X,d_X)).
\end{equation}
\end{corollary}
\begin{proof}
Since $\delta\ge 1$, we have
\begin{equation}\label{eq:stuiped}
\cone(X,\delta d_X)
  =
\cone \bigl(X,\min \{\delta d_X,\pi\}\bigr)
= \cone (X,\delta\min\{d_X,\pi\}).
\end{equation}
Applying the above, together with~\eqref{eq:truncation in cone} and~\eqref{eq:cone(Lp) in Lp}, we have
\begin{align*}
c_p(\cone(X,\delta d_X))
&\stackrel{\eqref{eq:stuiped}}{=}
c_p \bigl(\cone (X,\delta \min\{d_X,\pi\})\bigr) \\
&\stackrel{\eqref{eq:cone(Lp) in Lp}}{\lesssim}
c_p(X,\delta \min\{ d_X,\pi\})
\\ & = c_p(X, \min\{d_X,\pi\}) \displaybreak[0]
\\ & \leq c_{\cone(X,  d_X)}(X, \min\{d_X,\pi\}) \cdot c_p(\cone(X,  d_X))
\\& \stackrel{\eqref{eq:truncation in cone}}{\leq}
\tfrac \pi 2  c_p(\cone(X,  d_X)). \qedhere
\end{align*}
\end{proof}

We also recall the notions of product and union of pointed metric spaces.

\begin{definition}[$\ell_p$ product] \label{def:orthogonal-product}
Let $((X_n,d_n), x_n)_{n\in\mathbb N}$ be a sequence of metric spaces with
specified points $x_n\in X_n$, and $p\geq 1$.
The \emph{(pointed) $\ell_p$ product $X$ of the spaces $(X_n)_n$} is defined as follows
\[ X=\Bigl(\prod_n (X_n,x_n)\Bigr)_p=\Bigl \{ (y_n)_n\in  \prod_n X_n:\ \sum_n d_n(x_n,y_n)^p<\infty \Bigr\}\]
with the metric $d_X:X\times X\to[0,\infty)$, defined as
\(d_X\bigl ((y_n)_n,(z_n)_n\bigr)^p=\sum_{n=1}^\infty d_n(y_n,z_n)^p.\)
When $\bigl(\prod _n (X_n,x_n)\bigr)_p=\prod_n X_n$ as sets 
(e.g., when the product is finite) we will neglect to mention $\{x_n\}_n$.
The $\ell_2$ product is also called \emph{the orthogonal product}.
\end{definition}

\begin{definition}[$\ell_p$-union] \label{def:orthogonal-union}
Let $((X_n,d_n), x_n)_{n\in \Lambda}$ where $\Lambda$  is an arbitrary index set and $X_n$ are pairwise disjoint metric spaces with
specified points $x_n\in X_n$.
Let
\[X=\Bigl(\biguplus_n (X_n,x_n)\Bigr)_p= \Bigl(\biguplus_n X_n\Bigr) / \{x_n\}_n .\]
That is, we take the disjoint union and identify all the specified points as one point.
Define the metric $d_X:X\times X\to[0,\infty)$ as
\[d_X(y,z)^p=
\begin{cases}
d_n(y,z)^p & y,z\in X_n\\
d_n(y,x_n)^p+ d_m(z,x_m)^p & y\in X_n,\ z\in X_m,\ m\ne n.
\end{cases}
\]
$(X,d_X)$ is a metric space called the (pointed) \emph{$\ell_p$-union} of $(X_n)_n$.
\end{definition}

\begin{lemma} \label{lem:ell1-union}
Let $p\in[1,\infty)$, and let $((X_n,d_n),x_n)_{n\in\Lambda}$
a sequence of pointed metric spaces 
Then,
\begin{compactenum}[(i)]
\item $c_q\bigl(\bigl(\biguplus_n(X_n,x_n)\bigr)_p\bigr)\leq 2^{|1/q-1/p|}\sup_n c_q(X_n)$.
\item Suppose $\bigcup_n X_n\subseteq Y$  and $d_Y(X_m,X_n)\geq \pi$ for every $m,n\in\Lambda$, $m\neq n$. Denote $o_n$ the cusp of\/ $\cone(X_n)$. 
Then the canonical bijection between
$\cone(\bigcup_n X_n)$ and
$\bigl(\biguplus_n (\cone(X_n),o_n)\bigr)_1$ is an isometry.
\item If\/ $p=1$ and $\{X_n\}$ are Hadamard spaces,%
\footnote{Hadamard spaces are defined in\autoref{def:CAT(0)}.} then so is $X$.
\end{compactenum}
\end{lemma}
\begin{proof}
Suppose without loss of generality $\phi_n:X_n\to L_p$ is non-contractive, $L$-Lipschitz embedding such that $\phi(x_n)=0$. 
Then%
\footnote{By $\ell_p(\Lambda,L_p)$ we mean the set of functions 
$f:\Lambda\to L_p$ that vanish on all $\Lambda$ except a countable subset, and $\sum_{\lambda\in\Lambda} |f(\lambda)|^p<\infty$.} 
$\Phi:X\to \ell_p(\Lambda,L_p)$, defined as
$\Phi(x)_n= \phi_n(x)$  if  $x\in X_n$ and $\Phi(x)_n=0$ otherwise,
is clearly $2^{|1/q-1/p|}L$-bi-Lipschitz.

The second claim is immediate from the definitions.

The $\ell_1$ union is a special case of the gluing operation, see \cite[{\S}I.5,{\S}II.11]{BH99}.
In particular, if $\{X_n\}_{n\in\Lambda}$ are Hadamard spaces, then by
\cite[Theorem~II.11.3]{BH99}, so is $X$.
\end{proof}

\begin{proposition} \label{prop:ell2-union-sum}
Let $(Y_n)_n$ be a sequence of metric spaces, $X_n\subseteq Y_n$ be a sequence of subsets, and $x_n\in X_n$.
%
%
Then, for every $n\in \mathbb N$,
\[
X_n\hookrightarrow \Bigl(\biguplus_m (X_m,x_m)\Bigr)_p
\hookrightarrow \Bigl(\prod_m (X_m,x_m)\Bigr)_p
\hookrightarrow \Bigl(\prod_m (Y_m,x_m)\Bigr)_p,
\]
where here ``$\hookrightarrow$'' stands for the canonical embedding.
\end{proposition}
The proof of \autoref{prop:ell2-union-sum} is straightforward and therefore omitted.

\begin{proposition}[Metric transforms embed in weighted orthogonal product of Euclidean cones]
\label{prop:transform-in-producr-cones}
Let $(X,d)$ be an arbitrary metric space and $\varphi:[0,\infty)\to [0,\infty)$ be an arbitrary metric transform, and $q\in[2,\infty)$.
Then $(X,\varphi\circ d)$ can be embedded with distortion at most $6^{1/q}\pi/2$ in a (weighted) $\ell_q$ product of $X$ and Euclidean cones over rescalings of $(X,d)$.
\end{proposition}
\begin{proof}
Let $\omega(t)=\varphi\bigl ({t}^{1/q}\bigr )^q$.
By \autoref{prop:phi(sqrt(t))^2}, $\omega$ is also a metric transform, and therefore, by \autoref{prop:transform-approx} there exist $\alpha_n,\beta_n\ge 0$, for $n\ge 0$
for which
\begin{equation} \label{eq:approx-hat}
\omega(t)/2\le  \hat \omega(t)=\alpha t+ \sum_n \min\{\alpha_n t,\beta_n\}
\le 3 \omega(t).
\end{equation}

We may assume without loss of generality that each $\beta_n>0$.
It follows from \autoref{prop:cone-truncation} that
the metric
$\min\bigl\{({\alpha_n}/{\beta_n})^{1/q}\, \pi d ,\pi\bigr\}$
embeds in $(Y_n,\hat\rho_n)=\cone\bigl(X,({{\alpha_n}/{\beta_n}})^{1/q}\, \pi d\bigr)$ with a distortion at most $\pi/2$.
So for $n\ge 1$, the metric
$\min\bigl\{{{\alpha_n}}^{1/q}  d ,{\beta_n}^{1/q}\bigr\}$ admits bi-Lipschitz embedding  in $(Y_n,\hat\rho_n)$
in which for every $x,y\in X$
\begin{equation} \label{eq:rho-approximation}
\tfrac{2}{\pi} \cdot \min\bigl\{ \alpha_n^{1/q}\,d_X(x,y), {{\beta}^{1/q}_n}\bigr\}
\le \frac{{\beta^{1/q}_n}}{\pi}\hat \rho_n(x,y)\le
\min\bigl\{ \alpha_n^{1/q}\,d_X(x,y), {{\beta}^{1/q}_n}\bigr\}.
\end{equation}

For simplicity of the notation, we assume above that $X\subseteq Y_n$
and the embedding of $X$ in $Y_n$ is the identity mapping.

Denote $\rho_n=\frac{{\beta^{1/q}_n}}{\pi} \hat \rho_n$, and
$(Y_0,\rho_0)=(X,{\alpha^{1/q}_0} d)$.
Fix an arbitrary $o\in X$. Let $o_0=o\in Y_0$, and $o_n=(1,o)\in Y_n$ for $n\geq 1$. 
Let $(Y,\rho)$ be the $\ell_q$ product of the sequence
\( ((Y_n,\rho_n),o_n)_{n\ge 0}\).
We identify $X$ with the subset $\{ (x,x,\ldots)\in  \prod_{n\ge 0} Y_n:\ x\in X\}$.
We next prove that $X\subset Y$, that is, that
for every $x\in X$, $\sum_n \rho_n(o,x)^q<\infty$.
If we allow for $\infty$ values,
$\rho$ is well defined in $(X\times \prod_{n\ge 1} Y_n) \times (X\times \prod_{n\ge 1} Y_n)$,
so it is sufficient to prove that $\rho(o,x)<\infty$ for every $x\in X$.
This would follow as a special case of the estimates on $\rho(x,y)$ for every $x,y\in X$ as follows.
\begin{multline*}
    \rho(x,y)^q=  \sum_{n\ge 0} \rho_n(x,y)^q
    \stackrel{\eqref{eq:rho-approximation}}\le
    \alpha_0 d(x,y)^q + \sum_{n\ge 1} \min\bigl\{ \alpha_n\,d(x,y)^q, \beta_n\bigr\}
    \\ = \hat\omega\bigl (d(x,y)^q\bigr)
    \stackrel{\eqref{eq:approx-hat}}\le
    3\omega\bigl (d(x,y)^q\bigr) = 3 \varphi (d(x,y))^q.
\end{multline*}
\begin{multline*}
    \rho(x,y)^q= \alpha_0 d^q(x,y)+ \sum_{n\ge 1} \rho_n(x,y)^q
    \stackrel{\eqref{eq:rho-approximation}}\ge
    \bigl(\tfrac{2}{\pi}\bigr)^q \cdot \Bigl(\alpha_0 d^q(x,y) + \sum_{n\ge 1} \min\bigl\{ \alpha_n\,d(x,y)^q, \beta_n\bigr\} \Bigr )
    \\ =  \bigl(\tfrac{2}{\pi}\bigr)^q \cdot \hat\omega\bigl (d(x,y)^q\bigr)
    \stackrel{\eqref{eq:approx-hat}}\ge
    \tfrac12\cdot \bigl(\tfrac{2}{\pi}\bigr)^q\cdot \omega\bigl (d(x,y)^q\bigr)
     = \tfrac12\cdot \bigl(\tfrac{2}{\pi}\bigr)^q \cdot \varphi (d(x,y))^q.
\end{multline*}
Therefore, $(X,\varphi\circ d)$ embeds in $(Y,\rho)$ with a distortion at most $6^{1/q}\pi/2$. 
\end{proof}

We end the section with a result about extensions of Euclidean embeddings of combinatorial graphs to embeddings of the one-dimensional simplicial complexes
defined by these graphs.
For a weighted undirected graph $G=(V,E,w)$ and $e=(u,v)\in E$, we denote by $[e]$ the closed interval of length $w(u,v)$ whose endpoints are $u$ and $v$. 
Let $\Sigma(G)=\bigcup_{e\in E} [e]$ denotes the one-dimensional simplicial complex defined by $G$,
and $d_{\Sigma(G)}$ be the geodesic metric on $\Sigma(G)$.
The following theorem extends a bi-Lipschitz embedding into Hilbert space $\phi:G\to \mathcal H_1$ to a bi-Lipschitz embedding $\Phi:\Sigma(G)\to \mathcal H_2\supset \mathcal H_1$ into Hilbert super-space.
This type of extension was studied in~\cite{MMMR}, 
where it was coined \emph{bi-Lipschitz outer extension}.

\begin{lemma}\label{lem:biLip-extension}
Let $G=(V,E,w)$ be an undirected weighted graph (finite or infinite), and $W\subseteq V$. 
Suppose that there is a bi-Lipschitz embedding $\phi:(W,d_G)\to \mathcal H_1$ in Hilbert space such that for every $x,y\in W$,
\[
d_G(x,y)\le \|\phi(x)-\phi(y)\|_{\mathcal H_1} \le L\cdot d_G(x,y).
\]
Denote
\begin{equation}\label{eq:def-Sigma(W)}
\Sigma(W)=\bigcup\bigl\{[u,v]:\ u,v\in W,\; \{u,v\}\in E\bigr\}
\subseteq \Sigma(G).
\end{equation}
Then for any $\alpha\geq 0$ there exists an embedding\/ $\Phi=\Phi_\alpha:\Sigma(W)\to \mathcal H_2=(\mathcal H_1 \times \mathcal H_3)_2$
into Hilbert space such that for every $u\in W$, $\Phi(u)=(\phi(u),0)$ and
for any $x,y\in \Sigma(W)$,
\[
\frac{\alpha}{\sqrt{2+\alpha^2}}\, d_{\Sigma(G)}(x,y) \le \|\Phi(x)-\Phi(y)\|_{\mathcal H_2} \le \sqrt{L^2+\alpha^2(L+1)^2}\cdot d_{\Sigma(G)}(x,y).
\]
In particular,
\ for $\alpha=\sqrt{\sqrt2L/(L+1)}$,\;
\( \mathrm{dist}(\Phi)\leq (1+\sqrt{2})L+\sqrt{2}\).
\end{lemma}

\begin{corollary}\label{cor:biLip-extension}
For any undirected weighted graph $G=(V,E,w)$,
$c_2(\Sigma(G))\leq (1+\sqrt{2}) c_2(G)+\sqrt2$.
\end{corollary}

Observe that the combinatorial triangle $C _3$ is a uniform metric on three points and therefore $c_2(C_3)=1$. 
On the other hand, $\Sigma(C_3)\cong \tfrac{3}{2\pi}\mathbb S^1$ is a geodesic cycle for which $c_2(\mathbb S^1)=\frac\pi2$~\cite{LM00}.

\begin{question}
    What is the infimal $\alpha\in[\pi/2,1+2\sqrt{2}]$ such that $c_2(\Sigma(G))\leq \alpha\cdot c_2(G)$ for any undirected graph $G$ (possibly weighted and infinite)?
    Similarly, what is the infimal $\beta\in[1,1+\sqrt{2}]$ for which $c_2(\Sigma(G))\leq (\beta+o(1)) c_2(G)$?
\end{question}

\begin{proof}[Proof of \autoref{lem:biLip-extension}]
Define
$\Phi:\Sigma(W)\to  \mathcal H_2=\bigl(\mathcal H_1 \times\ell_2(E(W))\bigr)_2$ 
as a direct sum of two embeddings as follows.
The first is $\psi:\Sigma(W)\to \ell_2(E(W))$.
Fix  $\alpha=\alpha(L)>0$ to be determined later.
Fix an edge $\{u,v\}\in E(W)$, $u,v\in W$ and $x\in[u,v]$.
Let $\psi (x)=(\psi(x)_e)_{e\in E(W)}\in \ell_2(E(W))$ be
\[
\psi (x)_e =\begin{cases} \min \{d_{\Sigma(G)}(x,u),d_{\Sigma(G)}(x,v) \}\sqrt{\alpha} (L+1), & e=\{u,v\} \\
0, & e\neq \{u,v\}.
\end{cases}
\]
Observe that $\psi(u)=0$, $\forall u\in W$.
Next define $\tilde \phi:\Sigma(W)\to \mathcal H_1$ using the notation above, as follows.
\[
\tilde \phi(x)=
\frac{d_{\Sigma(G)}(x,v) \phi(u)+ d_{\Sigma(G)}(x,u)\phi(v)}
{d_G(u,v)}.
\]
It is straightforward to check that $\tilde \phi$ is an extension of $\phi$.
Unfortunately, it is not necessarily a bi-Lipschitz extension.
So, we define an outer extension
\[
\Phi(x)= \tilde\phi(x) \oplus \psi( x).
\]

Next, we bound the distortion of $\Phi$.
Fix $x,y\in \Sigma(W)$. Let $x\in[u,v]$, $y\in [r,s]$ be such that $\{u,v\},\{r,s\}\in E$,  and $u,v,r,s\in W$.
 Suppose, without loss of generality, that $u,v,r,s$ are such that
for the parameters $\gamma=d_{\Sigma(G)}(u,x)$, and $\delta= d_{\Sigma(G)}(r,y)$, we have
\begin{equation}
\label{eq:sigma-distance}
d_{\Sigma(G)}(x,y)=\begin{cases}
|\gamma-\delta|, & \ (u,v)=(r,s)\\
\gamma+ d_G(u,r) + \delta, & \ \text{otherwise}.
\end{cases}
\end{equation}

It is evident from~\eqref{eq:sigma-distance} that
to bound $\|\Phi\|_\mathrm{Lip}$
from above, it is sufficient to bound for $x,x'\in [u,v]$ for some edge $\{u,v\}\in E$, $u,v\in W$, and separately, for $r,s\in W$
\begin{align}
\nonumber \|\Phi(x)-\Phi(x')\|_{\mathcal H_2}^2 & \leq 
d_{\Sigma(G)}(x,x')^2\frac{\|\phi(u)-\phi(v)\|_{\mathcal H_1}^2}{d_G(u,v)^2}+ \alpha (L+1)^2 d_{\Sigma(G)}(x,x')^2
\\ 
\label{eq:extension-expansion} &\le (L^2+\alpha(L+1)^2) \, d_{\Sigma(G)}(x,x')^2;\\
\nonumber \|\Phi(r)-\Phi(s)\|_{\mathcal H_2}^2 & =\|\phi(r)-\phi(s)\|_{\mathcal H_1}^2 \le L^2 d_{G}(r,s)^2.
\end{align}
Hence,
\[
\|\Phi(x)-\Phi(y)\|_{\mathcal H_2} \leq \sqrt{L^2+\alpha(L+1)^2}\, d_{\Sigma(G)}(x,y).
\]

\medskip

Next, we bound the contraction of $\Phi$.
First, we consider the case  $\{u,v\}= \{r,s\}$. In this case, similar to~\eqref{eq:extension-expansion} above, we have
\begin{align*}
\|\Phi(x)-\Phi(y)\|_{\mathcal H_2} & \geq \|\tilde \phi(x) -\tilde \phi(y)\|_{\mathcal H_1} =
d_{\Sigma(G)}(x,y) \frac{\|\phi(u)-\phi(v)\|_{\mathcal H_1}}{d_G(u,v)}
\ge d_{\Sigma(G)}(x,y).
\end{align*}

Next, we consider $\{u,v\}\neq \{r,s\}$. 
Assume (without loss of generality) that
$\zeta:=d_{\Sigma(G)}(x,u)\leq d_{\Sigma(G)}(x,v)$ and
$\eta:=d_{\Sigma(G)}(y,r)\leq d_{\Sigma(G)}(y,s)$.
Denote $\tilde\zeta=\zeta/d_G(u,v)\leq 1/2$ 
and $\tilde \eta=\eta/d_G(r,s)\leq 1/2$.
By the triangle inequality
\[
d_{\Sigma(G)}(x,y)\leq \zeta+d_G(u,r)+\eta.
\]
Note that
\[\|\psi(x)-\psi(y)\|_2^2=\alpha (L+1)^2(\zeta^2+\eta^2)
\geq \alpha (L+1)^2(\zeta+\eta)^2/2,\] 
and
\begin{align*}
\|\tilde\phi(x)-\tilde\phi(y)\|_{\mathcal H_1} &=
\| \bigl((1-\tilde\zeta) \phi(u)+ \tilde \zeta \phi(v)\bigr) -
\bigl((1-\tilde \eta) \phi(r)+ \tilde \eta \phi(s)\bigr) \|_{\mathcal H_1}
\\ &=\|\phi(u)-\phi(r) +\tilde \zeta(\phi(v)-\phi(u)) +\tilde \eta(\phi(r)-\phi(s)) \|_{\mathcal H_1}
\\ &\ge \max \bigl\{\|\phi(u)-\phi(r)\|_{\mathcal H_1} - B ,0\bigr\}
\\ &\ge \max\{d_G(u,r) -B,0\},
\end{align*}
where $B$ is defined as 
\begin{equation*}
B= \|\tilde \zeta(\phi(v)-\phi(u)) +\tilde \eta(\phi(r)-\phi(s)) \|_{\mathcal H_1}
\le
\tilde \zeta \|\phi(v)-\phi(u)\|_{\mathcal H_1} +\tilde \eta \|\phi(r)-\phi(s) \|_{\mathcal H_1}
\le (\zeta+\eta)L.
\end{equation*}
So 
\begin{multline*}
\|\Phi(x)-\Phi(y)\|_{\mathcal H_2}^2 \geq
\max\{0, d_G(u,r) - (\zeta+\eta)L\}^2 +\alpha (\zeta+\eta)^2 (L+1)^2/2
\\ \geq 
\max\{0, d_{\Sigma(G)}(x,y) - (\zeta+\eta)(L+1)\}^2 +\alpha (\zeta+\eta)^2 (L+1)^2/2.
\end{multline*}
The right-hand side of the above is minimized when  
$\zeta+\eta=\frac{2d_{\Sigma(G)}(x,y)} {(L+1)(2+\alpha)}$ 
for which
\[ \|\Phi(x)-\Phi(y)\|_{\mathcal H_2} \geq \sqrt{\frac{\alpha}{{2+\alpha}}}\, {d_{\Sigma(G)}(x,y)}.\]

Hence
\[
\mathrm{dist}(\Phi)\leq \sqrt{\frac{{(2+\alpha)(L^2+\alpha(L+1)^2)}}{\alpha}}.
\]
Substituting $\alpha=\frac{\sqrt{2}L}{L+1}$
in the above estimate, we obtain
\(
\mathrm{dist}(\Phi)\leq (1+\sqrt{2})L+\sqrt{2}.
\)
\end{proof}

\section{Alexandrov spaces}
\label{sec:CAT(k)}

\subsection{CAT(0) spaces}

CAT$(0)$ spaces are geodesic metric spaces in which the distances between points along geodesic triangles are bounded from above by the corresponding distances
in a similar triangle in the Euclidean plane.
Formally:

\begin{definition}[CAT$(0)$ spaces] \label{def:CAT(0)}
A metric space $(X,d)$ is said to be a CAT$(0)$ space if every two
points $x,y\in X$ can be joined by a geodesic, and for every
$x,y,z\in X$, every (constant speed) geodesic $\phi:[0,1]\to X$ with $\phi(0)=y$ and
$\phi(1)=z$, and every $t\in [0,1]$ we have
\begin{equation}\label{eq:CAT(0)}
d_X(x,\phi(t))^2\le (1-t)d(x,y)^2+td(x,z)^2-t(1-t)d(y,z)^2.
\end{equation}
Complete CAT(0) spaces are also called \emph{Hadamard spaces}.
\end{definition}

The CAT$(0)$ property, when applied locally, is a property of
Riemannian manifolds with nonpositive sectional curvature.
There is a dual property which is a generalization of manifolds of nonnegative sectional curvature.

\begin{definition}[Alexandrov spaces of nonnegative curvature]
\label{def:NNC}
A metric space that satisfies \autoref{def:CAT(0)} but with the inequality~\eqref{eq:CAT(0)} \emph{reversed} is called
\emph{Alexandrov space of nonnegative curvature}.
\end{definition}

The following observation is straightforward. 
\begin{observation} \label{obs:scaled-CAT(0)}
Let $\alpha>0$, and $(X,d)$ be a CAT(0) space
[Alexandrov spaces of nonnegative curvature],
then so is $(X,\alpha d)$.
\end{observation}

\begin{proposition} \label{prop:CAT(0)-l2sum}
Let $((X_n,d_n),x_n\in X_n)_n$ be a sequence of pointed CAT$(0)$  spaces [Aleksandrov spaces of nonnegative curvature].
Then $X=\bigl(\prod_m (X_m,x_m)\bigr)_2$  is also a CAT$(0)$ space [an Aleksandrov space of nonnegative curvature].
\end{proposition}
\begin{proof}
We consider the CAT$(0)$ case.
Fix $x=(x_n)_n\in X$, $y=(y_n)_n\in X$, $z=(z_n)_n\in X$.
and a geodesic $\phi=(\phi_n)_n:[0,1]\to X$ between
$\phi(0)=y$ and $\phi(1)=z$.
It is easy to check that every $\phi_n$ is a geodesic between $y_n$ and $z_n$ in $X_n$. Applying~\eqref{eq:CAT(0)} on $x_n,y_n,z_n\in X_n$ and the geodesic $\phi_n$, we have
\begin{equation*}
d_{n}(x_n,\phi_n(t))^2\le (1-t)d_{n}(x_n,y_n)^2+td_{n}(x_n,z_n)^2
-t(1-t)d_{n}(y_n,z_n)^2.
\end{equation*}
Summing these inequalities over $n$ we obtain the CAT$(0)$ inequality in $X$.
For Aleksandrov spaces of nonnegative curvature, the proof is similar, with the inequalities reversed.
\end{proof}

Combining \autoref{prop:CAT(0)-l2sum} and \autoref{prop:ell2-union-sum} we arrive at the following corollary. 

\begin{corollary} \label{cor:CAT(0)-union}
Fix $D\ge 1$.
Let $((X_n,d_n),x_n)$ be a sequence of pointed metric spaces such that for every $n\in \mathbb N$, $X_n$
embeds in some CAT$(0)$ space [Alexandrov space of nonnegative curvature] with distortion at most $D$.
Then $\bigl(\prod_n (X_n,x_n)\bigr)_2$  embeds in some CAT$(0)$ space [Alexandrov space of nonnegative curvature] with distortion at most $D$.
\end{corollary}

We will also need the analog of CAT$(0)$ for the curvature $\kappa=1$.
\begin{definition}[CAT(1) spaces]
\label{def:CAT(1)}
A  metric space $(X,d_X)$ is a CAT$(1)$ space if it satisfies the following conditions. First, for every $x,y\in X$ with $d_X(x,y)<\pi$ there exists a geodesic joining $x$ to $y$, that is, a curve $\phi:[0,1]\to X$ that satisfies $d_X(\phi(t),x)=td_X(x,y)$ and $d_X(\phi(t),y)=(1-t)d_X(x,y)$
 for every $t\in [0,1]$. Suppose that $x,y,z\in X$ satisfy
 $$
 d_X(x,y)+d_X(y,z)+d_X(z,x)<2\pi,
 $$
and that $\phi_{x,y},\phi_{y,z},\phi_{z,x}:[0,1]\to X$ are the geodesics
that join $x$ to $y$, $y$ to $z$, and $z$ to $x$, respectively. Let
$\mathbb{S}^2$ be the unit Euclidean sphere in $\R^3$, and let $d_{\mathbb S^2}$
denote the geodesic metric on $\mathbb{S}^2$ (thus the diameter of $\mathbb S^2$
equals $\pi$ under this metric).
As explained in~\cite{BH99},
there exist $a,b,c\in \mathbb{S}^2$ such that $d_{\mathbb{S}^2}(a,b)=d_X(x,y)$,
$d_{\mathbb{S}^2}(b,c)=d_X(y,z)$ and $d_{\mathbb{S}^2}(c,a)=d_X(z,x)$. Let
$\phi_{a,b},\phi_{b,c},\phi_{c,a}:[0,1]\to \mathbb{S}^2$ be the geodesics that join $a$
to $b$, $b$ to $c$, and $c$ to $a$, respectively. Then the remaining
requirement in the definition of a CAT$(1)$ space is that for every
$s,t\in [0,1]$ we have $d_X(\phi_{x,y}(s),\phi_{y,z}(t))\le
d_{\mathbb{S}^2}(\phi_{a,b}(s),\phi_{b,c}(t))$.
\end{definition}

Note that \autoref{def:CAT(0)} can be cast in the same language as \autoref{def:CAT(1)} but replacing the model space $\mathbb{S}^2$ with the Euclidean plane $\mathbb R^2$.
For more on these fundamental notions,
see~\cite{BGP92,Per95,BH99,Stu99,BBI01,Sturm-NPC,Gro07}.
The following proposition is a straightforward corollary
of the relation between distances in $\mathbb R^2$ and
$\mathbb S^2$, see~\cite[Theorem~II.1.12]{BH99}.

\begin{proposition} \label{prop:CAT(0)=>CAT1}
 Every CAT(0) space is a CAT(1) space.
\end{proposition}

An important
theorem of Berestovski{\u\i}~\cite{Ber-Cone} (see also~\cite{ABN86}
and~\cite[Theoerm~II.3.14]{BH99}) asserts
\begin{theorem} \label{thm:berestovskii}
A metric space $(X,d_X)$ is CAT$(1)$ if and only if\/ $\cone(X)$ is  a CAT$(0)$ space.
\end{theorem}

\begin{corollary} \label{cor:cone-stable}
    A Euclidean cone of a subset of a CAT(0) space is a subset of \emph{some} CAT(0) space.
\end{corollary}

Another corollary of \autoref{thm:berestovskii} and \autoref{prop:cone-truncation}:
\begin{corollary} \label{cor:CAT1-trunc-is-CAT(0)}
    Let $(X,d_X)$ be a CAT$(1)$ space.
    The metric $(X,\min\{d_X,\pi\})$ embeds in some CAT$(0)$ space with distortion at most $\pi /2$.
\end{corollary}

Finally, we are in a position to prove that the class of CAT$(0)$ spaces is stable under metric transforms.

\begin{proof}[Proof of \autoref{thm:cat(0)-transform-closure}]
By \autoref{prop:transform-in-producr-cones}
$(X,\varphi\circ d)$
embeds with a distortion smaller at most $\pi \sqrt{3/2}$ in the space $Y$ that is an orthogonal product of a scaled $(X,d_X)$ and an orthogonal product of a sequence of scaled Euclidean cones of scaled versions of $(X,d_X)$.
But by \autoref{prop:CAT(0)-l2sum}, \autoref{obs:scaled-CAT(0)}, and \autoref{cor:cone-stable}, these operations are closed for CAT(0) spaces. Therefore, $Y$ is also a CAT(0) space.
\end{proof}

\begin{question}
Are subsets of Aleksandrov spaces of nonnegative curvature bi-Lipschitzly closed under metric transforms?
\end{question}

\subsection{Angles and barycenters}

In CAT(0) spaces there are natural notions of angles and tangent cones ~\cite[{I.1.12}, II.3.1]{BH99},
which we now define.

\begin{definition}
[Alexandrov angles in Hadamard spaces]
\label{def:alexandrov-angle}
Let $(X,d)$ be a Hadamard space and $p,x,y\in X$.
Let $\phi:[0,d(p,x)]\to X$ be the (unique) geodesic path connecting $p$ and $x$, and $\phi':[0,d(p,y)]\to X$ be the geodesic path connecting $p$ and $y$.
The (Alexandrov) angle $\angle_p(x,y)\in[0,\pi]$ between $\phi$ and $\phi'$ is defined as
\begin{equation} \label{eq:cat0-angle}
\angle_p(x,y)= \limsup_{t,t'\to 0^+}
\overline{\angle}_p(\phi(t),\phi'(t')).
\end{equation}
where $\overline{\angle}_p(u,v)\in[0,\pi]$ is defined using
\emph{the law of cosines},
\begin{equation}\label{eq:comp-triangle-law}
\cos \overline{\angle}_p(u,v)= \frac{d(p,u)^2+d(p,v)^2-d(u,v)^2}{2d(p,u)d(p,v)}.
\end{equation}
\end{definition}

\begin{proposition}[Triangle inequality for angles {(see~\cite[{}I.1.14]{BH99})}]
\label{prop:angle-triangle-ineq}
Let $(X,d)$ be a Hadamard 
space, and $p,x,y,z\in X$. Then
\[
\angle_p(x,z) \le \angle_p(x,y)+\angle_p(y,z).
\]
\end{proposition}

\begin{definition}[Tangent Cones {(see~\cite[Def~I.3.18]{BH99})}]
\label{def:tangent}
Let $(X,d)$ be a Hadamard 
space and $p\in X$.
Define an equivalence relation $\sim$ on $X\setminus \{p\}$:
$x\sim y$ iff $\angle_p(x,y)=0$. Denote by $[x]_p$ the equivalence class of $x\in X$ in the relation $\sim$.
The \emph{space of directions} $S_pX$ is the completion of the quotient
$\bigl((X\setminus \{p\})/\sim \bigr)$ with distance $d_{S_pX}(x,y)=\angle_p(x,y)$. By \autoref{prop:angle-triangle-ineq}, this is a complete metric space.
The \emph{tangent cone} at $p\in X$ is defined to be $T_pX=\cone(S_pX)$, the Euclidean cone over the space of directions at $p\in X$. Since $S_pX$ is complete, so is $T_pX$ (see~\cite[Proposition~I.5.9]{BH99}).
\end{definition}

We shall need the following fact about the tangent cones of CAT(0) spaces.

\begin{proposition}[{\cite{Nikolaev}, see also~\cite[Theorem~I.3.19]{BH99}}]
If $(X,d)$ is a CAT$(0)$ space and $p\in X$, then $S_pX$ is a CAT$(1)$ space and $T_pX$ is a CAT$(0)$ space.
\end{proposition}

\begin{definition}[Logarithmic map]
For a CAT(0) space $(X,d_X)$ and $p\in X$,
define the logarithmic map%
\footnote{While the exponential map from Riemannian geometry is not well-defined for general CAT(0) spaces, its inverse, the logarithmic map, is well-defined for any CAT(0) space but not necessarily injective.}
 $\Log_p:(X,d_X)\to T_pX$ by
$\Log_p(x)=(d_X(p,x),[x]_p)$.
For $\lambda\in[0,1]$ and $p,x\in X$ define $\lambda[p,x]\in X$  as follows: Let $\phi_{[p,x]}:[0,d_X(p,x)]\to X$ be the geodesic path connecting $p$ to $x$ in $X$.
Then $\lambda[p,x] =\phi_{[p,x]}(\lambda\cdot d(p,x))\in X$.
\end{definition}

\begin{lemma} \label{lem:tangent-distance}
Let $(X,d_X)$ be a Hadamard space and $p,x,y\in X$. Then
\begin{equation} \label{eq:T-distances}
d_{T_pX}(\Log_p(x),\Log_p(y))=
\lim_{\lambda\to 0} \frac{d_X(\lambda[p,x], \lambda[p,y])}{\lambda}.
\end{equation}
\end{lemma}
\begin{proof}
For $x\in X$, denote $\|x\|=d_X(p,x)$. Then
\begin{align*}
d_{T_pX}&(\Log_px,\Log_py) ^2
\\& \stackrel{\eqref{eq:def cone}}= \|x\|^2 + \|y\|^2 -2\|x\|\,\|y\|\cos \angle_p(x,y) \\
&\stackrel{\eqref{eq:cat0-angle},\eqref{eq:comp-triangle-law}}= \lim_{\lambda\to 0}  \biggl(
\|x\|^2 + \|y\|^2 -2\|x\|\,\|y\| \frac{(\lambda\|x\|)^2+(\lambda\|y\|)^2
-d_X(\lambda[p,x],\lambda[p,y])^2}
{2\lambda\|x\|\lambda\|y\|}
\biggr)\\
&= \lim_{\lambda\to 0} \frac{d_X(\lambda[p,x],
\lambda[p,y])^2}{\lambda^2}. \qedhere
\end{align*}
\end{proof}

It follows immediately from the above lemma and from the convexity of the distance function in CAT(0) spaces
(see also~\cite[Lemma~II.3.20]{BH99}) that
$\Log_p$ is $1$-Lipschitz, i.e.,
\begin{equation} \label{eq:busemann}
\forall \ p,x,y\in X, \qquad d_{T_p X}(\Log_p (x),\Log_p (y)) \le d_X(x,y).
\end{equation}

Similarly to Hilbert spaces, Hadamard spaces have the concept of barycenter, which we discuss next.

\begin{theorem}[{Barycenters in Hadamard space, see~\cite[Proposition~4.3]{Sturm-NPC}}]
\label{thm:barycenter}
Let $(X,d)$ be a Hadamard space and $\mu$ be a Borel measure with finite second moment (i.e., $\int_X d(x_0,x)^2\dd\mu(x)<\infty$
for some $x_0\in X$).
Then there exists a unique point in $X$, denoted $\mathcal{B}(\mu)$, that minimizes the expression
\(
\int_X d(b,x)^2\dd\mu(x)
\)
over $b\in X$.
\end{theorem}

\begin{proposition}[{$2$-barycentric property, see~\cite[Theorem~4.9]{Sturm-NPC}}]
Let $(X,d)$ be a Hadamard space and $\mu$ be a Borel measure with finite second moment. Then for any $z\in X$,
\begin{equation} \label{eq:2-barycentric}
\int_X d(z,x)^2d\mu(x)\ge d(z,\mathcal{B}(\mu))^2+
\int_X d(\mathcal{B}(\mu),x)^2 d\mu(x).
\end{equation}
\end{proposition}

Let $X$ and $Y$ be probability spaces and $f:X\to Y$ measurable. Let $\mu$ be a probability measure on $X$. The \emph{push-forward} measure $\mu_{\#f}$ on $Y$ is defined as $\mu_{\#f}(B)=\mu(f^{-1}(B))$, for any measurable $B\subseteq Y$.
\begin{proposition} \label{prop:cusp-barycenter}
Let $(X,d)$ be a Hadamard space and $\mu$ be a Borel measure with finite second moment. Let $b=\mathcal{B}(\mu)$ be the barycenter of $\mu$. Then in $T_bX$, $\mathcal{B}(\mu_{\#\Log_b})=0_{T_bX}$, the cusp of $T_bX$.
\end{proposition}

\begin{lemma} \label{lem:fixed-barycenter}
Let $(X,d)$ be a Hadamard space and $\mu$ be a Borel measure with finite second moment. Let $b=\mathcal{B}(\mu)$ be the barycenter of $\mu$. Let $(\lambda [b,\cdot]):X\to X$ be defined as
$(\lambda[b,\cdot])(x)=\lambda [b,x]$.
Then for every $\lambda\in [0,1]$ we have
\(
\mathcal{B}(\mu_{\#(\lambda[b,\cdot])})=b.
\)
\end{lemma}
\begin{proof}
Fix $\lambda\in[0,1]$, and
denote $c=\mathcal{B}(\mu_{\#(\lambda[b,\cdot])})$,
$\delta= d_X(b,c)$, and
\[
 A^2=\int_X d_X(b,x)^2 \dd\mu(x),
\quad B^2=\int_X d_X(c,y)^2 \dd\mu_{\#(\lambda[b,\cdot])}(y)
=\int_X d_X(c,\lambda[b,x])^2 \dd\mu(x).
\]
Our goal is to prove that $\delta=0$.

From the $2$-barycentric property~\eqref{eq:2-barycentric} for
$\mu_{\#(\lambda [b,\cdot])}$,
\begin{equation} \label{eq:2-bary-c}
B^2 +\delta^2 \le \int_X d(b, \lambda[b,x])^2 \dd\mu(x)
=\lambda^2 A^2.
\end{equation}
From the $2$-barycentric property~\eqref{eq:2-barycentric} for $\mu$,
\begin{equation} \label{eq:2-bary-b}
\begin{aligned}
A^2+\delta^2 & \leq \int_X d_X(c,x)^2 \dd\mu(x)\\
& \leq  \int_X \bigl( d_X(c,\lambda[b,x])+d_X(\lambda[b,x],x) \bigr )^2\dd\mu(x)\\
&= B^2+ (1-\lambda)^2 A^2+{2(1-\lambda)}
\int_X \bigl (d_X(c,\lambda[b,x])\cdot d_X(b,x)\bigr)\dd\mu(x)\\
&\stackrel{(*)}\le B^2+(1-\lambda)^2 A^2 + 2(1-\lambda) BA\\
&\stackrel{\eqref{eq:2-bary-c}}\le  B^2+(1-\lambda)^2 A^2 +2\lambda(1-\lambda) A^2,
\end{aligned}
\end{equation}
where Inequality (*) is implied by Cauchy-Schwarz  inequality on the integral.
Summing \eqref{eq:2-bary-c} and \eqref{eq:2-bary-b}, and isolating $\delta^2$, we obtain
\begin{equation*}
2\delta^2 \le \lambda^2 A^2-B^2+B^2+(1-\lambda)^2A^2+2\lambda(1-\lambda)A^2-A^2=0. \qedhere
\end{equation*}
\end{proof}

\begin{proof}[Proof of \autoref{prop:cusp-barycenter}.]
Let $\tilde z =(s,[z]_b) \in T_bX \setminus\{0\}$ be an arbitrary non-cusp point in $T_bX$,
where $z\in X\setminus\{b\}$, and $s>0$.
Let $\eta =s/d(b,z)$.
Observe that
$\lambda \eta[b,z]\in X$ is well defined for $\lambda\in[0,\min\{1,\eta^{-1}\}]$.
Then,
\begin{align*}
\int_{T_bX}  d_{T_bX} (\tilde z,y)^2\dd\mu_{\#\Log_b}(y)
& =
\int_X d_{T_bX} (\tilde z,\Log_b(x))^2\dd\mu(x) \\
& \stackrel{\eqref{eq:T-distances}}{=}
\int_X \lim_{\lambda\to 0} \frac{d_X(\lambda\eta[b,z] , \lambda [b,x])^2}{\lambda^2}\dd\mu(x) \\
& = \lim_{\lambda\to 0} \lambda^{-2} \int_X
d_X(\lambda\eta[b,z], \lambda[b,x])^2 \dd\mu(x) \displaybreak[0]\\
& = \lim_{\lambda\to 0} \lambda^{-2} \int_X
d_X(\lambda\eta[b,z], y)^2 \dd\mu_{\#(\lambda[b,\cdot])}(y) \\
&\stackrel{(*)}\ge \lim_{\lambda\to 0} \lambda^{-2}
\Bigl( d_X(\lambda \eta[b,z],b)^2+ \int_X
d_X(b, y)^2 \dd\mu_{\#(\lambda[b,\cdot])}(y) \Bigr)\displaybreak[0]\\
&= \eta^2 d_X(z,b)^2+ \int_X d_X (b,  x)^2 d\mu(x)\\
&>\int_X d (b,  x)^2 d\mu(x)\\
&= \int_X d_{T_bX} (0, \Log_b x)^2 d\mu(x)\\
& = \int_{T_bX}  d_{T_bX} (0,y)^2\dd\mu_{\#\Log_b}(y).
\end{align*}
Inequality~(*) follows from \autoref{lem:fixed-barycenter} and the $2$-barydentric inequality~\eqref{eq:2-barycentric}.
The above inequality means that $\tilde z=0$ is the unique minimizer of $\tilde z\mapsto \int_{T_bX}  d_{T_bX} (\tilde z,y)^2\dd\mu_{\#\Log_b}(y)$ which means that $\mathcal{B}(\mu_{\#\Log_b})=0$ (see \autoref{thm:barycenter}).
\end{proof}

The $2$-barycentric property can be generalized to the $p$-barycentric property for any fixed $p\in [2,\infty)$.
Although this property is natural and not hard to prove, we have not been able to find it in the literature.
We therefore provide a proof here.  
The proof is a straightforward adaptation of the proof in~\cite[§6]{MN-superexpanders} of an analogous proposition for $2$-convex Banach spaces.
Independently, a similar proposition with a somewhat weaker constant has been proved by Gietl~\cite[Theorem~2]{Gietl25}.

\begin{proposition}[$p$-barycentric property of CAT(0), for $p\in[2,\infty)$]
\label{prop:p-barycentric-CAT0}
Let $p\in[2,\infty)$.
Let $(X,d_X)$ be a Hadamard space and $\mu$ be a Borel measure with finite $p$-th moment, and let $\mathcal{B}(\mu)$ its barycenter (according to \autoref{thm:barycenter}). Then for any $z\in X$,
\begin{equation} \label{eq:p-barycentric}
\int_X d_X(z,x)^pd\mu(x)\ge d(z,\mathcal{B}(\mu))^p+
\frac{1}{2^{p-1}-1} \int_X d_X(\mathcal{B}(\mu),x)^p d\mu(x).
\end{equation}
\end{proposition}
\begin{proof}

Apply~\eqref{eq:CAT(0)} with $t=1/2$ and 
 $m$ being the midpoint between $x$ and $y$,
\[
\Bigl(d_X(z,m)^2+\frac{d(x,y)^2}{2^2}\Bigr)^{\frac12}
\leq 
\Bigl(\frac{d_X(z,x)^2+d_X(z,y)^2}{2}\Bigr)^{\frac12}
\]
Since
\[
\Bigl(d_X(z,m)^p+\frac{d(x,y)^p}{2^p}\Bigr)^{\frac1p}
\leq
\Bigl(d_X(z,m)^2+\frac{d(x,y)^2}{2^2}\Bigr)^{\frac12}
\ \text{and}\ 
\Bigl(\frac{d_X(z,x)^2+d_X(z,y)^2}{2}\Bigr)^{\frac12}
\leq
\Bigl(\frac{d_X(z,x)^p+d_X(z,y)^p}{2}\Bigr)^{\frac1p},
\]
we conclude that 
\begin{equation} \label{eq:p-convexity}
d_X(z,m)^p+2^{-p}d(x,y)^p
\leq 
\frac{d_X(z,x)^p+d_X(z,y)^p}{2}.
\end{equation}
Denote by $L_p(X)$ the Borel measures on $X$ with finite $p$-th moment.
Denote
\begin{equation} \label{eq:phi-nu0-0}
\theta:= \inf \biggl\{
\frac{\int_X d_X(z,x)^p d\nu(x)- d_X(\mathcal{B}(\nu),z)^p}
{\int_Xd_X(x,\mathcal{B}(\nu))^pd\nu(x)}\; :\; z\in X \land \nu\in L_p(X)
\land \int_Xd_X(x,\mathcal{B}(\nu))^pd\nu(x)>0
\biggr\}.
\end{equation}
By Jensen's inequality and the $2$-barycentric property~\eqref{eq:2-barycentric},  
\[
\Bigl(\int_X d(z,x)^p d\nu(x)\Bigr)^{1/p} \geq 
\Bigl(\int_X d(z,x)^2 d\nu(x)\Bigr)^{1/2} \geq 
d(z,\mathcal{B}(\nu)),
\]
and therefore
$\theta\geq 0$. Our goal is to show 
$\theta\geq 1/(2^{p-1}-1)$.
Fix $\phi>\theta$.
Then there exists $z_0\in X$ and a Borel measure $\nu_0$
for which
\begin{equation}\label{eq:phi-nu0-1}
    \phi {\int_Xd_X(x,\mathcal{B}(\nu_0))^pd\nu_0(x)}
    >
 {\int_X d_X(z_0,x)^p d\nu_0(x)- d_X(\mathcal{B}(\nu_0),z_0)^p}   .
\end{equation}
For $x\in\mathrm{supp}(\nu_0)$, Apply~\eqref{eq:p-convexity} for $z=z_0$, $y=\mathcal{B}(\nu_0)$, and  $m=m_x$ the mid-point between $x$ and $\mathcal{B}(\mu)$
to get the point-wise estimate
\begin{equation}\label{eq:phi-nu0-2}
2 d_X(z_0,m_x)^p+ 2^{1-p}d_X(x,\mathcal{B}(\nu_0))^p \leq
d_X(z_0,x)^p+ d_X(z_0,\mathcal{B}(\nu_0))^p.
\end{equation}
Denote $m:X\to X$, $m(x)=m_x$, i.e., $m$ maps a point $x$ to the mid-point between $x$ and $\mathcal{B}(\nu_0)$.
Let $\nu_1=(\nu_0)_{\# m}$ the push-forward of $\nu_0$ under $m$.
By \autoref{lem:fixed-barycenter}, 
\begin{equation} \label{eq:B(nu1)=B(nu0)}
    \mathcal{B}(\nu_1)=\mathcal{B}(\nu_0).
\end{equation}
Hence,
\begin{align*}
        \phi& {\int_Xd_X(x,\mathcal{B}(\nu_0))^pd\nu_0(x)}
     \\ &\stackrel{\eqref{eq:phi-nu0-1}} >
 {\int_X d_X(z_0,x)^p d\nu_0(x)- d_X(\mathcal{B}(\nu_0),z_0)^p}  \\
 & \stackrel{\eqref{eq:phi-nu0-2}}\geq 
 2\biggl(\int_X d_X(z_0,m_x)^p d\nu_0(x)- d_X(z_0,\mathcal{B}(\nu_0))^p \biggr)
 +2^{1-p} \int_X d_X(x,\mathcal{B}(\nu_0))^p d\nu_0(x)
 \\ & \stackrel{\eqref{eq:B(nu1)=B(nu0)}}= 
  2\biggl(\int_X d_X(z_0,x)^p d\nu_1(x)- d_X(z_0,\mathcal{B}(\nu_1))^p \biggr)
 +2^{1-p} \int_X d_X(x,\mathcal{B}(\nu_0))^p d\nu_0(x)
 \\ & \stackrel{\eqref{eq:phi-nu0-0}}\geq 
 2\theta  \int_X d_X(x,\mathcal{B}(\nu_1))^p d\nu_1(x) 
  +2^{1-p} \int_X d_X(x,\mathcal{B}(\nu_0))^p d\nu_0(x)
  \\ &\stackrel{\eqref{eq:B(nu1)=B(nu0)}} =
  (2^{1-p}\theta+ 2^{1-p}) \int_X d_X(x,\mathcal{B}(\nu_0))^p d\nu_0(x).
\end{align*}
Thus, 
\[\phi> 2^{1-p}\theta+ 2^{1-p}.\]
Since the above holds for every $\phi>\theta$,
we conclude that $\theta\geq 1/(2^{p-1}-1)$, proving~\eqref{eq:p-barycentric}.
\end{proof}

\begin{remark}
The $p$-barycentric inequality
\begin{equation} \label{eq:p-bi-fails-0}
\int_X d_X(z,x)^pd\mu(x) \geq d(z,\mathcal B(\mu))^p +\e \int_X d_X(\mathcal B(\mu),x)^pd\mu(x), 
\end{equation}
fails for any fixed $\e\in(0,1]$ and $p\in (0,2)$ and any barycenter map even for $X=\mathbb R$.
In fact, fix $\e\in(0,1]$, and $p\in (0,2)$. 
Let $\mu=\tfrac12\delta_{-\eta}+\tfrac12\delta_{\eta}$, 
where 
$\eta=1$ when $p\in(0,1]$ and $\eta\ll \bigl( \frac{\e}{p(p-1)} \bigr)^{{1}/{(2-p)}}$ when $p\in(1,2)$.
Let $\alpha=\mathcal B (\mu)\in\mathbb R$. 
We will check~\eqref{eq:p-bi-fails-0} for $z\in\{-1,1\}$ only: 

\begin{multline*} \label{eq:p-bi-fails}
\max_{z\in\{\pm 1\}}\int_{\mathbb R} |z-x|^p d\mu(x)=\frac{(1+\eta)^p+(1-\eta)^p}{2} \\
\stackrel{(*)}< 1+\tfrac{\e}{2}\eta^p\leq (1+|\alpha|)^p+ \tfrac{\e}{2} (|\alpha-\eta|^p+|\alpha+\eta|^p)=
\max_{z\in\{\pm 1\}} |z-\alpha|^p +\e \int_{\mathbb R} |\alpha-x|^p d\mu(x).
\end{multline*}
Inequality ($*$) is elementary to check using the Taylor expansion of the left-hand side.
\end{remark}

\begin{remark}
The barycenter $\mathcal{B}(\mu)$ defined in \autoref{thm:barycenter} is also called the $2$-barycenter.
Naor and Silberman~\cite{NS11} defined $p$-barycenters similarly but with $p$ powers instead of $2$ powers.
The notion of $p$-barycenters was further studied in~\cite{Kuwae}.
The proof above of $p$-barycentric inequality used $2$-barycenter.
It is possible that a similar $p$-barycenteric inequality can be proved using $p$-barycenters, but we have not studied that. 
For the application of the $p$-barycentric inequality~\eqref{eq:p-barycentric} in this paper, 
any point $c(\mu)\in X$
that satisfies~\eqref{eq:p-barycentric}, or a similar inequality with a different constant replacing $\frac{1}{2^{p-1}-1}$, would suffice.
\end{remark}

\section{Counterexample to nonlinear Kwapie\'n formulations}
\label{sec:no-kwapien}

In this section, we prove \autoref{thm:no-Kwapien}.
The space utilized in its proof, $\SE$, is made up of $\frac12$ snowflakes of high-girth expander graphs.
On the one hand, it was proven in~\cite{ANN} that the $\frac12$ -snowflake of any finite metric space admits bi-Lipschitz embedding into some Aleksandrov space of nonnegative curvature.
On the other hand, it is observed in~\cite{Gromov-rand-groups,Kondo} that high-girth graphs
admit bi-Lipschitz embedding in some CAT$(0)$ space.
Using \autoref{thm:cat(0)-transform-closure}, we now know that so do their $\frac12$-snowflakes.
However, expander graphs are known to be poorly embeddable
into Hilbert spaces by almost any standard,
which we demonstrate for coarse embedding on average here.

To state the argument rigorously, we need a few standard graph-theoretic concepts that we will define now.
An (undirected) \emph{graph} $G=(V,E)$ is a pair:
A finite set $V$ whose elements are called vertices and a subset $E\subseteq \binom{V}{2}$ of unordered pairs of vertices, called edges.
A graph is called $d$-regular if every vertex has exactly $d$ edges adjacent to it.
A \emph{simple path} $P=(a_0,\ldots,a_{\ell})$
of length $\ell\ge 1$
is a sequence of  $\ell+ 1$
 distinct vertices $a_0,a_1,\ldots, a_\ell\in V$ such that every consecutive pair of vertices along the path constitutes an edge, i.e.,
$\{a_i,a_{i+1}\}\in E$ for every $i\in\{0,\ldots,\ell-1\}$.
A \emph{simple cycle} $C=(a_0,\ldots,a_{\ell-1})$ of
length $\ell\ge 3$ is a simple path of length $\ell-1$
in which $\{a_0,a_{\ell-1}\}\in E$ is also an edge.
The \emph{girth} of the graph $G$, denoted $\girth(G)$, is the length of the shortest simple cycle in $G$. For forests (i.e., graphs without simple cycles), the girth is defined as $\infty$.
The shortest path metric on the connected graph $G=(V,E)$, denoted by $d_G(u,v)$ for $u,v\in V$,
is defined as
the length of the shortest path between $u$ and $v$ (and $0$ if $u=v$).
The diameter of the graph $G=(V,E)$, is $\diam(G)=\diam((V,d_G))=\max_{u,v\in V} d_G(u,v)$.

Fix a $d$-regular graph $G=(V,E)$.
Define $\gr(G)=\lfloor \girth(G)/2\rfloor/\diam(G)$.
Let $\gamma(G)=1/(1-\lambda_2(G))$ be the $2$-Poincar\'e constant of $G$, where $\lambda_2(G)$ is the second largest eigenvalue of the Markov operator associated with the random walk in $G$ (that is, the adjacency
matrix of $G$ divided by $d$).
An $n$-vertex $d$-regular graph $G$ with $\gr(G)\ge \gr$ and $\gamma(G)\le \gamma$ is called a $(n,d,\gamma,\gr)$-graph.
An infinite family $\mathcal{G}$ of graphs is called a $(d,\gamma,\gr)$-family of graphs if every graph in $\mathcal{G}$
is an $(n,d,\gamma,\gr)$-graph for some $n\in\mathbb N$,
and for every $N$, there exist $n>N$ and an $(n,d,\gamma,\gr)$-graph in $\mathcal{G}$.
\emph{High-girth expander graphs}, is a $(d,\gamma,\gr)$-family of graphs for some fixed $d\in\{3,4,\ldots\}$,  $\gr\in(0,1)$, and $\gamma \in [1,\infty)$.

\begin{theorem} \label{thm:hi-girth-expander}
There exists a family of high-girth expander graphs.
Quantitatively, there exists a $(6,5,0.1)$-family of graphs which we denote by $\mathcal E$.
\end{theorem}

The first to explicitly state the existence of high-girth expander graphs that prove \autoref{thm:hi-girth-expander} were probably Margulis~\cite{Margulis} and Lubotzky, Phillips, and Sarnak~\cite{LPS}, from which we borrowed the quantitative bounds in \autoref{thm:hi-girth-expander}.
Their proofs were certainly the first explicit construction, but the existence of those graphs may have been known before, at least at some level.
See~\cite[\S~5.14]{Ostrovskii-book} for a different and very simple construction due to Margulis and Lubotzky, and
see~\cite[\S~5.17]{Ostrovskii-book} for more on the history of those graphs.

We need a generalization of Matou\v{s}ek's extrapolation
proved in~\cite[Lemma~52]{Naor-avg-john}.
\begin{proposition}[\cite{Naor-avg-john}]
\label{prop:naor-extrapolation}
Fix a regular graph $G=(V,E)$ that satisfies the $2$-Poincar\'e inequality with constant $\gamma>0$, and $p,q\ge 1$.
Then for every $f:V\to \ell_p$,
\begin{equation} \label{eq:naor-extrapolation}
\biggl(\frac{1}{|V|^2}\sum_{u,v\in V} \|f(u)-f(v)\|_p^q \biggr)^{\frac1q} \lesssim \bigl({\gamma}(p^2+q^2)\bigr)^{\frac{1}{\min\{p,2\}}}
\cdot \biggl( \frac{1}{|E|}\sum_{\{u,v\}\in E} \|f(u)-f(v)\|_p^{\max\{p,q\}} \biggr)^{\frac{1}{\max\{p,q\}}}.
\end{equation}
\end{proposition}

Let $B_G(u,t)=\{v\in V: d_G(u,v)\le t\}$ be the (closed) ball around the vertex $u\in V$ in the graph $G=(V,E)$.
The following proposition follows from the simple fact that in $n$-vertex $d$-regular graph $G=(V,E)$,
$|B_G(u,t)|\le 1+d \sum_{i=0}^{t-1} (d-1)^i$, for every $u\in V$.
\begin{proposition} \label{prop:log_d-diameter}
For $\ell\le \log_{d-1}\bigl(\frac{n+5}{3}\bigr)$,
$|B_G(u,\ell)|\le n/2$ for every vertex $u\in V$
in $n$-vertex $d$-regular graph $G=(V,E)$.
\end{proposition}

\begin{proof}[Proof of \autoref{thm:no-Kwapien}]
Fix $\mathcal{E}$ to be the high-girth expander graphs from \autoref{thm:hi-girth-expander}.
As observed by Gromov~\cite{Gromov-rand-groups} and Kondo~\cite{Kondo},
high-girth graphs admit bi-Lipschitz embedding in some CAT$(0)$ space.
Quantitatively, by~\cite[Lemma~3.9]{MN-expanders2}, an $(n,d,\cdot,\gr)$-graph embeds in some CAT$(0)$ space
with distortion at most $\pi/(4\gr)$.
Therefore, each $(G,d_G)\in \mathcal E$ embeds in some CAT(0) with a distortion at most $2.5\pi<8$.

Fix $G\in\mathcal E$.
On the one hand,
as proved in~\cite{ANN},
the $\tfrac12$-snowflake of any finite metric space (and in particular, $(V(G),\sqrt{d_G})$)
embeds in $W_2(\mathbb R^3)$
with distortion%
\footnote{The distortion can be improved to $1+\varepsilon$, for any $\varepsilon>0$.}
at most~$2$.
On the other hand, applying \autoref{thm:cat(0)-transform-closure} to the $\frac12$ -snowflake of the image of the biLipschitz embeddings of $(V(G),{d_G})$ in CAT$(0)$,
we conclude that $(V(G),\sqrt{{d}_G})$ embeds in
some CAT$(0)$ space $Y_G$ with distortion at most $4\cdot\sqrt{8} < 12$.
Let
 \[\sqrt{\mathcal{E}} =\Bigl(\biguplus \Bigl\{ \bigl( \bigl(V(G),\sqrt{\smash[b]{d_G}}\bigr), x_G\bigr):\; {G\in\mathcal E} \Bigr \}\Bigr)_2,\]
where $x_G\in V(G)$ is an arbitrary vertex in $G$.
By \autoref{cor:CAT(0)-union},
$\sqrt{\mathcal{E}}$ embeds with a distortion less than $12$ in the CAT$(0)$ space $\bigl(\prod\{Y_G:\; G\in \mathcal{E}\}\bigr)_2$,
and with distortion at most 2 in some Aleksandrov space of nonnegative curvature (a countable orthogonal product of $W_2(\mathbb R^3)$).

We next prove that $\SE$ is poorly embedded in Hilbert space.
Assume toward a contradiction the existence of a coarse embedding on the average of $\SE$ in a Hilbert space $\mathcal H$, and let $\omega,\rho,\Omega:[0,\infty) \to [0,\infty)$ be as stated in \autoref{def:coarse-average}.
Observe that by replacing $\Omega(t)$ with $\Omega(t)+1$  we may assume without loss of generality that $\Omega(1)>0$.
Fix $G=(V,E)\in \mathcal E$ to be an $n$-vertex graph for some   sufficiently large $n$ that satisfies
\begin{equation} \label{eq:alpha>beta-coarse}
\rho \Bigl(\tfrac12 \omega \Bigl (\bigl (\log_{5}\tfrac{n+5}{3} \bigr)^{1/2} \Bigr)\Bigr)>
 21\cdot  \Omega(1).
\end{equation}
This is possible since the limit on the left-hand side of~\eqref{eq:alpha>beta-coarse} is $\infty$ when $n\to \infty$.

Let $\mu$ be the uniform distribution over $V(G)$, the vertices of $G$, in $\SE$.
Let $f_\mu:V\to \mathcal H$ be an embedding of $(V,d_G^{1/2})$ in $\mathcal H$ satisfying~\eqref{eq:coarse-ub} and~\eqref{eq:coarse-avg-lb}.
First, by \autoref{prop:log_d-diameter}, the monotonicity of $\rho$ and $\omega$, and \eqref{eq:alpha>beta-coarse},
\begin{equation}\label{eq:avg-diam-lb}
\rho \Bigl(\frac{1}{n^2}\sum_{u,v\in V}\omega\bigl (d_G(u,v)^{1/2}\bigr) \Bigr)
\ge  \rho\Bigl( \tfrac12 \omega \Bigl (\bigl (\log_{5}\tfrac{n+5}{3} \bigr)^{1/2} \Bigr) \Bigr)\stackrel{\eqref{eq:alpha>beta-coarse}}{>}
21 \cdot  \Omega(1).
\end{equation}
On the other hand, applying~\eqref{eq:naor-extrapolation} with $q=1$ and $p=2$,
\begin{equation} \label{eq:avg-diam-ub}
\begin{aligned}
\rho \Bigl(\frac{1}{n^2}\sum_{u,v\in V}\omega\bigl (d_G(u,v)^{1/2}\bigr)  \Bigr)
& \stackrel{\eqref{eq:coarse-avg-lb}}{\le}
\frac{1}{n^2}\sum_{u,v\in V}\bigl \|f_\mu(u)-f_\mu(v)\bigr\|_{\mathcal H} \\
& \stackrel{\eqref{eq:naor-extrapolation}}{\leq}
20 \Bigl(\frac{1}{|E|}\sum_{\{u,v\}\in E}
\|f_\mu(u)-f_\mu(v)\|_{\mathcal H}^{2}\Bigr)^{ 1/{2}}\\
& \stackrel{\eqref{eq:coarse-ub}}\le
20 \cdot \Omega(1).
\end{aligned}
\end{equation}
Inequalities~\eqref{eq:avg-diam-lb} and~\eqref{eq:avg-diam-ub} contradict each other and therefore embedding $f_\mu$
does not exist.
\end{proof}

\section{Matou\v{s}ek extrapolation for Hadamard spaces}
\label{sec:mat-extrapolation}


In this section we prove \autoref{thm:extrapolation-cat0} --- Matou\v{s}ek extrapolation for Hadamard spaces.
We begin with the case $0<p<q$.

\begin{proposition} \label{prop:extrapolation-down}
Fix $q\in(0,\infty)$, a metric transform
$\varphi:[0,\infty)\to [0,\infty)$, and a Hadamard space $(X,d_X)$.
Then for any regular graph $G$,
\[
\gamma_q(G,(X,\varphi\circ d_X)) \le 4^q \gamma_q(G,\mathrm{Hadamard}).
\]
In particular, for $0<p\le q$,
\[
\gamma_p(G,X) \le 4^q \gamma_q(G,\mathrm{Hadamard}).
\]
\end{proposition}
\begin{proof}
Fix a Hadamard space $(X,d_X)$.
Apply \autoref{thm:cat(0)-transform-closure} on $(X,\varphi\circ d_X)$.
We conclude that $(X,\varphi\circ d_X)$ admits a bi-Lipschitz embedding in some Hadamard space with a distortion smaller than $4$.
It therefore satisfies the $q$-Poincar\'e inequality for Hadamard spaces with constant multiplied by $4^q$.
The bound on $\gamma_p(G,X)$ is achieved from the above bound using the metric transform $\varphi(t)=t^{p/q}$.
\end{proof}

Next, we consider the case $p\ge q\ge 1$.

\begin{proposition} \label{prop:extrapolation-2-upward}
Fix\/ $\infty>p\geq q\geq 1$.
Then for every regular graph $G$,
\[
\gamma_p(G,\mathrm{Hadamard}) \le C_{p,q} \gamma_q(G,\mathrm{Hadamard})^{p/q}.
\]
\end{proposition}

The rest of the section is therefore devoted to the proof of \autoref{prop:extrapolation-2-upward} using Mazur maps for tangent cones and
following the proof of Matou\v{s}ek extrapolation for Banach spaces in~\cite{Cheng,Laat-Salle} (see also the exposition in \cite[Section~3]{Esk22}).

The Mazur map $\psi=\psi_{p,q}:L_p(\mathbb R)\to L_q(\mathbb R)$,
for fixed $1\le p,q<\infty$,
is defined as
\begin{equation} \label{eq:mazur-map-reals}
(\psi_{p,q}(f))(\omega)= \|f\|_p^{1-p/q} |f(\omega)|^{p/q} \sgn(f(\omega)).
\end{equation}
It is easy to check that $\psi_{p,q}^{-1}= \psi_{q,p}$.

\begin{theorem}[See {\cite[Proposition~9.2]{BL}}]
Fix $1\le p,q<\infty$. There exists $c=c_{p/q}>0$ such that
for every $f_1,f_2\in L_p$,
if\/ $\|f_i\|_p\le R$, then
\begin{gather}
\label{eq:real-mazur-0}
\|\psi_{p,q}(f_i)\|_q=\|f_i\|_p, \text{\quad and}\\
\label{eq:real-mazur}
\| \psi_{p,q}(f_1)-\psi_{p,q}(f_2)\|_q \le
\begin{cases}
\frac{p}{q} \|f_1-f_2\|_p & p\ge q \\
c R^{1-p/q} \|f_1-f_2\|_p^{p/q} & p<q.
\end{cases}
\end{gather}
\end{theorem}

We next extend Mazur maps to $L_p$ spaces over
Euclidean cones,
in the spirit of the extension of Mazur maps to $L_p$ spaces
over arbitrary Banach spaces~\cite{Cheng,Laat-Salle}.
Fix a probability space $(\Omega,\mu)$. Let $X$ be an arbitrary metric space and $Y=\cone(X)$. Recall the notation from \autoref{def:cone}:
For $(s,x)\in Y$: $|(s,x)|=s\geq0$, and $\arg(s,x)=x\in X$.
Define $L_p(\Omega,\mu,Y)$ as the metric space
of functions $f:\Omega\to Y$ satisfying
\[
\int_\Omega |f(\omega)|^p \dd\mu(\omega) <\infty,
\]
with the distance
\[d_{L_p(\Omega,\mu,Y)}(f_1,f_2)=\biggl(\int_\Omega d_Y(f_1(\omega),f_2(\omega))^p\dd \mu(\omega)\biggr)^{1/p}.\]
We also use the norm notation
\[
\|f\|_p=d_{L_p(\Omega,\mu,Y)}(0,f)=\|(|f|)\|_{L_p(\Omega,\mu,\mathbb{R})}.
\]
When there is no confusion, we may also denote $L_p(\Omega,\mu,Y)$ by $L_p(\Omega,Y)$ or even $L_p(Y)$.

We define the Mazur map for Euclidean cones
$\psi_{p,q}:L_p(Y) \to L_q(Y)$ as follows.
Fix $f\in L_p(\Omega,Y)$. Denote $|f|\in L_p(\Omega,\mathbb R)$ as $|f|(\omega)=|f(\omega)|$ and $\arg(f):\Omega \to X$ as
$\arg(f)(\omega)=\arg(f(\omega))$. Define
\[
\psi_{p,q}(f)=\bigl (\psi_{p,q}(|f|), \arg(f)\bigr)\in L_q(Y),
\]
where ``$\psi_{p,q}$" on the right-hand side of the equation is the classical Mazur map~\eqref{eq:mazur-map-reals} on $L_p(\Omega,\mathbb R)$.
It follows directly from the classical case that
$\psi_{p,q}^{-1}=\psi_{q,p}$.

\begin{theorem}[Mazur map for cones]
\label{thm:mazur-cone}
Let $Y=\cone(X)$, and
fix $1\le p,q<\infty$.
For every $f\in L_p(\Omega,\mu,Y)$,
\begin{equation}
\label{eq:mazur-bound-0}
\|\psi_{p,q}(f)\|_q=\|f\|_p,
\end{equation}
and every $f_1,f_2\in L_p(\Omega,\mu,Y)$,
if\/ $\|(|f_i|)\|_p \le R$, then
\begin{equation}
\label{eq:mazur-bound}
d_{L_q(Y)}\bigl (\psi_{p,q}(f_1),\psi_{p,q}(f_2)\bigr) \le
\begin{cases}
\bigl (\tfrac pq+1 \bigr ) d_{L_p(Y)}(f_1,f_2) & p\ge q \\
\bigl (c+2) R^{1-p/q} d_{L_p(Y)}(f_1,f_2)^{p/q} & p<q,
\end{cases}
\end{equation}
where $c=c_{p/q}>0$ depends only on $p/q$.
\end{theorem}

Before proving \autoref{thm:mazur-cone}, we need the following lemma.

\begin{lemma}
For every $s,t\ge 0$, $x,y\in X$, and $Y=\cone(X)$
\begin{equation} \label{eq:cone-bound}
d_Y((s,x),(s,y)) \le {2} \cdot d_Y((s,x),(t,y)).
\end{equation}
\end{lemma}
\begin{proof}
For brevity, denote $c=\cos(\min\{\pi, d_X(x,y)\})$.
To minimize $d_Y((s,x),(t,y))^2$ as a function of $t$, we differentiate
\[
\frac{\dd }{\dd t} d_Y((s,x),(t,y))^2
\stackrel{\eqref{eq:1-cos identity}}=2(t-s) +2s(1-c)=2t -2sc.
\]
Solving for $0$ and recalling that $t\ge 0$, we have $t_{\min}=\max\{sc,0\}$.
Thus,
\begin{multline*}
4 d_Y((s,x),(t,y))^2\ge 4d_Y((s,x),(t_{\min},y))^2
\stackrel{\eqref{eq:def cone}}= 4(s^2- s^2c\max\{0,c\})
\\ =2s^2\cdot 2(1-c\max\{0,c\})
\stackrel{(*)}\ge 2s^2(1-c)
\stackrel{\eqref{eq:1-cos identity}}= d_Y((s,x),(s,y))^2.
\end{multline*}
Inequality~(*) is deduced by analyzing the cases
$c\in(0,1]$ and $c\in [-1,0]$ separately.
\end{proof}

\begin{proof}[Proof of \autoref{thm:mazur-cone}]
From~\eqref{eq:1-cos identity}, we have $\bigl |\,|f_1(\omega)|-|f_2(\omega)|\, \bigr| \le d_Y(f_1(\omega),f_2(\omega)) $, and therefore
\begin{equation} \label{eq:power<=distance}
\bigl\| |f_1|-|f_2| \bigr \|_p \le d_{L_p(Y)}(f_1,f_2).
\end{equation}
We will also need the bound
\begin{equation} \label{eq:eq-height}
d_Y((1,x),(1,y)) \le 2,
\end{equation}
which follows directly from~\eqref{eq:1-cos identity}.

Define $g_i=\psi_{p,q}(f_i)$.
By the triangle inequality in $Y$
\begin{equation*}
d_Y(g_1(\omega),g_2(\omega))\le
\bigl |\psi_{p,q}(|f_1|)(\omega) - \psi_{p,q}(|f_2|)(\omega)\bigr |
 + d_Y\bigl(\, \bigl(\psi_{p,q}(|f_2|)(\omega)  ,\arg(f_1(\omega))\bigr)\,,\,g_2(\omega)\bigr).
\end{equation*}
Taking the $L_q$ norm over the above inequality, and using the triangle inequality in $L_q$ (Minkowski inequality), we have
\begin{equation} \label{eq:one-bound}
d_{L_q(Y)}(g_1,g_2) \le \bigl \| \psi_{p,q}(|f_1|) - \psi_{p,q}(|f_2|)\bigr \|_{q}+
 d_{L_q(Y)}\bigl((\psi_{p,q}(|f_2|),\arg(f_1)),g_2\bigr )
\end{equation}

We begin bounding the first term on the right-hand side of~\eqref{eq:one-bound}.
Using~\eqref{eq:real-mazur} and~\eqref{eq:power<=distance},  it is bounded from above by
$\frac pq\, d_{L_p(Y)}(f_1,f_2)$ when $p\geq q$,
and by $c R^{1-p/q} d_{L_p(Y)}(f_1,f_2)^{p/q}$ when
$p<q$.

The second term on the right-hand side of~\eqref{eq:one-bound} is bounded
as follows.
\begin{equation} \label{eq:1a}
\begin{aligned}
 d_{L_q(\Omega,\mu,Y)}
 \!\!\!\!\!\!\!\!&\,\,\,\,\,\,\,\,
 \bigl((\psi_{p,q}(|f_2|),\arg(f_1)),g_2\bigr )
\\ & =
\|\psi_{p,q}(|f_2|)\|_q
\biggl(\int_\Omega d_{Y} ((1,\arg(f_1)(\omega)),(1,\arg(f_2)(\omega)))^q
\frac{\psi_{p,q}(|f_2|)(\omega)^q}
{\|\psi_{p,q}(|f_2|)\|_q^q } d\mu(\omega) \biggr)^{1/q}
\\ & \stackrel{\substack{\eqref{eq:mazur-map-reals}\\ \eqref{eq:real-mazur-0}}}{=}
\|f_2\|_p
\biggl(\int_\Omega d_{Y} ((1,\arg(f_1)(\omega)),(1,\arg(f_2)(\omega)))^q
{\frac{|f_2(\omega)|^p}{\|f_2\|_p^{p}}}
 d\mu(\omega) \biggr)^{1/q}
 \\ & =\|f_2\|_p d_{L_q(\Omega,\nu,Y)}((1,\arg(f_1)),(1,\arg(f_2))),
\end{aligned}
\end{equation}
where $\nu$ is the probability measure
    $\dd \nu= \|(|f_2|)\|_p^{-p} |f_2|^p \dd\mu$.
When $q\le p$, we continue~\eqref{eq:1a} using the bound $\|h\|_{L_q(\Omega,\nu)}\le \|h\|_{L_p(\Omega,\nu)}$:

\begin{equation} \label{eq:1b}
\begin{aligned}
  \|f_2\|_p \,& d_{L_q(\Omega,\nu,Y)}((1,\arg(f_1)),(1,\arg(f_2)))
  \\ &\leq   \|f_2\|_p\, d_{L_p(\Omega,\nu,Y)}((1,\arg(f_1)),(1,\arg(f_2)))
\\ & =
\Bigl(\int_\Omega d_{Y} ((|f_2(\omega)|,\arg(f_1)(\omega)),(|f_2(\omega)|,\arg(f_2)(\omega)))^p
d\mu(\omega) \Bigr)^{1/p}
\\ & \stackrel{\eqref{eq:cone-bound}}\le
{2}\cdot \Bigl(\int_\Omega d_{Y} ((|f_1(\omega)|,\arg(f_1)(\omega)),(|f_2(\omega)|,\arg(f_2)(\omega)))^p
d\mu(\omega) \Bigr)^{1/p}
\\ & = {2}\cdot d_{L_p(\Omega,\mu,Y)}(f_1,f_2).
\end{aligned}
\end{equation}
We conclude that
\(
d_{L_q(Y)}(g_1,g_2) \le \bigl (\tfrac pq +2\bigr )d_{L_p(Y)}(f_2,f_1)
\) when $q\le p$.

When $p<q$ we continue~\eqref{eq:1a} using the bound
$\|h\|_{L_q(\Omega,\nu)}\le \|h\|_{L_p(\Omega,\nu)}^{p/q} \|h\|_{L_\infty(\Omega,\nu)}^{1-p/q}$
as follows:

\begin{align*}
  \|f_2\|_p\, & d_{L_q(\Omega,\nu,Y)}((1,\arg(f_1)),(1,\arg(f_2)))
  \\ &\leq  \Bigl(\|f_2\|_p\,  d_{L_p(\Omega,\nu,Y)}((1,\arg(f_1)),(1,\arg(f_2)))\Bigr)^{p/q}
 \\ & \qquad\cdot  \Bigl(\|f_2\|_p\,
  d_{L_\infty(\Omega,\nu,Y)}((1,\arg(f_1)),(1,\arg(f_2)))\Bigr)^{1-p/q}
    \\ &\leq   (2R)^{1-p/q} \Bigl(\|f_2\|_p\,  d_{L_p(\Omega,\nu,Y)}((1,\arg(f_1)),(1,\arg(f_2)))\Bigr)^{p/q}
\\ & \stackrel{\eqref{eq:1b}}{\leq}
(2R)^{1-p/q} \bigl(2d_{L_p(\Omega,\mu,Y)}(f_1,f_2)\bigr)^{p/q}
\\& = 2 R^{1-p/q}d_{L_p(\Omega,\mu,Y)}(f_1,f_2)^{p/q}.
\end{align*}
We conclude that
\(
d_{L_q(Y)}(g_1,g_2) \le  (c +2) R^{1-p/q} d_{L_p(Y)}(f_2,f_1)^{p/q}
\)
when $p< q$.
\end{proof}

We are now in a position to use the Mazur map to deduce extrapolation.

\begin{proof}[Proof of \autoref{prop:extrapolation-2-upward}]
Fix a $d$-regular graph $G=(V,E)$ on $n$ vertices satisfying the $q$-Poincar\'e inequality \eqref{eq:p-poincare-metric} with constant $\gamma=\gamma_q(G,\mathrm{Hadamard})>0$ with respect to any Hadamard space.
Fix also a Hadamard space $X$ and a function $f:V\to X$.
By rescaling the metric $d_X$ if necessary, assume that
\begin{equation} \label{eq:constraint-for-extrapolation}
\frac{1}{|V|^2} \sum_{u,v\in V} d_X(f(u),f(v))^p =1.
\end{equation}

Denote by $b=\mathcal{B}(f(V))$, the barycenter of the uniform distribution over $f(V)$ in $X$.
Let $Y=T_bX$ be the tangent cone of $X$ at $b$.
$Y$ is a Euclidean cone and a CAT(0) space (see~\autoref{def:tangent}).
Denote by $0$ the cusp of $Y$.
Since $p\geq1$, for any fixed point $z\in X$, we have
\begin{equation} \label{eq:jensen-cat0}
\frac{1}{|V|}\sum_{v\in V} d_X(z,f(v))^p \geq \bigg( \frac{1}{|V|} \sum_{v\in V} d_X(z,f(v))\bigg)^{p} \geq d_X(z,b)^p
\end{equation}
by Jensen's inequality and the $1$-Lipschitz property of the barycenter mapping in $X$ (see~\cite[Lem. 4.2]{Lang-PS} or \cite[Thm. 6.3]{Sturm-NPC}). Therefore, \eqref{eq:constraint-for-extrapolation}, together with~\eqref{eq:jensen-cat0}, yield
\begin{equation} \label{eq:norm-smaller-1}
\frac{1}{|V|} \sum_{u\in V} d_X(f(u),b)^p \leq \frac{1}{|V|^2} \sum_{u,v\in V} d_X(f(u),f(v))^p = 1.
\end{equation}
Now, recall the definition of the logarithmic map $\Log_bf: V\to Y$ and consider the mapping $g:V\to Y$ given by $g = \psi_{p,q} (\Log_bf)$.
 By the definitions of the Mazur map, the tangent cone and the logarithmic map,  \eqref{eq:norm-smaller-1} gives
\begin{equation} \label{eq:mazur-constraints}
\|g\|_q \stackrel{\eqref{eq:mazur-bound-0}}= \| \Log_b f \|_p = \bigg( \frac{1}{|V|} \sum_{u\in V} d_Y(\Log_bf(u),0)^p \bigg)^{1/p} = \bigg( \frac{1}{|V|} \sum_{u\in V} d_X(f(u),b)^p \bigg)^{1/p} \leq 1.
\end{equation}
Conversely,  by the triangle inequality we also have
\begin{equation}
1=\frac{1}{|V|^2} \sum_{u,v\in V} d_X(f(u),f(v))^p \leq \frac{2^p}{|V|} \sum_{u\in V} d_X(f(u),b)^p = \frac{2^p}{|V|} \sum_{u\in V} d_Y(\Log_bf(u),0)^p.
\end{equation}
Applying precisely the same reasoning as in \eqref{eq:norm-smaller-1}  for the function $\Log_bf$ which takes values in the CAT(0) space $Y$ and has barycenter 0 by  \autoref{prop:cusp-barycenter}, we therefore deduce that
\begin{align*}
1&\leq \frac{2^p}{|V|^2} \sum_{u,v\in V} d_Y( \Log_bf(u),\Log_bf(v))^p
\\ &= \frac{2^p}{|V|^2} \sum_{u,v\in V} d_Y (\psi_{q,p}g(u),\psi_{q,p}g(v))^p \displaybreak[0]
\\ &
\stackrel{\eqref{eq:mazur-bound}, \eqref{eq:mazur-constraints}}{\lesssim_{p,q}} \frac{1}{|V|^2} \sum_{u,v\in V} d_Y(g(u),g(v))^q \displaybreak[0]
\\ & \leq \frac{\gamma}{|E|} \sum_{(u,v)\in E} d_Y(g(u),g(v))^q
\\ & = \frac{\gamma}{|E|} \sum_{(u,v)\in E} d_Y(\psi_{p,q}\Log_bf(u),\psi_{p,q}\Log_bf(v))^q
\\ &\stackrel{\eqref{eq:mazur-bound}}{\lesssim_p} \gamma \bigg( \frac{1}{|E|}\sum_{(u,v)\in E} d_Y(\Log_bf(u),\Log_bf(v))^p \bigg)^{q/p} \displaybreak[0]
\\ & \stackrel{\eqref{eq:busemann}}{\leq} \gamma \bigg( \frac{1}{|E|}\sum_{(u,v)\in E} d_X(f(u),f(v))^p \bigg)^{q/p}
\\ &  \le \gamma^{p/q}\frac{1}{|E|}\sum_{(u,v)\in E} d_X(f(u),f(v))^p,
\end{align*}
where the last inequality follows by raising the penultimate expression to the power of $p/q\ge 1$ and using the first inequality to argue that it is at least~$1$.
Taking an infimum of the right-hand side over all CAT(0) spaces $X$ and all functions $f:V\to X$ satisfying \eqref{eq:constraint-for-extrapolation}, we obtain
$\gamma_p(G,\mathrm{Hadamard})\lesssim_{p,q} \gamma^{p/q}$,
which completes the proof.
\end{proof}

\section{Expanders with respect to random regular graphs}
\label{sec:extrapolation-rrg}

In this section we prove \autoref{thm:MN-rrg-expander-p} --- the existence of an expander with respect to random regular graphs that satisfies the $p$-Poincar\'e
inequalities for any $p\in(0,\infty)$.

\subsection{High-level approach}

The proof contains three adaptations to the proof of
\autoref{thm:MN-rrg-expander-2}:
\begin{itemize}
\item 
Working with any metric transform of random regular graphs using the same tools used in earlier parts of the paper to handle metric transforms of CAT(0) spaces.
\item 
The arguments of~\cite{MN-expanders2} all make use of squared distances. They will be extended to work with any power $q\geq 2$ of the distances.
\item 
Adding a ``diagonalization zigzag step"
from~\cite[Lemma~4.3]{MN-superexpanders} to the zigzag iteration to weave the different expanders for different powers into one expander.
\end{itemize}

In this section we give high level description of the proof
in~\cite{MN-expanders2} and explain the first (and most novel) adaptation listed above.
We will use the following stronger variant of the Poincar\'e inequality:

\begin{definition}[Non-bipartite Poincar\'e inequality with respect to $X$]\label{def:bi-tension}
Fix $p\in(0,\infty)$ and a metric space $(X,d_X)$.
For a regular graph $G=(V,E)$
we define by $\gamma^+_p(G,X)$ the infimal $\gamma^+$
such that for every $f,g: V\to X$,
\begin{equation} \label{eq:abs-p-poincare-metric}
\frac{1}{|V|^2}\sum_{u,v\in V} d_X(f(u),g(v))^p \le \frac{\gamma^+}{|E|} \sum_{(u,v)\in E} d_X(f(u),g(v))^p.
\end{equation}
\end{definition}

We already mentioned that $\gamma_2(G,\mathbb R)=1/(1-\lambda_2(G))$.
Analogously, $\gamma_2^+(G,\mathbb R)=1/(1-\lambda(G))$,
where $\lambda(G)=\max\{\lambda_2(G),-\lambda_n(G)\}$
is the second normalized eigenvalue of $G$, \emph{in absolute value}.
Compared with \autoref{def:X-PI} to \autoref{def:bi-tension},  it is also plainly obvious that $\gamma_p(\mathcal E,\X)\le \gamma_p^+(\mathcal E,\X)$ for any $\mathcal E$, $\X$ and $p> 0$. For more on $\gamma^+_p$ and its relation to $\gamma_p$, see~\cite{MN-superexpanders}.

In~\cite{MN-expanders2} the following stronger version of \autoref{thm:MN-rrg-expander-2} is actually proved at
the end of the proof of~\cite[Theorem~4.1]{MN-expanders2}.%
\footnote{
The corresponding inequality in~\cite{MN-expanders2}
appears immediately after Inequality~(85) there.
Unfortunately, due to an editing error in~\cite{MN-expanders2}
the inequality there is not written as a uniform bound over the random regular graph $H$ (with high probability).
A similar problem appears in~\cite{MN-expanders2} in formulae (60), (61),
the two displayed formulae after Lemma~3.16,
(84), (85), and the last displayed formula in the proof of Theorem~4.1.}
Fix $d\ge 3$.
There exists $\mathcal E=\{G_n\}$, a family of 3-regular expander graphs,
and $\Gamma\in(0,\infty)$ with the following property.
Fix $m\ge 3$.   Let $\RRG\sim \mathcal G_{m,d}$ be a random $d$-regular graph on $m$ vertices.
Then there exists $\sigma_\RRG>0$ such that with probability $1-o_m(1)$,
\begin{equation} \label{eq:poincare-cone-MN-expanders}
\sup_n \gamma^+_2\Bigl(G_n, {\cone\bigl(\Sigma(\RRG),\sigma_\RRG d_{\Sigma(\RRG)}\bigr)}\Bigr) <\Gamma\ \land\ \sigma_\RRG \asymp \frac{1}{\log_d m}.
\end{equation}
$\Sigma(\RRG)$ is the one-dimensional simplicial complex defined by $\RRG$,
and $d_{\Sigma(\RRG)}$ is the metric on $\Sigma(\RRG)$ induced by the one-dimensional simplicial complex (see \autoref{sec:transform-cone}.)
For the purpose of proving \autoref{thm:MN-rrg-expander-2} the fact that $\sigma_\RRG\asymp\frac{1}{\log_dm}$ has the property that with probability $1-o_m(1)$, $\diam(\Sigma(\RRG))\approx \diam(\RRG)\asymp 1/\sigma_\RRG$.
Hence, $\bigl(\Sigma(\RRG), d_{\Sigma(\RRG)}\bigr)$
admits a bi-Lipschitz embedding in $\cone\bigl(\Sigma(\RRG),\sigma_\RRG d_{\Sigma(\RRG)}\bigr)$
and therefore \eqref{eq:poincare-cone-MN-expanders} implies
$\gamma^+_2(\mathcal E,\RRG)\lesssim \Gamma$ with probability $1-o_m(1)$.
For our purposes, the crucial observation is that~\eqref{eq:poincare-cone-MN-expanders} actually holds  for
$\cone\bigl(\Sigma(\RRG),\tau d_{\Sigma(\RRG)}\bigr)$   for any $\tau\ge \sigma_\RRG$ and not just for $\tau=\sigma_\RRG$.
Furthermore, this statement can be made for any power $q\geq 2$ of the distances and not just $q=2$.

\begin{theorem} \label{thm:poincare-cone}
There exists a sequence of 3-regular graphs of increasing size $\mathcal E$ with the following property:
For every $q\in\{2,3,\infty)$
There exists
$\Gamma_q>0$ such that for any $d\ge 3$,
\begin{equation}
\label{eq:poincare-cone}
\lim_{m\to \infty}\Pr_{\RRG\sim \mathcal G_{m,d}} \biggl[\sup_{\tau\ge\sigma_{\RRG}}  \gamma_q^+ \Bigl (\mathcal E, {\cone\bigl(\Sigma(\RRG),\tau d_{\Sigma(\RRG)}\bigr)}\Bigr)<\Gamma_q\ \land\ \sigma_{\RRG}\asymp \tfrac{1}{\log_dm}\biggr]=1.
\end{equation}
\end{theorem}

This generalization of~\eqref{eq:poincare-cone-MN-expanders} is important here
because it implies that $\mathcal E$ is also an expander with respect
to $(\RRG,\min \{d_\RRG,\alpha\})$, for any $\diam(\RRG)\gtrsim \alpha>0$, 
i.e., with respect to any truncation of $(\RRG,d_\RRG)$ and with any power.
Using \autoref{prop:transform-in-producr-cones},
we will next deduce that
$\mathcal E$ is also an expander with respect
to any metric transform of $d_\RRG$ that satisfies the 
$q$-Poincar\'e inequality for any power $q\in\{2,3,\ldots\}$.

\begin{proof}[Proof of \autoref{thm:MN-rrg-expander-p} based on \autoref{thm:poincare-cone}]
By
\autoref{prop:transform-in-producr-cones},
$(\RRG,\varphi\circ d_{\RRG})$ embeds in a $\ell_q$ product
\[
\biggl(\alpha_0 \RRG \times \Bigl(\prod_{i\in\mathbb N} \bigl(\alpha_i  \cone(\beta_i \RRG)\bigr)\Bigr)_q\biggr)_q
\]  
with distortion at most $4$.
Since $\RRG$ has a finite diameter, we may assume without loss of generality that $\alpha_0=0$ and $\beta_i\ge 2\pi/\diam(\Sigma(\RRG))$.
Since $2\pi/\diam(\Sigma(\RRG))\asymp\sigma_\RRG$, we can
also assume that $\beta_i\ge \sigma_\RRG$.
The cost of this assumption is an increase in the distortion of the embedding by a factor of at most
\[
F=\max\{\sigma_\RRG \diam(\Sigma(\RRG)),1\}\lesssim 1.
\]

The structure of the Poincar\'e inequality for expanders indicates that for any regular graph $G$, any $q\in[1,\infty)$ and any series of metric spaces $(Y_i)_i$,
\begin{equation} \label{eq:gamma-prod-q}    
\gamma^+_q\Bigl(G,\Bigl(\prod_i Y_i\Bigr)_q\,\Bigr)= \sup_i \gamma^+_q(G,Y_i).
\end{equation}
Therefore, with probability $1-o_m(1)$,
\[
\gamma^+_q(\mathcal E,(\RRG,\varphi\circ d_{\RRG}))\le 4^q
\gamma^+_q\Bigl(\mathcal E,\Bigl(\prod_i \alpha_i \cone(\beta_i\RRG)\Bigr)_q\,\Bigr)
\leq (4F)^q \sup_{i\in \mathbb N} \gamma^+_q(\mathcal E, \alpha_i \cone(\beta_i \RRG))\stackrel{\eqref{eq:poincare-cone}}\le (4F)^q \Gamma_q,
\]
where $\mathcal E$ is the expander graph constructed in \autoref{thm:poincare-cone}.
\end{proof}

To gain some understanding of the proof of \autoref{thm:poincare-cone},
we next explain informally what $\sigma_\RRG$ in~\eqref{eq:poincare-cone-MN-expanders} is and why it is reasonable that~\eqref{eq:poincare-cone-MN-expanders} can be extended to~\eqref{eq:poincare-cone}.
It is known that random regular graph $\RRG\sim\mathcal G_{m,d}$ is ``close'' to a high-girth graph (with probability $1-o_m(1)$).
More precisely, there is a small subset of edges of $\RRG$
that after removing them,
the resulting graph $\RRG'$ has a girth which is asymptotically equivalent to its diameter. We define $\sigma_\RRG=2\pi/\girth(\RRG')$.
As explained in~\cite{MN-expanders2}, this means that the metric
$(\Sigma(\RRG'),\sigma_\RRG d_{\Sigma(\RRG')})$ admits bi-Lipschitz embedding in a CAT(1) metric, which allows us to apply the machinery developed in \cites{ALNRRV}{MN-superexpanders}{MN-barycentric}{MN-expanders2}
to construct expander $\mathcal E$ with respect to the CAT(0) space $\cone (\Sigma(\RRG'),\sigma_\RRG d_{\Sigma(\RRG')})$.
Furthermore, as explained in~\cite{MN-expanders2},
the set difference  $\cone \bigl(\Sigma(\RRG)\setminus \Sigma(\RRG'),\sigma_\RRG d_{\Sigma(\RRG)}\bigr)$  is a $\cone(L_1)$ metric.
Furthermore, it is known that any expander is also an $L_1$-expander and, as proved in~\cite{MN-expanders2}, this property extends to $\cone(L_1)$.
The two Poincar\'e inequalities ``can be merged'' into a Poincar\'e inequality for the union of the cones.
That is, $\mathcal E$ is also an expander with respect to
$\cone (\Sigma(\RRG),\sigma_\RRG d_{\Sigma(\RRG)})$.

Now consider the metric $\cone (\Sigma(\RRG),\tau d_{\Sigma(\RRG)})$, for arbitrary $\tau\ge\sigma_\RRG$.
Using the notation above,
$(\Sigma(\RRG'), \tau d_{\Sigma(\RRG)})$ is a magnification of $(\Sigma(\RRG'), \sigma_\RRG d_{\Sigma(\RRG)})$ when $\tau>\sigma_\RRG$
and therefore it is also CAT$(1)$.
The same machinery will prove that $\mathcal E$ is also an expander with respect to
$\cone (\Sigma(\RRG),\tau d_{\Sigma(\RRG)})$, for any $\tau \ge \sigma_\RRG$.
This generalization of \cite{MN-expanders2}
is rather straightforward.
Unfortunately, we do not see how to convincingly reduce the whole problem to the $\sigma_\RRG$ factor of~\cite{MN-expanders2}.
Instead, we generalize some definitions and reduce parts of the argument to lemmas from~\cite{MN-expanders2},
obtaining a relatively short proof with many references to lemmas from~\cite{MN-expanders2}.

\subsection{Zigzag expanders}    
We next describe the construction of an expander with respect to a metric while abstracting the metric properties needed for it to work.
The construction is a variant of the zigzag iterative procedure of Reingold, Vadhan, and Wigderson~\cite{RVW}.
It was extended to work with respect to more general metric spaces in~\cite{MN-superexpanders,MN-barycentric,MN-expanders2}. The construction here is yet another variant based on the constructions in~\cites{MN-superexpanders,MN-expanders2}. 
We begin by describing the relevant graph operations and their properties and then describe the iterative construction. 

\subsubsection{Graph operations}\label{sec:products}

Fix $d_1,d_2\in \N$. 
Let $G_1$ be a $d_1$-regular graph and let
$G_2$ be a $d_2$-regular graph. Suppose that $|V_{G_2}|=d_1$. Then
one can construct a new graph, called the {\em zigzag product} of
$G_1$ and $G_2$ and denoted 
$G_1\oz G_2$. 
This construction is due
to Reingold, Vadhan and Wigderson~\cite{RVW}. We will not need to
recall the definition of $G_1\oz G_2$ here:
All that we will use below is that $G_1\oz G_2$ is $d_2^2$-regular,  $|V_{G_1\oz
G_2}|=|V_{G_1}|\cdot|V_{G_2}|$, and for every metric space $(X,d_X)$ and any $q\in(0,\infty)$
we have
\begin{equation}\label{eq:zigzag sub multiplicativity}
\gamma_q^+\!\left(G_1\oz G_2,X\right)\le \gamma_q^+\!\left(G_1,X\right)\cdot \gamma_q^+\!\left(G_2,X\right)^2.
\end{equation}
Inequality~\eqref{eq:zigzag sub multiplicativity} in this abstract setting is due
to~\cite{MN-superexpanders}, generalizing a slightly better bound when $q=2$ and $X=L_2$.

Fix $d\in \N$ and suppose that $G$ is a $d$-regular graph. For every
integer $D\ge d$ one can define a new graph called the $D$-{\em edge
completion} of $G$, and denoted $\mathcal{C}_D(G)$.
See~\cite[Def.~2.8]{MN-superexpanders} for the definition of
$\mathcal{C}_D(G)$. 
All we will use below is that $\mathcal{C}_D(G)$ is
$D$-regular, $V_{\mathcal{C}_D(G)}=V_G$, and for every metric space $(X,d_X)$ and every $q\in[1,\infty)$ we have
\begin{equation}\label{eq:completion gamma+}
\gamma_q^+\!\left(\mathcal{C}_D(G),X\right)\le 2\gamma_q^+\!\left(G,X\right).
\end{equation}
The proof of~\eqref{eq:completion gamma+} 
is contained in~\cite[Lem.~2.9]{MN-superexpanders}.

We will also work in the following with {\em Ces\`aro averages} of graphs. Given $d,m\in \N$ and a $d$-regular graph $G$, its $m$th Ces\`aro average $\A_m(G)$ is a new graph defined by $V_{\A_m(G)}=V_G$ and
$$
\forall\, (u,v)\in V_G\times V_G,\qquad E_{\A_m(G)}(u,v)\eqdef \sum_{t=0}^{m-1} d^{m-1-t} \left(A_G^t\right)_{u,v},
$$
where we recall that $A_G$ is a the normalized adjacency matrix of $G$. One checks that $\A_m(G)$ is $md^{m-1}$-regular and that the adjacency matrix of $\A_m(G)$ is the corresponding Ces\`aro average of $A_G$, i.e.,
$$
A_{\A_m(G)}=\frac{1}{m}\sum_{t=0}^{m-1} A_G^t.
$$
The Ces\'aro average is used during the construction of the expander
to decrease the 
Poincar\'e constant $\gamma^+$ in certain cases. 
In particular we will say that a metric space $X$ 
satisfies \emph{the spectral calculus condition}%
\footnote{As observed in~\cite[§~9.2]{MN-superexpanders}, 
for the construction of zigzag expanders one could be satisfied with a weaker condition of \emph{uniform decay} of the Poincar\'e constant.}
with power $q$ 
if there exists $K=K_{X,q}>0$ such that 
for every regular graph $G$ and every $m\in\mathbb N$, we have
\begin{equation}\label{eq:spectral calculus condition}
\gamma_q^+\!\left(\mathcal A_m(G),X\right)\le K\max\left\{1,\frac{\gamma_q^+(G,X)}{m}\right\}.
\end{equation}

Hilbert spaces (and, in particular, $\mathbb R$) are easily seen to satisfy~\eqref{eq:spectral calculus condition} with power $q=2$ using the spectral decomposition of symmetric matrices. It was generalized to uniformly convex Banach spaces in~\cite{MN-superexpanders} and to CAT(0) spaces (with power $q=2$) in~\cite{MN-barycentric}.
The proof for CAT(0) spaces in~\cite{MN-barycentric} uses
the $2$-barycentric property of CAT(0)~\eqref{eq:2-barycentric}.
Using the $q$-barycentric property of CAT(0) space for $q\geq 2$ that is stated in \autoref{prop:p-barycentric-CAT0}, together with~\cite[Theorem~1.8]{MN-barycentric},
we conclude:
\begin{lemma} \label{lem:q-speactral-calculus}
    For any $q\in[2,\infty)$ there exists $K_q>0$, such that~\eqref{eq:spectral calculus condition} holds with constant $K_q$ for any CAT(0) space $X$.
\end{lemma}

It is also a straightforward observation that if the metric spaces $X_1,X_2,\ldots$ satisfy~\eqref{eq:spectral calculus condition} with power $q$ and constant $K$, then so does $\bigl(\prod_i X_i\bigr )_q$.

The following simple lemma is proved, for example, as part of the proof of~\cite[Lemma~3.1]{MN-expanders2}.
\begin{lemma}\label{lem:from MN plus replacement}
Fix two integers $d,n\ge 3$ and let $G$ be a $d$-regular graph with
$|V_G|=n$. 
Then there exists a $3$-regular graph $G^*$, 
determined uniquely and deterministically by $G$, satisfying
$|V_{G^*}|=36dn=18d|V_G|$ such that for every metric space
$(X,d_X)$,
\begin{equation}\label{eq:d^4 goal}
\gamma_q^+\left(G^*,X\right)\lesssim d^{2q}\gamma_q^+\left(G,X\right).
\end{equation}
\end{lemma}

We will require a Matou\v{s}ek's extrapolation-type lemma 
for $\gamma^+$.
\autoref{prop:naor-extrapolation} implies that
for any $q\geq p\geq 1$ and any regular graph $G$,
\begin{equation}\label{eq:naor-extrapolation-simplified}
    \gamma_q(G,L_p) \lesssim \bigl(\gamma_2(G,\mathbb R)(p^2+q^2) \bigr)^{\frac{q}{\min\{p,2\}}}.
\end{equation}
A similar bound holds for $\gamma^+$:
\begin{lemma} \label{lem:gamma+-extrapolation}
For any regular graph $G$ and $q\geq p\geq 1$,
\begin{equation} \label{eq:gamma+-extrapolation}
    \gamma^+_q(G,L_p) \lesssim \bigl(\gamma^+_2(G,\mathbb R)(p^2+q^2) \bigr)^{\frac{q}{\min\{p,2\}}}.
\end{equation}
\end{lemma}
\begin{proof}
 In \cite[Lemma~2.3]{MN-superexpanders} 
 it is stated that for any regular graph $G$ and any metric space $X$, the double cover of $G$, $H=G\otimes e$ 
 (where $e$ is the graph consisting of one edge and $\otimes$ is the graphical tensor product), satisfies
 \[\frac{2}{2^q+1} \gamma_q (H,X) \leq \gamma_q^+(G,X)\leq 2\gamma_q(H,X).\]
 Furthermore, by the spectral decomposition theorem, when $X=\mathbb R$ and $q=2$ we have :
 \[ \gamma_2^+(G,\mathbb R)=\frac{1}{1-\lambda(G)}=\frac{1}{1-\lambda_2(H)}=\gamma_2(H,\mathbb R).\]
 Applying the above to~\eqref{eq:naor-extrapolation-simplified},
 \[\gamma_q^+(G,L_p)\leq 2 \gamma_q(H,L_p)
 \lesssim \bigl (\gamma_2(H,\mathbb R)(p^2+q^2) \bigr)^{\frac{q}{\min\{p,2\}}}
 = \bigl(\gamma^+_2(G,\mathbb R)(p^2+q^2) \bigr)^{\frac{q}{\min\{p,2\}}} . \qedhere \]
\end{proof}

\subsubsection{Abstract zigzag iteration}

The following  is the basic zigzag iteration which is similar to \cite[Theorem~3.2]{MN-expanders2} (which in turn is a variant of a similar iteration from~\cite{RVW}). 
Compared to~\cite{MN-expanders2}, the main difference here is that 
the bound on the Poincar\'e constant depends only on the underlying metric space, but not on the base graph. 
This will be important next when we construct one expander family of graphs that works for all powers of the metric. 
The cost is a higher (but still constant) degree of the expander: the number of vertices of the base graph instead of the base graph's degree squared.
This achieved by simply 
externalizing the graphs after the Ces\'aro average instead of after the zigzag product. 

\begin{lemma}\label{lem:zigzag simple} 
Fix $q_0\in(0,\infty)$ and $K_0\in [1,\infty)$ and two integers $n_0,d_0\ge 3$ with $n_0\ge d_0^3$. Suppose that $(X,d_X)$ is a metric space satisfying~\eqref{eq:spectral calculus condition}. 
Suppose further that there exists a $d_0$-regular graph $H$ with $|V_H|=n_0$ such that
\begin{equation}\label{eq:base graph assumption}
\gamma_{q_0}^+\!\left(H,d_X\right)\le \sqrt{\frac{m_0}{2K_0}},
\quad \text{where} \quad
m_0\leq \left\lfloor\frac{\log n_0}{3\log d_0}\right\rfloor.
\end{equation}

Then there exists a sequence of\/ $n_0$-regular graphs
$\{G_j\}_{j=1}^\infty$  satisfying
\begin{equation}\label{eq:Gj cardinality}
\forall\, j\in \N,\qquad \left|V_{G_j}\right|=n_0^j,
\end{equation}
and
\begin{equation}\label{eq:Gj gamma+ vanilla}
\forall\, j\in \mathbb N,\qquad \gamma_{q_0}^+\!\left(G_j,d_X\right)\leq 2 K_0.
\end{equation}
Furthermore, $\{G_j\}_j$
depends uniquely and deterministically only on $H$ and $m_0$.
\end{lemma}
\begin{proof}
Define $G_1$ to be the complete graph with self-loops  on $n_0$ vertices of degree $n_0$.
And indeed $\gamma_q^+(G_1)=1\leq 2K_0$.
Supposing that $G_j$ has been defined and satisfies~\eqref{eq:Gj cardinality} and~\eqref{eq:Gj gamma+ vanilla}. 
Let 
\begin{equation}\label{eq:Gj+1}
G_{j+1}\eqdef \mathcal{C}_{n_0}\bigl(\A_{m_0}(G_j\oz H)\bigr).
\end{equation}

The zigzag product $G_j\oz H$ is a regular graph with $n_0^j\cdot n_0=n_0^{j+1}$ vertices and degree $d_0^2$.
The Ces\`aro average
$\A_{m_0}(G_j\oz H)$ is therefore with the same number of vertices but degree $m_0d_0^{2(m_0-1)}$.
Recalling~\eqref{eq:base graph assumption}, we have
$m_0d_0^{2(m_0-1)}\le n_0$. We can therefore form the edge completion
$G_{j+1}=\mathcal{C}_{n_0}\bigl(\A_{m_0}(G_j\oz H)\bigr)$, which is an $n_0$-regular
graph, with $|V_{G_{j+1}}|=n_0^{j+1}$.
Bounding the Poincar\'e constant:

\begin{align*}
    \gamma_{q_0}^+(G_{j+1},X) &
    \stackrel{\eqref{eq:completion gamma+}}\leq  2 \gamma_{q_0}^+(\A_{m_0}(G_j\oz H),X)
\\ & \stackrel{\eqref{eq:spectral calculus condition}} \leq 2 K_0 \max\biggl\{ 
1,\frac{\gamma_{q_0}^+(G_j\oz H,X)}{m_0}
\biggr\}
\\ & \stackrel{\eqref{eq:zigzag sub multiplicativity}}\leq 2 K_0 \max\biggl\{ 
1,\frac{\gamma_{q_0}^+(G_j,X)\gamma_q^+(H,X)^2}{m_0}
\biggr\}
\\ & \stackrel{\substack{\eqref{eq:Gj gamma+ vanilla}\\ \eqref{eq:base graph assumption}}}\leq 2 K_0 \max\biggl\{ 
1,\frac{2K_0\cdot (m_0/2K_0)}{m_0}
\biggr\}
\\ & = 2 K_0 . 
\qedhere
\end{align*}
\end{proof}

The above iteration constructs for each power $q$ a different expander.
In order to extract one expander with respect to all powers of the metric we use an argument of~\cite{MN-superexpanders}. Specifically, the following lemma.

\begin{lemma}[{Main zigzag iteration \cite[Lemma~4.3]{MN-superexpanders}}]\label{lem:zigzag main iteration}
Let $\{d_k\}_{k=1}^\infty$ be a sequence of integers and for each $k\in \N$ let $\{n_j(k)\}_{j=1}^\infty$ be a strictly increasing sequence of integers. For every $j,k\in \N$ let $F_j(k)$ be a regular graph of degree $d_k$ with $n_j(k)$ vertices. Suppose that $\K$ is a family of kernels such that
\begin{equation}\label{eq:K finiteness assumption}
\forall\, K\in \K,\ \forall\, j,k\in \N,\quad \gamma_+(F_j(k),K)<\infty.
\end{equation}
Suppose also that the following two conditions hold true.
\begin{itemize}
\item For every $K\in \K$ there exists $k_1(K)\in \N$ such that
\begin{equation}\label{eq:k1}
\sup_{\substack{j,k\in \N\\k\ge k_1(K)}}\gamma_+(F_j(k),K)\le k_1(K).
\end{equation}
\item For every $K\in \K$ there exists $k_2(K)\in \N$ such that every regular graph $G$ satisfies the following spectral calculus inequality.
\begin{equation}\label{eq:k2}
\forall\, t\in \N,\quad \gamma_+\left(\A_t(G),K\right)\le k_2(K)\max\left\{1,\frac{\gamma_+(G,K)}{t^{1/k_2(K)}}\right\}.
\end{equation}
\end{itemize}
Then there exists $d\in \N$ and a sequence of $d$-regular graphs $\{H_i\}_{i=1}^\infty$ with
$$
\lim_{i\to \infty} |V(H_i)|=\infty
$$
and
\begin{equation}\label{eq:good for all K}
\forall\, K\in \K,\quad \sup_{j\in \N}\gamma_+(H_j,K)<\infty.
\end{equation}

Furthermore the family of graphs $\{H_i\}_i$ is determined deterministically by the family of graphs 
$\{(F_j(k))\}_{j,k}$.
\end{lemma}
\autoref{lem:zigzag main iteration} is a direct quote from~\cite{MN-superexpanders}, except the additional sentence in the end, that is evident when following the proof of Lemma~4.3 in~\cite{MN-superexpanders}.
``Kernels" in the above lemma include powers of metrics.

The construction of expnder with respect to a space $\mathcal Y$ in~\cite{MN-expanders2} requires a partition of $\mathcal Y$ into a ``large" part and a``small" part. 
To make the presentation more modular here we 
``axiomatize" the partition as follows.

\begin{definition} \label{def:size-partition}
 A metric space (or a class of metric spaces) $\mathcal Y$ has an abstract partition according to size if   
for every $s\in\mathbb N$ there exist 
$\mathcal Y_{<s}, \mathcal Y_{\geq s}\subset \mathcal Y$ such that
\begin{compactitem}
\item    
$\mathcal Y_{<s}$ is monotonous non-decreasing in $s$
(i.e., $\mathcal Y_{<s}\subseteq \mathcal Y_{<t}$ whenever
$s\leq t$.) 
$\mathcal Y_{\geq s}$ is monotonous non-increasing in $s$.
\item 
$c_2(\mathcal Y_{<s})\leq g(s)$,
where $g:N\to [1,\infty)$ is some universally fixed function (for our application, one can take $g(s)= C\log (s+1)$.)
\item 
\( \gamma_q^+(G,\mathcal Y) \leq 
h(q,\gamma_q^+(G,\mathcal Y_{<s}), \gamma_q^+(G,\mathcal Y_{\geq s}))
\)
for any regular graph $G$, $s\in\mathbb N$ and $q\in[1,\infty)$, for some universally fixed function $h$ (in our application one can take $h(q,x,y)=2^{q-1}\max\{x,y\}$.)
\end{compactitem}
\end{definition}

\begin{lemma} \label{lem:bootstraping-main-zigzag-iteration}
Fix a class of metric spaces $\mathcal Y$ having an abstract partition according to \autoref{def:size-partition} with the following properties.

\begin{itemize}
    \item 
$\mathcal{Y}$ satisfies the spectral calculus condition~\eqref{eq:spectral calculus condition} with power $q$ and constant $K_q$,  $\forall q\in\{2,3,\ldots\}$.

\item 
There exist non-decreasing function
$f:[1,\infty)\times [1,\infty)\to [1,\infty)$ 
and non-decreasing series
$(s_n)_{n\in\mathbb N}$ 
such that for any $n\in\mathbb N$ 
and any graph $G$ on $n$ vertices,
\begin{equation} \label{eq:gamma_q^+(G,Y)-bound}
    \gamma_q^+(G,\mathcal{Y}_{\geq s_n})\leq f(q,\gamma_2^+(G,\mathbb R)).
\end{equation}
\end{itemize}

Then:
\begin{itemize}
\item There exists
a $3$-regular infinite family of graphs  $\mathcal E$ 
such that
\begin{equation} \label{eq:gamma_q(E,Y)<infty}
    \forall q\in\{2,3,\ldots\},\quad  \gamma_q^+(\mathcal E,\mathcal Y)<\infty.
\end{equation}
\item 
For every $q\in\{2,3,\ldots\}$ there exists
a $3$-regular infinite family of graphs  $\mathcal E(q)=(G_n(q))_n$ 
such that
\begin{equation} \label{eq:E_q bound}
    \forall q\in\{2,3,\ldots\},\quad  \gamma_q^+(\mathcal E_q,\mathcal Y)<\infty,
\end{equation}
and $|V_{G_j(q)}|=a_q\cdot b_q^j$, for some $a_q\in\mathbb N$, $b_q\in\{2,3\ldots\}$.
\end{itemize}
Furthermore, $\mathcal E$ and $\mathcal E_q$ are determined uniquely and deterministically 
by the sequences $(K_q)_q$, $(s_n)_n$ as well as the function $f(\cdot,\cdot)$.
\end{lemma}

\begin{remark}
The construction of $\mathcal E$ in~\eqref{eq:gamma_q(E,Y)<infty} above in this abstract setting only guarantees that the sizes of the graphs in
$\mathcal E$ increase \emph{at least} exponentially in the index. For the algorithmic applications of the construction in \autoref{sec:UA-proofs} we need the increase to be \emph{at most} exponential. This is the purpose of the expanders $\mathcal E_q$ from~\eqref{eq:E_q bound} in the conclusion of the lemma.
\end{remark}

\begin{proof}[Proof of \autoref{lem:bootstraping-main-zigzag-iteration}]
The proof uses \autoref{lem:zigzag main iteration}. To this end,
we begin with constructing for every $q\in\{2,3,\ldots\}$ a sequence of graphs $F_n(q)$ that satisfies~\eqref{eq:k1}  (\eqref{eq:K finiteness assumption} is automatically satisfied when the kernel is a power of metric and the graphs are not bipartite).
Our starting point is some universally fixed classical $3$-regular non-bipartite expander 
    $(T_{n})_{n\in 2\mathbb N}$ such that $|T_n|=n$ and $\sup_n\gamma_2^+(T_n,\mathbb R)= c_2<\infty$.
Let $n_2=28$ the minimal even number 
such that $n_2\geq 3^3$ and
\[ c_2 \leq \sqrt{m_2/(2K_2)}, \quad \text{where} \quad 
m_2=\biggl \lfloor \frac{\log n_2}{3\log 3}\biggr\rfloor.
\]
Let $H_2=T_{n_2}$.
Both~\eqref{eq:spectral calculus condition} and~\eqref{eq:base graph assumption} hold
with $q_0=2$, $H=H_2$, $K_0=K_2$, $m_0=m_2$, $n_0=n_2$ and $X=\mathcal Y_{\geq s_{n_2}}$. 
So by 
\autoref{lem:zigzag simple}, there exists 
a family $(F_2(n))_{n\in\mathbb N}$ of $n_2$-regular graphs of increasing size such that
\[ \sup_{n\in\mathbb N} \gamma_2^+(F_n(2),\mathcal Y_{\geq s_{n_2}})\leq 2K_2.\]

Now suppose that  $(n_2,\ldots,n_{q-1})$ and 
$(F_n(p))_{n\in\mathbb N,\ p\in\{2,\ldots,q-1\}}$
are already defined. Further, by the inductive hypothesis, for every $p\in\{2,\ldots,q-1\}$, $(F_n(p))_n$
is a family of $n_p$-regular graphs of increasing size and
\begin{equation}
    \sup_{\substack{n\in\mathbb N\\ r\in\{p,\ldots,q-1\}}} \gamma_p^+(F_n(r), \mathcal Y_{\geq s_{n_p}})\leq 2 K_p.
\end{equation}

Denote $c_{q}=\sup_n f_{q}(\gamma_2^+(F_n(q-1),\mathbb R)$. Since $(F_n(q-1))_n$ is a classical expander, $c_{q}<\infty$.
Let $n_{q}$ be the minimal natural number such that 
$n_{q}\geq n_{q-1}$, 
$\frac{\log n_q}{3\log n_{q-1}}\geq m_{q-1}$ 
and
\[ \max\{c_{q},2K_{q-1}\} \leq \sqrt{m_{q}/(2K_{q})}, \quad \text{where} \quad 
m_{q}=\biggl \lfloor \frac{\log n_{q}}{3\log n_{q-1}}\biggr\rfloor.
\]

Let $H_{q}=F_{n_{q}}(q-1)$.
Both~\eqref{eq:spectral calculus condition} and~\eqref{eq:base graph assumption} hold
with $q_0=q$, $H=H_{q}$, $K_0=K_{q}$, $m_0=m_q$, $n_0=n_{q}$ and $X=\mathcal Y_{\geq s_{n_q}}$. 
So, by \autoref{lem:zigzag simple}, there exists 
a family of $n_{q}$-regular graphs $(F_{n}(q))_n$
that depends only on $H_q$ and $m_q$
such that $|V_{F_n(q)}|=(n_q)^n$ and
\[ \sup_{n\in\mathbb N} \gamma_{q}^+(F_n(q),\mathcal Y_{\geq s_{n_q}})\leq 2K_q.\]

Before continuing, observe that applying
\autoref{lem:from MN plus replacement}
on $(F_n(q))_n$ one gets the sequence $\mathcal E_q$ that satisfies~\eqref{eq:E_q bound}.

We now continue with the construction of $\mathcal E$.
Observe also that for any $p\in\{2,\ldots, q-1\}$
 both~\eqref{eq:spectral calculus condition} and~\eqref{eq:base graph assumption} also hold 
 for $q_0=p$, $H=H_q$, $K_0=K_p$, $m_0=m_q$, $n_0=n_q$, and
 $X=\mathcal Y_{\geq s_{n_p}}$.
 So by \autoref{lem:zigzag simple} the same sequence of graphs $(F_n(q))_n$ (which depends only on $H=H_q$ and $m_0=m_q$) also satisfies
 \[ \sup_n\gamma_p^+(F_n(q),\mathcal Y_{\geq s_{n_p}})\leq 2K_p.\]

Thus, by induction,~\eqref{eq:k1} holds for the sequence of expanders $(F_n(q))_{n,q}$.
Observe that \eqref{eq:spectral calculus condition}
implies~\eqref{eq:k2}. 
Hence, the conditions of \autoref{lem:zigzag main iteration} are met.
Furthermore, observe that by induction, the graphs $H$ and parameter $m_0$ in all applications of \autoref{lem:zigzag simple} are determined deterministically by only the sequences $(K_q)_q$, $(s_n)_n$ and the function $f(\cdot,\cdot)$
(as well as the family $(T_n)_n$ but this family is considered fixed in advance). 
Hence, by \autoref{lem:zigzag simple} so is the family $(F_n(q))_{n,q}$.

Therefore by \autoref{lem:zigzag main iteration} there exists a $d$-regular expander $\mathcal F$ for some constant $d$, uniquely determined by 
 $(K_q)_q$, $(s_n)_n$ and $f(\cdot)$,
such that for every $q\in\{2,3,\ldots,\}$,
\[ \gamma_q^+(\mathcal F,\mathcal Y_{\geq s_{n_q}})<\infty .  \]
Since $\mathcal F$ is a classical expander and by \autoref{def:size-partition} and Matou\v{s}ek's extrapolation~\eqref{eq:gamma+-extrapolation} for $\gamma^+$, 
\[ \gamma_q^+(\mathcal F,\mathcal Y_{< s_{n_q}})
\leq  \bigl (g(s_{n_q})^2\gamma_2^+(\mathcal F,\ell_2)\cdot(q^2+2^2)\bigr)^{q/2}<\infty.
\]
Again, by \autoref{def:size-partition} the above bounds can be merged, and so $\gamma_q^+(\mathcal F,\mathcal Y)<\infty$ for any $q\in\{2,3,\ldots\}$.
Lastly, by \autoref{lem:from MN plus replacement},
$\mathcal F$ can be converted deterministically 
into a family $\mathcal E$ of $3$-regular graphs such that for any $q\in\{2,3,\ldots,\}$,
$\gamma_q^+(\mathcal E,\mathcal Y)\lesssim d^{2q} \gamma_q(\mathcal F,\mathcal Y)<\infty$.
\end{proof}

\subsection{Proof of \autoref{thm:poincare-cone}}

Fix a universal constant $K\ge 1$ that will be used throughout the proof as a sufficiently large constant. .
Also, fix $\delta\ge 1$.
We define  $\tau=\delta\sigma_{\RRG}$.

\begin{definition}\label{def:F}
Let $\F$ be the family of all connected graphs $G$ that satisfy
\begin{compactenum}
\item \( {\diam(G)}{}\le K \girth(G),\) and
\item For every $S\subseteq \Sigma(G)$ with $|S|\le \sqrt{|V_G|}$ we have
\begin{equation}\label{eq:F-embed-L1}
c_1\bigl(\cone\bigl(S, \tfrac{2\pi}{\girth(G)}\cdot d_{\Sigma(G)}\bigr)\bigr)\le K,
\end{equation}
where $c_1(\cdot)=c_{L_1}(\cdot)$ is the distortion of embedding into $L_1$.
\end{compactenum}
\end{definition}

Define $U_\F$ to be the formal disjoint union of the one-dimensional simplicial complexes $\{\Sigma(G)\}_{G\in\F}$.
Define a metric $d_\F:U_\F\times U_\F\to[0,\infty)$ on $U_\F$ as follows.
\[
d_\F(x,y)=\begin{cases}
\frac{2\pi d_{\Sigma(G)}(x,y)}{\girth(G)} &
\text{if } \exists G\in\F \text{ s.t. } x,y\in\Sigma(G),\\
2\pi (K+1) & \text{otherwise}.
\end{cases}
\]
As proved in~\cite[\S~3.5]{MN-expanders2}, $d_\F$ is a CAT(1) metric.
By the same reasoning, $\delta d_{\F}$ is also a CAT(1) metric when $\delta\ge 1$.%
\footnote{Note that this statement is not necessarily true when $\delta<1$.}
Define
\begin{equation}\label{eq:def:X-delta}
\X_\delta= \cone(U_\F, \delta d_{\F}).
\end{equation}

The following is a small generalization of~\cite[Lemma~3.9]{MN-expanders2} from $\delta=1$ to $\delta\geq 1$.
\begin{lemma} \label{lem:X_delta-hadamard}
 $\X_\delta$ is Hadamard space, and for any $G\in \F$, and
\begin{equation}
c_{\X_\delta}(\Sigma(G),\min\{d_{\Sigma(G)},\diam(\Sigma(G))/\delta\}) \le {\pi (K+1)}.
\end{equation}
\end{lemma}
\begin{proof}
As mentioned above, $(U_\F,\delta d_\F)$ is a complete CAT(1) space, so by \autoref{thm:berestovskii}, $\X_\delta$ is a Hadamard space.
The metric $\min\{d_{\Sigma(G)}, \diam(\Sigma(G))/\delta\}$ is equivalent to the metric
$ \min\bigl\{\frac{\delta 2\pi}{\girth(G)}d_{\Sigma(G)},
\frac{2\pi \diam(\Sigma(G))}{\girth(G)}\bigr\}$ up to scaling.
Bounding from below:
\[
\min\Bigl\{\frac{\delta 2\pi}{\girth(G)}d_{\Sigma(G)},
\frac{2\pi \diam(\Sigma(G))}{\girth(G)}\Bigr\}
\ge \min\Bigl\{\frac{\delta 2\pi}{\girth(G)}d_{\Sigma(G)}, \pi
\Bigr\} \ge \min\{\delta d_\F, \pi\}.
\]
Bounding from above:
\begin{equation*}
\min\Bigl\{\frac{\delta 2\pi}{\girth(G)}d_{\Sigma(G)},
\frac{2\pi \diam(\Sigma(G))}{\girth(G)}\Bigr\}
\le  \min \{\delta d_\F , 2\pi (K+1)\} \le 2(K+1) \min\{\delta d_\F,\pi\}.
\end{equation*}
Application of \autoref{prop:cone-truncation} on $\delta d_\F$, concludes the proof.
\end{proof}

In~\eqref{eq:F-embed-L1} we require a bi-Lipschitz embedding of the $\cone(S,d_\F)$ in $L_1$, where in fact here we will need a bi-Lipschitz embedding of $\cone(S,\delta d_{\F})$.
We chose the ``wrong'' definition to be compatible with~\cite{MN-expanders2} so we can quote (in the proof of \autoref{prop:rrg-structure}) the embedding property from~\cite[Lemma~3.12]{MN-expanders2} verbatim.
However, \autoref{cor:blow-out-cone-in-cone} shows that $\F$ in fact satisfies the ``right'' requirement:

The next proposition describes the structure of the Euclidean cone over scaled random regular graph as being made of two overlapping Euclidean cones:
One that is embeddable in $L_1$, and the other is over (approximate) high girth graph and therefore CAT(0).
The statement (except for \autoref{it:VI} below) is just a slight generalization of \cite[Proposition~3.15]{MN-expanders2}.

\begin{proposition} \label{prop:rrg-structure}
For every integer $d\ge 3$ there exists $C'(d)>0$ such that for every $n\in \mathbb N$, if\/ $\RRG\sim \mathcal{G}_{m,d}$ then with probability at least $1-C'(d)/\sqrt[3]{m}$ the graph $\RRG$
is connected and there exists $\sigma_\RRG\in(0,\infty)$ and two
subsets $A_1,A_2\subseteq \Sigma(\RRG)$ such that for every $\delta\ge 1$ the following assertions hold true.
\begin{compactenum}[(I)]
\item \label{it:I} $A_1\cup A_2=\Sigma(\RRG)$.
\item \label{it:II} $d_{\Sigma(\RRG)}(A_1\setminus A_2,A_2\setminus A_1)\gtrsim 1/\sigma_\RRG$.
\item \label{it:III}
$c_1(\cone(A_1,\delta\sigma_{\RRG} d_{\Sigma(\RRG)}))\leq K$.
\item \label{it:IV}
$c_{\X_\delta}(\cone(A_2,\delta\sigma_\RRG d_{\Sigma(\RRG)}))\leq K$,
where $\X_{\delta}$ is as
in~\eqref{eq:def:X-delta}.
\item \label{it:V}
$c_{\cone(\Sigma(\RRG),\delta \sigma_{\RRG} d_{\Sigma(\RRG)})}(\Sigma(\RRG),\min\{d_{\Sigma(\RRG)}, \diam(\Sigma(\RRG))/\delta\})\leq K.$
\item \label{it:VI}
$\frac{1}{\sigma_\RRG}\asymp \diam(\Sigma(\RRG))$.
\end{compactenum}
\end{proposition}

\begin{proof}
The statement of \cite[Proposition~3.15]{MN-expanders2}
is a special case of
\autoref{prop:rrg-structure} (without item~(\ref{it:VI})) when $\delta=1$.
Fortunately, we will not need to repeat the whole proof of
\cite[Proposition~3.15]{MN-expanders2}.
Instead, we will see that the same choice of
$A_1,A_2\subseteq \Sigma(\RRG)$ and $\sigma=\sigma_\RRG>0$  in the proof of \cite[Proposition~3.15]{MN-expanders2} which satisfies
items~(\ref{it:I})--(\ref{it:V}) for $\delta=1$, will also satisfy them and item~(\ref{it:VI}) for any $\delta\ge 1$.

We begin with item~(\ref{it:VI}):
As asserted in formula~(204) in~\cite{MN-expanders2},
$\sigma_\RRG\asymp 1/\log_dm$. But
with probability at least $1-C_d/{m^{1/3}}$,
$\diam(\Sigma(\RRG))\asymp \log_d m$.
Indeed, the lower bound holds for any graph in the support of $\mathcal G_{m,d}$, by a simple volume argument.
The upper is a well-known bound that holds with probability
at least $1-C_d/{m^{1/3}}$.
See for example formula~(197) in~\cite{MN-expanders2} where
a rough bound of $\diam(\RRG) \lesssim \log _d m$ is proved.%
\footnote{The actual statement in~\cite[formula~(197)]{MN-expanders2} is $\diam(L) \lesssim \log_d m$, where $L=(V_L,E_L)$ is a graph with $V_L=V_{\RRG}$ and $E_L\subset E_{\RRG}$ and therefore $\diam(\RRG)\le \diam(L)$.}
A precise estimate of $\diam(\RRG)$ is proved in~\cite{Bol-rrg-diameter}.

Item~(\ref{it:I}) and item~(\ref{it:II}) do not depend on $\delta$ and therefore are proven in~\cite{MN-expanders2} as is.

Item~(\ref{it:III}):
Since item~(\ref{it:III}) is already proved for $\delta=1$ in
\cite[Proposition~3.15]{MN-expanders2}, the extension for $\delta\ge 1$ is just an application of~\eqref{eq:blow-out-cone-in-cone}.

Item~(\ref{it:IV}):
Here we use \emph{the proof} of \autoref{it:IV} in
\cite[Proposition~3.15]{MN-expanders2}.
It is proved (in Inequality~(203)) there that
there is a canonical embedding
$\phi:(A_2,\sigma_\RRG d_{\Sigma(\RRG)}) \to (U_\F,d_\F)$ such that
\[
\sigma_{\RRG} d_{\Sigma(\RRG)}(x,y)\le d_\F(\phi(x),\phi(y)) \le 3
\sigma_{\RRG} d_{\Sigma(\RRG)}(x,y).
\]
Obviously, this is equivalent to
\[
\delta \sigma_{\RRG} d_{\Sigma(\RRG)}(x,y)\le \delta d_\F(\phi(x),\phi(y)) \le 3
\delta \sigma_{\RRG} d_{\Sigma(\RRG)}(x,y).
\]
Thus, by~\cite[Fact~3.5]{MN-expanders2},
\begin{equation}
c_{\X_\delta}(\cone(A_2, \delta\sigma_{\RRG} d_{\Sigma(X)}))  = c_{\cone(U_\F,\delta d_\F)}(\cone(A_2, \delta\sigma_{\RRG} d_{\Sigma(X)}))\leq 3.
\end{equation}

Item~(\ref{it:V}):
The metric $\min\{d_{\Sigma(\RRG)},\diam(\Sigma(\RRG))/\delta\}$
is isometric up to rescaling to
\begin{equation}\label{eq:some-metric}
\min\{\delta\sigma_{\RRG} d_{\Sigma(\RRG)}, \sigma_{\RRG}\diam(\Sigma(\RRG))\}.
\end{equation}
By item~(\ref{it:VI}), $\sigma_{\RRG}\diam(\Sigma(\RRG))\asymp 1$.
Therefore, the metric~\eqref{eq:some-metric} is bi-Lipschitz equivalent to the metric
$\min\{\delta\sigma_{\RRG} d_{\Sigma(\RRG)}, \pi\}$.
Therefore, by~\eqref{eq:truncation in cone},
\begin{multline*}
c_{\cone(\Sigma(\RRG),\delta \sigma_{\RRG} d_{\Sigma(\RRG)})}(\Sigma(\RRG),\min\{d_{\Sigma(\RRG)}, \diam(\Sigma(\RRG))/\delta\})
\\ \asymp
c_{\cone(\Sigma(\RRG),\delta \sigma_{\RRG} d_{\Sigma(\RRG)})}
(\Sigma(\RRG),\min\{\delta\sigma_{\RRG} d_{\Sigma(\RRG)}, \pi\})
\stackrel{\eqref{eq:truncation in cone}}\leq \pi/2.  
 \quad\qed
\end{multline*}
\renewcommand{\qedsymbol}{}
\end{proof}

\medskip

The next lemma combines the Poincar\'e inequality with respect to $\cone(A_2,\delta\sigma_{\RRG} d_{\Sigma(\RRG)})$ and
the Poincar\'e inequality with respect to
$\cone(A_1,\delta\sigma_{\RRG} d_{\Sigma(\RRG)})$,
to a Poincar\'e inequality with respect to their union, $\cone(\Sigma(\RRG),\delta\sigma_{\RRG} d_{\Sigma(\RRG)})$.
The statement of the following lemma is a slight generalization of
\cite[Lemma~3.16]{MN-expanders2} where it is stated only for $q=2$.

\begin{lemma} 
\label{lem:gamma-2-union}
Fix $q\in[1,\infty)$ and $\beta \in (0,\pi]$ and $n\in \mathbb N$.
Let $(X,d_X)$ be a metric space and suppose that $A,B\subseteq X$ satisfy $A\cup B=X$ and
\(
d_X(A\setminus B, B\setminus A)\ge \beta.
\)
Then,  every regular graph $G$  satisfies
\[
\gamma_q^+(G, \cone(X)) \lesssim \frac{\gamma_q^+(G,\cone(A))+\gamma_q^+(G,\cone(B))}{\beta^{2q}}.
\]
\end{lemma}
\begin{proof}[Outline of a proof]
The proof of \cite[Lemma~3.16]{MN-expanders2} uses
\cite[Lemma~5.8]{MN-expanders2}.
This lemma is extended in a straightforward way fro power $q=2$ to general powers $q\in[1,\infty)$.
In order to apply the generalization of \cite[Eq.~(129)]{MN-expanders2} to powers $q\in[1,\infty)$, it suffice to use the simple bound $\sqrt{a^2+b^2}\leq 2(a^q+b^q)^{1/q}$, for any $q\in[1,\infty)$ and therefore \cite[Eq.~(129)]{MN-expanders2}
implies a similar inequality where on the left-hand side there is an $\ell_q$ distance instead of $\ell_2$ and on the right hand side in the denominator, $\kappa$ is replaced by $2\kappa$.
\end{proof}

We also need the following simpler variant of \autoref{lem:gamma-2-union}, as stated (for $q=2$) in \cite[Lemma~4.2]{MN-expanders2}:
\begin{lemma} 
\label{lem:gamma-2-union-simple}
Let $X$ be a metric space and $q\in[1,\infty)$.
Suppose $\{B_i\}_{i\in I}$
are subsets of $X$ such that
\[ X=\bigcup_{I}B_i \quad \text{and}\quad i\ne j \Rightarrow d_X(B_i,B_j)\ge \pi.\]
Then, for any regular graph $G$,
\[
\gamma_q^+(G,\cone(X))\le 2^{q-1} \sup_I \gamma_q^+(G,\cone(B_i)).
\]
\end{lemma}
Again, the proof \autoref{lem:gamma-2-union-simple} is a straightforward extension of the proof of
\cite[Lemma~4.2]{MN-expanders2} to general $q\in[1,\infty)$.

The following lemma constructs an expander with respect to $\cone(A_2,\delta\sigma_\RRG d_{\Sigma(\RRG)})$, 
as defined in \autoref{prop:rrg-structure}.

\begin{lemma} \label{lem:X_delta-expander}
There exists a family of\/ $3$-regular expander graphs
$\mathcal E$ that for any $q\in\{2,3,\ldots\}$ there exists
$\Gamma_q>0$ such that for any
 $\delta\ge 1$, 
\(
\gamma_q^+(\mathcal E, \X_\delta)\lesssim \Gamma_q,
\)
where $\X_\delta$ is defined in~\eqref{eq:def:X-delta}.
Furthermore, there exists $\mathcal E_q$ satisfying~\eqref{eq:E_q bound}, where $\mathcal X_\delta\subseteq \mathcal Y$ for any $\delta\geq 1$.
\end{lemma}
\begin{proof}
The proof follows parts of the proof of \cite[Theorem~4.1]{MN-expanders2}.
Define $\mathcal Y=\bigl(\biguplus_{\delta\in[1,\infty)} \mathcal X_\delta\bigr )_1$ 
the $\ell_1$-union of $\{\mathcal X_\delta\}_{\delta\geq 1}$ with the specified points being the cusps of the cones $\{\mathcal X_\delta\}_{\delta\geq 1}$.
By \autoref{lem:ell1-union}, $\mathcal Y$ is a CAT(0) space. Our goal is now is to construct an expander $\mathcal E$, such that $\gamma_q^+(\mathcal E,\mathcal Y)<\infty$ for any $q\in\{2,3,\ldots\}$ 
using \autoref{lem:bootstraping-main-zigzag-iteration}.
For this end we now establish the hypotheses of 
\autoref{lem:bootstraping-main-zigzag-iteration}.

For a given $q\in\{2,3,\ldots\}$ let 
$K_q\geq 1$ be a sufficiently large number depending only on $q$ to satisfy \autoref{def:F}, \autoref{prop:rrg-structure} (both are satified by some universal constant independent of $q$), as well as \autoref{lem:q-speactral-calculus} (which depends on $q$).
In particular $\mathcal Y$ satisfies the spectral calculus
condition~\eqref{eq:spectral calculus condition} with power $q$ and constant $K_q$,  $\forall q\in\{2,3,\ldots\}$.

We next define the notion of size partition compatible with 
\autoref{def:size-partition}.
Recall that $\X_\delta$ is defined as a cone over a family $\F$ of graphs. We first partition $\F$ and $\X_\delta$ to ``small'' and ``large'' graphs/spaces. Let
\begin{equation}
 \label{eq:def:X_delta^*}
 \F_{\geq s}=\{G\in\F: |V_G|\ge s\},
\text{ and }
(\X_\delta)_{\geq s}=\cone(U_{\F_{\geq s}},d_{\delta \F}).
\end{equation}
and the complement
\begin{equation*}
 \F_{< s}=\{G\in\F: |V_G|< s\},
\text{ and }
(\X_\delta)_{< s}=\cone(U_{\F_{< s}},d_{\delta \F}).
\end{equation*}
Finally,
\begin{equation*}
\mathcal Y_{\geq s}=\Bigl(\biguplus_{\delta\in[1,\infty)} \mathcal (\X_\delta)_{\geq s}\Bigr )_1,
\ \text{and} \ 
\mathcal Y_{< s}=\Bigl(\biguplus_{\delta\in[1,\infty)} \mathcal (\
\X_\delta)_{< s}\Bigr )_1.
\end{equation*}

We next check that this partition satisfies the requirements of \autoref{def:size-partition}. Monotonicity is clear. 
Regarding the Euclidean distortion of $\mathcal Y_{<s}$:
First, fix $G\in\F_{<s}$. By \autoref{cor:biLip-extension} and Bourgain's embedding,
\[ c_2(\Sigma(G))\lesssim c_2(G)\lesssim \log(s+1),\]
Hence, for any $\tau>0$, 
\[
c_2(\cone(\Sigma(G),\tau d_{\Sigma(G)}))
\stackrel{\eqref{eq:cone(Lp) in Lp}}\lesssim
c_2(\Sigma(G),\tau d_{\Sigma(G)}) =
c_2(\Sigma(G)) \lesssim \log(s+1).
\]
Since 
\[ \mathcal Y_{<s}= \Bigl(\biguplus_{\delta\geq 1} \biguplus_{G\in\F_{<s}} \cone(\Sigma(G),\frac{2\pi\delta}{\girth(G)}d_{\Sigma(G)}\Bigr)_1
\]
(see \autoref{lem:ell1-union}), by \autoref{lem:ell1-union}
\[ 
c_2(\mathcal Y_{<s}) \leq \sqrt{2}\sup_{\tau>0}\sup_{G\in\F_{<s}} c_2(\Sigma(G),\tau d_{\Sigma(G)})\lesssim \log (s+1).
\]
Hence the Euclidean embedding condition is satisfied with $g(s)=C\log (s+1)$ for some universal constant $C>0$.
Lastly, observe (using \autoref{lem:ell1-union})
that $\mathcal Y_{<s}$ and $\mathcal Y_{\geq s}$ satisfy the condition of \autoref{lem:gamma-2-union-simple} and therefore
\[ \gamma_q^+(G,\mathcal Y)\leq 2^{q-1} \max\{\gamma_q^+(G,\mathcal Y_{<s}), \gamma_q^+(G,\mathcal Y_{\geq s})\},\]
so the last condition of \autoref{def:size-partition} holds with $h(q,x,y)=2^{q-1}\max\{x,y\}$. 

Define $(s_n)_n$ by $s_n=4n^2$, and $f_q(t)=(Cq^2t)^q$, for some sufficiently large universal constant. We next prove that these parameters satisfy~\eqref{eq:gamma_q^+(G,Y)-bound}.
Fix $n$-vertex graph $H$.
Fix $\delta\geq 1$ and $G\in\F_{\geq 4n^2}$.
Fix
\[\phi,\psi: V_H \to \cone\Bigl(\Sigma(G),\frac{2\pi \delta}{\girth(G)}d_{\Sigma(G)} \Bigr).
\]
Observe that
\[
\bigl| \phi(V_H) \cup \psi(V_H) \bigr|\le 2n_0
\stackrel{\eqref{eq:def:X_delta^*}}\leq \sqrt{|V_G|}
\]
so by~\eqref{eq:F-embed-L1} and~\eqref{eq:blow-out-cone-in-cone},
\[
c_1\Bigl(\phi(V_H) \cup \psi(V_H)\Bigr)\lesssim  K^2.
\]
Hence, 
\[
\gamma_q^+\Bigl(H,\cone\Bigl(\Sigma(G),\frac{2\pi \delta}{\girth(G)}d_{\Sigma(G)} \Bigr)\Bigr) \lesssim K^{2q} \gamma_q^+(H,L_1) \stackrel{\eqref{eq:gamma+-extrapolation}}\lesssim  K^{2q}(1+q^2)^q \gamma_2^+(H,\mathbb R)^q\leq f_q(\gamma_2^+(H,\mathbb R)).
\]
Since this is true for any $G\in \F_{\geq s_n}$ and any $\delta\geq 1$, by
\autoref{lem:gamma-2-union-simple}, \eqref{eq:def:X_delta^*}, and \autoref{lem:ell1-union}
\[
\gamma_q^+(H,\mathcal Y_{\geq 4n^2})\lesssim_q  f_q(\gamma_2^+(H,\mathbb R)),
\]
which means that~\eqref{eq:gamma_q^+(G,Y)-bound} is satisfied.

We have just demonostrated that all the hypotheses of
\autoref{lem:bootstraping-main-zigzag-iteration} hold for the space $\mathcal Y$.
Therefore there exists a $3$-regular expander graph $\mathcal E$ for which 
$\gamma_q^+(\mathcal E,\mathcal Y)<\infty$, for any $q\in\{2,3,\ldots\}$. 
Furthermore, by the same lemma, there exists $\mathcal E_q$ satisfying~\eqref{eq:E_q bound}.
Since $\mathcal X_\delta\subset \mathcal Y$, this proves the lemma. 
\end{proof}

\begin{proof}[Proof of \autoref{thm:poincare-cone}]
Let $\mathcal E$ be the family of expander graphs whose existence is asserted in \autoref{lem:X_delta-expander}.
For $m>d\ge 3$
Let $\RRG\sim \mathcal G_{m,d}$. With probability $1-o_m(1)$, $\RRG$ satisfies
items I--VI of \autoref{prop:rrg-structure} for any $\delta\ge 1$.
We assume from now on that it does, and fix $\delta \ge 1$.

Let $A_1,A_2\subset \Sigma(\RRG)$ be as in
\autoref{prop:rrg-structure}.
Denote by $h:\Sigma(\RRG)\times\Sigma(\RRG)\to [0,\infty)$ the metric $h=\delta\sigma_\RRG d_{\Sigma(\RRG)}$.
By \autoref{it:IV} of \autoref{prop:rrg-structure} and \autoref{lem:X_delta-expander}, we have
\[
\gamma_q^+(\mathcal E,\cone(A_2,h))\lesssim_q 1 .
\]
Since $\mathcal E$ is a classical expander and by \autoref{it:III} of \autoref{prop:rrg-structure},
\[
\gamma_q^+(\mathcal E,\cone(A_1,h)) \lesssim_q 
\gamma_q^+(\mathcal E,L_1)
\stackrel{\eqref{eq:gamma+-extrapolation}}\lesssim_q 
\gamma_2^+(\mathcal E,\mathbb R)^q\lesssim_q 1.
\]

By \autoref{it:II} of \autoref{prop:rrg-structure},
$h(A_1\setminus A_2, A_2\setminus A_1) \ge \delta$.
Hence, by \autoref{lem:gamma-2-union} we have
\[
\gamma_q^+(\mathcal E,\cone(\Sigma(\RRG),h)) \leq
\frac{\gamma_q^+(\mathcal E,\cone(A_1,h))+\gamma_q^+(\mathcal E,\cone(A_2,h))}{\pi^{2q}} \lesssim_q 1.
\]
The same bound
for any $\tau=\delta\sigma_\RRG\ge \sigma_\RRG$, which proves the theorem.
\end{proof}

\begin{corollary} \label{cor:rrg-expander-p-exp-growth}
For every $q\in(0,\infty)$ there exists 
an infinite family $\mathcal{E}_q$ of\/ $3$-regular graphs
and $\Gamma_q\in(0,\infty)$ such that 
for any integer $d\ge 3$, 
\[
\lim_{N\to \infty}\Pr_{\RRG\sim \mathcal G_{N,d}}
\left[\sup_\varphi \gamma_q\bigl (\mathcal{E}_q,\bigr (\RRG,\circ d_{\RRG}\bigr )\bigr)\le \Gamma_q\right]=1.
\]
Furthermore, the $j$-th graph in $\mathcal{E}_q$ has
$a_q\cdot b_q^j$ vertices.
\end{corollary}
This corollary is a weaker version of \autoref{thm:MN-rrg-expander-p} but with the growth of the expander graph explicitly stated to be exponential. 
This will be needed for the algorithmic application in \autoref{sec:UA-proofs}. 
The proof of the corollary follows word for word the proof of \autoref{thm:MN-rrg-expander-p} using \autoref{thm:poincare-cone} but instead of using the expander $\mathcal E$ of \autoref{lem:X_delta-expander}, we use the expander $\mathcal E_q$ of the same lemma.

\section{Estimation of the average distance in random regular graphs}
\label{sec:UA-proofs}

\begin{proof}[Outline of the proof of \autoref{thm:u-approxer}]
Any $X$-expander $\mathcal E$ with power $p=1$ with at most an exponential growth induces a universal approximator of the average distance with respect to $X$. 
This is essentially proved in~\cite[§2]{MN-expanders2}.
Here we just outline the proof. 

The straightforward approximator is as follows. 
For $n\in \mathbb N$, let $G_n=([n],E)\in\mathcal E$ be an $n$-vertex $X$-expander graph. The subset $E$ is the set of pairs produced by the approximator per~\eqref{eq:UA}.
On the response $w:E\to[0,\infty)$, the algorithm outputs the estimate $|E|^{-1}\sum_{(i,j)\in E} w(i,j)$.

We next prove that if $([n],d)$ is a metric on a subset of $X$ then
\begin{equation}\label{eq:def-UA}
\frac{1}{2|E|}\sum_{(i,j)\in E} d(i,j) \leq 
  \frac{1}{n^2}\sum_{i,j\in [n]}  d(i,j) \lesssim 
  \frac{1}{|E|}\sum_{(i,j)\in E} d(i,j)
\end{equation}

First, observe that the right inequality in~\eqref{eq:def-UA} is essentially equivalent to $X$-expander with power $p=1$.
The only difference is that in the definition of  approximators there exists an approximator for every size.
This issue is addressed in~\cite[§2]{MN-expanders2} in a straightforward way as follows.
In our case, by \autoref{cor:rrg-expander-p-exp-growth}, the sizes of the graphs in the expander constitute a geometric progression
$|V(G_j)|=a_1 b_1^j $.
The approximator for $n$-point subset $S$, where
$a_1 b_1^{j-1}< n\leq a_1 b_1^j$ will use a ``balanced" surjection $V(G_j)\to S$
such that for any two points in $S$, the sizes of their inverse images differ by at most one.
This argument degrades the approximation factor by a factor of at most $2$.

The left inequality in~\eqref{eq:def-UA} actually holds for any regular graph and any metric $X$ 
simply by using the triangle inequality on every edge $\{i,j\}\in E_n$ and every vertex $k\in[n]$,
\[d_X(x_i,x_j)\le d_X(x_i,x_k)+d_X(x_k,x_j).\]
Then, averaging these inequalities over all the edges $\{i,j\}\in E_n$ and the vertices $k\in [n]$.
\end{proof}

\begin{remark}
The proof above used $k$-vertex graph in the expander in order to approximate the average distance in $n$-point subset, and it required that $n\leq k$. Note that the number of distance queries is $3k/2$, so in order for this to be $O(n)$, it is required that $k=O(n)$.
This is the reason we need the growth of the number of vertices of the graphs in the expander to be at most exponential. 
\end{remark}

\autoref{prop:reg-universal} is proved by the following two propositions.

\begin{proposition} \label{prop:const-degree-universality}
    For any finite metric space $(X,d_X)$, and any $\e>0$, there exists a simple undirected, unweighted, $3$-regular graph $G=(V,E)$ equipped with its shortest path metric $d_G$ such that
    $c_{(V,d_G)}(X)\leq 1+\e$  (that is, $X$ embeds in ${(V,d_G)}$ with distortion at most $1+\e$.)
\end{proposition}
\begin{proof}
    Assume $|X|\geq 3$. Fix $\e\in(0,1)$. Let $\delta=\min_{x\ne y\in X}d_X(x,y)$.
    Define $\bar w(x,y)= \bigl \lceil 3 d_X(x,y)/(\e\delta)\bigr \rceil$ and notice that clearly for every $x,y\in X$
    \begin{equation}\label{eq:Prop-7.1-1} \frac{3}{\e\delta} d_X(x,y) \leq \bar{w}(x,y) \leq \frac{3}{\e\delta} d_X(x,y)+1 \leq \frac{1}{\delta}\Big(\frac{3}{\e}+1\Big) d_X(x,y).
    \end{equation}
    The shortest path metric $\hat w:X\times X \to \{0\} \cup \mathbb N$ on the weighted graph
    $(X,\binom{X}{2},\bar w)$ is the largest metric on $X$ satisfying $\hat{w}\leq\bar{w}$. Therefore,
    $(X,\hat w)$ is a metric space consisting of integer distances satisfying
    $$\frac{3}{\e\delta} d_X(x,y) \leq \hat{w}(x,y) \leq \frac{1}{\delta}\Big(\frac{3}{\e}+1\Big) d_X(x,y),$$
    and thus $c_{(X,\hat w)}(X,d_X)\le 1+\e/3$.

    We then construct a short sequence of graphs as follows.
   First, $G_1=(X,\binom{X}{2},w)$ is the complete graph on $X$ with weights defined as $w(x,y)=\hat w(x,y)\cdot \lceil 3|X|/\e\rceil$.
    Then, $G_2=(V_2,E_2)$ is obtained from $G_1$ by subdividing every edge $e$ of $G_1$ to $w(e)$ edges, each of weight $1$.
    $(X,w)$ is canonically and isometrically identified with a subset $X_2$ of $(V_2,d_{G_2})$.
    Note that in $G_2$ all vertices have degree at least $2$, and the only vertices of degree strictly greater than $2$ are in $X_2$.
    Afterwards, $G_3=(V_3,E_3)$ is obtained from $G_2$ by replacing every
    vertex $x\in X_2$ having degree $\deg_{G_2}(x)>3$,  with a
    cycle $x_1,\ldots, x_{\deg_{G_2}(x)}$ of length $\deg_{G_2}(x)$ and every vertex of this cycle is an endpoint of a unique $G_2$-edge adjacent to $x$.
    In $G_3$, the degrees are all either $2$ or $3$.
    We map every $x\in X_2$ to a vertex $\iota(x)\in V_3$ on the corresponding cycle.
    Then, given $x,y\in X_2$, clearly
    $$d_2(x,y) \leq d_3(\iota(x),\iota(y)) \leq d_2(x,y) + \frac{\deg_{G_2}(x)+\deg_{G_2}(y)}{2} \leq d_2(x,y) + |X| \leq \Big(1+\frac{\e}{3}\Big) d_2(x,y).$$
    That is, the embedding $\iota:X_2\to V_3$ has distortion of at most $1+\e/3$.
    Lastly, let $G'_3=(V'_3,E'_3)$ be a disjoint ``tagged'' copy of $G_3$.
    We end with the graph $G=(V_4,E_4)$ where $V_4=V_3\cup V'_3$ and
    \[E_4=E_3\cup E'_3\cup\{(u,u'):\  u\in V_3,\; \deg_{G_3}(u)=2\},\]
    where $u'\in V'_3$ is the tagged version of $u\in V_3$.
    $G$ is clearly $3$-regular, and $(V_3,d_{G})=(V_3,d_{G_3})$ as metric spaces.
    In particular, $(X,d_X)$ embeds in $G$ with distortion at most $(1+\e/3)^2\le 1+\e$.
\end{proof}

\begin{proposition}
For any simple $3$-regular graph $G=(V,E)$ and any $d\in\{4,5,\ldots\}$, 
there exists a simple $d$-regular graph $H=(U,F)$ in which $G$ embeds isometrically.
\end{proposition}
\begin{proof}
    Let $(G^i=(V^i,E^i))_{i=1}^{d-2}$ be $d-2$ disjoint copies of $G$.
    Denote $u^i\in V^i$, the copy of $u\in V$.
Let
\[ U=\bigcup_i V^i \quad \text{ and } \quad F=\bigcup_i E^i \cup \bigcup_{i\ne j} \{(u^i,u^j): u^1\in V^1\}.\]
    Clearly, $H=(U,F)$ is $d$-regular and $G$ is isometrically equivalent to any of its copies in $H$.
\end{proof}

Before proving lower bound on deterministic adaptive approximators, 
we begin, for pedagogical reasons, with proving the special case of same lower bound on deterministic \emph{non-adaptive} approximators. 
The proof is simpler and shorter, and is a convenient introduction to the proof of the general case.

\begin{proposition} \label{prop:nonadaptive-lb}
Fix $k\in\mathbb N$.
Any deterministic non-adaptive $\alpha$-approximator of the average distance
for arbitrary $n$-point metric spaces that uses at most $o_n\bigl(n^{\frac {k+1}{k}}\bigr)$ queries, 
must have $\alpha\ge (1-o_n(1))(2k+2)$.
\end{proposition}
\begin{proof}
Let $[n]$ be a set of $n$ points.
Suppose the approximator queries the set of pairs%
\footnote{We may assume without loss of generality that there is no query of the form $(u,u)$ in $E$.}
$E\subseteq \binom{[n]}{2}$, and $|E|\leq \theta_n n^{\frac {k+1}{k}}$, where 
\begin{equation}\label{eq:theta-non-adaptive}
    \theta_n\in(0,1),\quad \theta_n \to 0.
\end{equation}
Denote $\eta_n=\sqrt{\theta_n}\, n^{\frac{1}{k}}$.
We further assume that $\eta_n\geq 3$ for almost all $n$ except a finite number of them.
We will explain at the end of the proof why this assumption can be made without loss of generality.

Consider the graph $G=([n],E)$,
let $S=\{v\in [n]:\; \deg_G(v)\ge \eta_n\}$, and denote $s=|S|$.
Define distances $w:E\to [k]$ as follows:
\begin{equation*}
    w(x,y)=\begin{cases}
        1 & (x,y)\in E \cap \bigl(([n]\setminus S)\times ([n]\setminus S)\bigr)\\
        k+1 & (x,y)\in E \setminus \bigl(([n]\setminus S)\times ([n]\setminus S)\bigr).
    \end{cases}
\end{equation*}

We show the existence of two different metrics on $[n]$ that are compatible with the weighted graph $G=([n],E,w)$ as follows.
The first metric, $\ud$, is defined for every $x,y\in V$ by
\[
\ud(x,y)=\begin{cases}
    0 & x=y\\
    \tfrac12& (x,y)\in \bigl(([n]\setminus S)\times ([n]\setminus S)\bigr) \setminus E\ \text{ and } \ x\neq y\\
    1& (x,y)\in \bigl(([n]\setminus S)\times ([n]\setminus S)\bigr) \cap  E\\
    k+1& \text{o/w ($x\neq y$ and either } x\in S \text{ or } y\in S).
\end{cases}
\]
Clearly, $\ud$ is a metric compatible with $G=([n],E,w)$.
We bound its average distance from above.
Clearly,
\[s\le \frac{2|E|}{\eta_n}\leq \frac {2\theta_nn^{1+\frac{1}{k}}}{\sqrt{\theta_n} n^{\frac{1}{k}}} = o(n).\]
The average distance in $\ud$ is therefore
\[ \avg(\ud)\leq
\frac{2\theta_n n^{1+\frac1k}+ \tfrac12 n^2 + 2s(n-s)(k+1)}{n^2}
\leq \frac{o(n^2)+ \tfrac12 n^2+o(n^2)(k+1)}{n^2}\leq \tfrac12+o_n(1).
\]

The second metric, $\od$, is the shortest path metric on the weighted graph $G=([n],E,w)$.
Again, it is clear that $\od$ is a metric and compatible with $G$.
We bound its average distance from below.
Fix some $u\in V\setminus S$.
Consider a ball $B=B_{\od}(u,k)$ of radius $k$ in $G$ around $u$.
This ball cannot contain any vertex of $S$, so all the vertices in $B$
have degree smaller than $\eta_n$, so
\[|B|\leq \sum_{i=0}^{k}\eta_n^i\lesssim
\eta_n^{k} \leq \theta_n^{\frac{k}{2}} n= o(n).\]
Points outside $B$ are at distance at least $k+1$ from $u$.
Obviously, all points (except one) are at distance at least $k+1$ from every point in $S$.
Therefore, the average distance in $\od$ is at least
\[ \avg(\od) \geq \frac{n(n-o(n))(k+1)}{n^2}\geq (1-o_n(1))(k+1).
\]

Since $\ud$ and $\od$ agree on $E$, the algorithm cannot distinguish between them,
and therefore its approximation cannot be better than
$\frac{\avg(\od)}{\avg(\ud)}\geq (1-o_n(1))2(k+1)$.

We are left to explain why $\eta_n\geq 3$ without loss of generality. 
Indeed, Suppose that $\eta_n<3$ for some sufficiently large $n$.
In this case, $|E|\leq 3 n^{1-\frac{1}{k}}<n-1$. This means that the graph $G=([n],E,w)$ is 
disconnected. In this case 
for any $M\geq n(k+1)$ one can define the metric $d_M:[n]\times[n]\to [0,\infty)$ as follows.
\[ d_M(x,y)=
\begin{cases} d_G(x,y) & \text{ if } x,y \text{ in the same connected component,}\\
M & \text{ otherwise.}
\end{cases}
\]
It is straightforward to check that $d_M$ is a metric compatible with the weights $w$, 
for any $M\geq n(k+1)$. Here
$\avg(d_{n(k+1)})\leq n(k+1)$, while
$\avg(d_{(nk)^{10}})\geq n^8k^{10}$.
\end{proof}

\bigskip

We now present the lower bound on deterministic adaptive algorithms.
Define a deterministic adaptive algorithm to be an algorithm that behaves like a deterministic adaptive $\alpha$-approximation algorithm of the average distance (as in \autoref{def:adapt-approximator}), except that it skips the last step of the estimation of the average distance.

\begin{theorem} \label{thm:det-adaptive-lb}
Fix $k\in\mathbb N$.
Let $A$ be a deterministic adaptive approximator according to \autoref{def:adapt-approximator} that makes 
 $m=o(n^{\frac{k+1}{k}})$ queries for $n$ points.
Then there exists an adversary answering $A$'s queries such that
 for the queries $E\subseteq \binom{[n]}2$
and answers $w:E\to [k]$,  the following properties hold:
There exist two bounded metrics on $[n]$, $\od,\ud:[n]\times[n]\to [0,k+1]$ such that
\begin{compactenum}[a. ]
    \item For every $(u,v)\in E$, $\od(u,v)=\ud(u,v)=w(u,v)$. \label{it:(a)}
    \item $\od(x,y)=k+1$ for all pairs $(x,y)\in([n]\times[n])\setminus X$, for some $X\subset [n]\times[n]$, $|X|=o(n^2)$. \label{it:(b)}
    \item $\ud(x,y)=\tfrac12$ for all pairs $(x,y)\in ([n]\times[n])\setminus Y$, for some $Y\subset [n]\times[n]$, $|Y|=o(n^2)$. \label{it:(c)}
\end{compactenum}
\end{theorem}

\begin{proof}[Proof of\/ \autoref{thm:det-adaptive-lb}]
By the hypotheses of the theorem, there exists a sequence $(\theta_n)_n$, $\theta_n\in(0,e^{-1})$
satisfying:
\begin{equation} \label{eq:theta-eta}
m\leq \theta_n n^{\frac{k+1}{k}}, 
\quad \lim_{n\to\infty} \theta_n\to 0.
\end{equation}

Per the sequence of queries, denote $E_0=\emptyset$,
and $E_i=\{(x_1,y_1),\ldots, (x_i,y_i)\}$.
Let $w:E_i\to [k]$ the weights answered by the adversary up to time $i$.
Denote $G_i=([n],E_i,w)$ the associated weighted graph,
by $d_{G_i}:[n]\times[n]\to\{0,\infty\}\cup \mathbb N$ the shortest path metric on the weighted graph $G_i$.
We use the convention of $d_i(x,y)=\infty$ if $x$ and $y$ are disconnected in $G_i$, and by $d_i=\min\{d_{G_i},k+1\}$.

The adversary, on the $i$-th query $(x_i,y_i)$ will return $w(x_i,y_i)$  according to the following strategy.
Using the conventions $\max\emptyset=-\infty$, and $\max\{-\infty,a\}=a$, let
\begin{equation}\label{eq:def:h_i}
h_{i}(x)=\max  \{h\in\{0,\ldots,k-1\} :\;  \deg_{G_{i}}(x) \geq 
\theta_n^{1/2}\,(n^{\frac{h}{k}}-1)
\}).
\end{equation}
and
\begin{equation} \label{eq:def-w}
    w(x_i,y_i)= \max \left\{
    \begin{matrix}
    \min\bigl\{\max\{h_{i-1}(x_i),h_{i-1}(y_i)\}+1\,,\, d_{G_{i-1}}(x_i,y_i)\bigr\}, \\
    \max \bigl\{ w(u,v)-d_{G_{i-1}}(u,x_i)-d_{G_{i-1}}(v,y_i):\;(u,v)\in E_{i-1} \bigr\}
    \end{matrix}
    \right\}
\end{equation}
Clearly, $\deg_{G_1} \leq \deg_{G_2} \leq \cdots$, so we have the pointwise inequalities $h_1\leq h_2\leq\cdots$.
\begin{claim}  \label{cl:w in [k]}
 For every $x\in [n]$ and $i\in[m]$, $h_i(x)\in \{0,\ldots, k-1\}$.
 For every $(u,v)\in E_m$, $w(u,v)\in[k]$.
\end{claim}
\begin{proof}[Proof of \autoref{cl:w in [k]}]
The first assertion follows from~\eqref{eq:def:h_i}. 
The proof that $w(x_i,y_i)\in[k]$ now proceeds by induction
on $i$ and using \eqref{eq:def-w}.
\end{proof}

\begin{claim}\label{cl:d_i=w}
For every $i\in[m]$, and $(x,y)\in E_i$, we have
\(
d_i(x,y)=d_{G_i}(x,y)=w(x,y).
\)
\end{claim}
\begin{proof}[Proof of~\autoref{cl:d_i=w}]
    By definition of the shortest path metric $d_{G_i}$, clearly $d_{G_i}(x,y)\leq w(x,y)$ for every $i\in[m]$ and every $(x,y)\in E_i$, which further implies that $d_i(x,y)=d_{G_i}(x,y)$ for such $(x,y)$ due to  \autoref{cl:w in [k]}.
    Assume towards a contradiction that \autoref{cl:d_i=w} fails and let $i$ be the smallest index for which there exists  a pair $(x,y)\in E_i$ with
    \begin{equation} \label{eq:d_i(x,y)<w(x,y)}
    d_i(x,y)=d_{G_i}(x,y)<w(x,y).
    \end{equation}
     Let $j\le i$ be the index for which $(x,y)=(x_j,y_j)$. 
     We first prove that $w(x,y)=d_j(x,y)$. 
     This is clearly true if $j=1$, so assume that $j\geq2$. 
     Fix $(u,v)\in E_{j-1}$.
    By the inductive assumption, $d_{j-1}(u,v)=d_{G_{j-1}}(u,v)=w(u,v)$.
    The triangle inequality in $d_{G_{j-1}}$ then gives
    \[d_{G_{j-1}}(x,y) \ge d_{G_{j-1}}(u,v)-d_{G_{j-1}}(u,x)- d_{G_{j-1}}(v,y) = w(u,v)-d_{G_{j-1}}(u,x)- d_{G_{j-1}}(v,y).\]
    Observing~\eqref{eq:def-w}, this means that $w(x,y)\le d_{G_{j-1}}(x,y)$, and therefore $w(x,y)=d_{G_j}(x,y)=d_j(x,y)$.
    This means that $i>j$ and thus  $(x_i,y_i)\neq (x,y)$. By the definition of $x$, $y$ and the minimality of $i$ in~\eqref{eq:d_i(x,y)<w(x,y)}, we have
    \[
    d_i(x,y)<w(x,y)=d_{i-1}(x,y).
    \]
    At time $i$, the only way $d_i(x,y)$  can change relative to $d_{i-1}(x,y)$
    is via a shortest path that goes through the new edge $(x_i,y_i)$. Hence,
    \[
    d_{i-1}(x,x_i)+ w(x_i,y_i) +d_{i-1}(y_i,y) =d_i(x,y)<d_{i-1}(x,y)=w(x,y).
    \]
But this inequality contradicts the second line of~\eqref{eq:def-w} (when choosing $(u,v)=(x,y)$).
    This concludes the proof of \autoref{cl:d_i=w}.
\end{proof}

Let $E=E_m$, and $G=G_m=([n],E,w)$ and $d_G=d_{G_m}$ the shortest path metric on the weighted graph $G$.
Also, for $x\in{[n]}$, let
\begin{equation} \label{eq:h(x)}
h(x)= h_m(x)= \max\bigl\{ h\in \{0,\ldots,k-1\}: \deg_{G}(x)\geq 
\theta_n^{1/2} (
n^{\frac{h}{k}}-1
) 
\bigr\}.
\end{equation}

We next define the two metrics $\od$ and $\ud$ on $[n]$ that are compatible with the weights $w:E\to [k]$,
but $\od$ is mostly distance $k+1$, while $\ud$ is mostly distance $1/2$.
First, $\od$ is defined as
\[\od(x,y)=d_m(x,y)=\min\{d_G(x,y),k+1\}.\]
By \autoref{cl:w in [k]} and \autoref{cl:d_i=w}, $\od$ is indeed compatible with $w$,
i.e., it satisfies \autoref{it:(a)} of \autoref{thm:det-adaptive-lb}.

We next prove that most distances in $\od$ are $k+1$.
Let 
\[
H=\{x\in[n]: \deg_G(x)\geq \theta_n^{1/2}(n-1).\}
\]

\begin{claim} \label{cl:H_(k-1) bound}
    $|H|\lesssim \theta_n^{1/2} n^{\frac{1}{k}}$.
\end{claim}
\begin{proof}
By definition, for $x\in H$,
    $\deg_G(x)\geq \theta_n^{1/2} (n-1)$, and therefore
    \[
 |H| \leq \frac{2\theta_n n^{\frac{k+1}{k}}}{\theta_n^{1/2} (n-1)}  \lesssim \theta_n^{1/2} n^{\frac{1}{k}}
 \ . \qedhere
    \]
\end{proof}

\begin{claim} \label{cl:ball-size}
For every $u\in[n]$ and every $r\in[k]$, if either $u\notin H$ or $r\leq k-1$, then 
\[\bigl|B_{d_G}(u,r)  \bigr|\lesssim 2^{k-1} \theta_n^{1/2}n^{r/{k}}. \]
\end{claim}
\begin{proof}[Proof of \autoref{cl:ball-size}]
We first prove that
\begin{equation}\label{eq:sphere-size}
\bigl|B_{d_G}(u,r) \setminus B_{d_G}(u,r-1) \bigr|\lesssim 2^{k-1} \theta_n^{1/2} n^{r/k} 
.
\end{equation}
Let $S=\{(x_i,y_i)\in E: \max\{h_{i-1}(x_i),h_{i-1}(y_i)\}\geq d_{i-1}(x_i,y_i)\}$, where $h_{i-1}(x)$ is defined in~\eqref{eq:def:h_i}.
Per~\eqref{eq:def-w} (and the triangle inequality), for $(x_i,y_i)\in S$, $w(x_i,y_i)=d_{i-1}(x_i,y_i)$, which means that the edge $(x_i,y_i)$
does not affect the distances in $d_G$. 
In other words, the shortest path metric on the weighted graph $([n],E\setminus S,w)$ is the same as the metric $d_G$.

Per~\eqref{eq:def-w}, every edge $(x_i,y_i)\in E\setminus S$ has $w(x_i,y_i)\ge \max\{h_{i-1}(x_i),h_{i-1}(y_i)\}+1$. 
Fix $v\in[n]$ and $h\in\{0,\ldots,k-1\}$, such that either $v\notin H$ or $h\leq k-2$.
Then, 
\begin{equation}\label{eq:i-low-h}
\bigl|\bigl\{i: \ v\in\{x_i, y_i\}, \ (x_i,y_i)\in E\setminus S \mbox{ and } w(x_i,y_i) \leq h+1 \bigr\} \bigr| \leq \bigl| \bigl\{ i: \ v\in\{x_i, y_i\} \mbox{ and } h_{i-1}(v) \leq h\bigr\} \bigr|.
\end{equation}
Denote the set on the left-hand side of~\eqref{eq:i-low-h} as
$\{i_1<\ldots<i_\ell\}$. 
So  $h_{i_\ell-1}(v)\leq h$ since $h_1\leq h_2\leq\cdots$. 
Observe that $\ell\leq \deg_{G_{i_\ell}}(v)$,
and from the definition of $h(x)$ and $H$,
$\deg_{G_{i_\ell}}(v)< \theta_n^{1/2} (n^{(h+1)/{k}} -1)$.
So, 
\begin{equation} \label{eq:ell}
    \ell \leq \deg_{G_{i_\ell}}(v) \lesssim
\theta_n^{1/2} n^{\frac{h+1}{k}}.
\end{equation}
In other words, 
the number of edges in $E\setminus S$ adjacent to $v$ of weight at most $h+1$ is at most $O(\theta_n^{1/2} n^{\frac{h+1}{k}})$.

$|B_{d_i}(u,r) \setminus B_{d_i}(u,r-1)|$ is bounded above by the number of paths in $E\setminus S$ that begin at $u$ and have a (weighted) length of exactly $r$. 
This can be counted as follows: 
Fix a partition of $r$ to a sum of natural numbers
$r=a_1+a_2+\ldots+a_s$, the number of paths with $s$ edges which begin at $u$ such that $t$-th edge has weight $a_t$, is bounded as follows. 
\begin{itemize}
    \item 
If $r<k$ or $s>1$, then $a_j\leq k-1$ for $j\in\{1,\ldots,s\}$, so applying~\eqref{eq:ell}, 
we conclude that the number of paths is at most
$\prod_{t=1}^s 
\theta_n^{1/2}  n^{a_t/k} 
\leq\theta_n^{1/2} n^{r/k} 
$.
\item 
If $r=k$ and $s=1$, then $a_1=k$.
But in this case $u\notin H$ so by definition of $H$,
the number of edges of weight $k$ adjacent to $u$ is at
most $\theta_n^{1/2}n$.
\end{itemize}

The number of ways to partition $r$ in this fashion is $2^{r-1}$.
This concludes the proof of~\eqref{eq:sphere-size}.
Summing over all the spheres centered at $u$ of radii at most $r$ concludes the proof of
\autoref{cl:ball-size}.
\end{proof}

\begin{corollary} \label{cor:avg-od-lb}
The subset
\[X=\{(x,y)\in[n]\times[n]: \od(x,y)\leq k\} \]
satisfies \autoref{it:(b)} of \autoref{thm:det-adaptive-lb}.
\end{corollary}
\begin{proof}[Proof of \autoref{cor:avg-od-lb}]
By \autoref{cl:ball-size}, for every $u\in [n]\setminus H$,
\[|B_{\od}(u,k)|=|B_{d_G}(u,k)|\leq 2^{k-1}\theta_n^{1/2 }n 
\stackrel{\eqref{eq:theta-eta}}=o(n).\]
By \autoref{cl:H_(k-1) bound} $|H|=o(n)$,
hence, $|X|\leq n\cdot |H|+\sum_{u\in[n]\setminus H } |B_{\od}(u,k)|=o(n^2) $.
\end{proof}

The second metric  $\ud$ will be defined as a solution of a linear program as follows.
 Let $\ud:\binom{[n]}{2}\to \mathbb (0,\infty)$ be the variables of the program.
\begin{align}
\nonumber   \text{Optimize:} \qquad \qquad  &  \min \sum_{\{x,y\}\in\binom{[n]}{2}}\ud(\{x,y\})\\
\nonumber \text{subject to:} \qquad\qquad\\
 \label{eq:LP-boundary-cond} \ud(\{u,v\})&= w(u,v) & \forall \{u,v\}\in E  \\
  \label{eq:LP-ud-lb} \ud(\{x,y\}) &\geq
 \min\bigl\{ \od(x,y), \max\{h(x),h(y)\}+\tfrac12\bigr\}
  & \forall \{x,y\}\in \tbinom{[n]}{2}\setminus E
  \\
 \label{eq:LP-triangle-ineq} \ud(\{x,y\})+\ud(\{y,z\}) & \geq \ud(\{x,z\})   & \forall (x,y,z)\in [n]_{3}  \\
\label{eq:LP-ud-ub} \ud(\{x,y\}) &\leq
  \od(x,y) &   \forall \{x,y\}\in \tbinom{[n]}{2}
\end{align}

\begin{remark}
Any feasible solution to the above LP is clearly a metric compatible with $G=(V,E,w)$,
and therefore satisfies \autoref{it:(a)} of \autoref{thm:det-adaptive-lb}.
However, the optimal
$\ud$ does not necessarily have the smallest average distance.
To obtain the minimal (pseudo)metric, the constraint
~\eqref{eq:LP-ud-lb} should be replaced with $\ud(\{x,y\})\geq 0$,
and the constraint~\eqref{eq:LP-ud-ub} should be removed.
We added them to simplify the analysis. 
\end{remark}

The set of feasible solutions is closed, convex, and bounded.
It is nonempty since it contains the metric $\od$.
Hence, there is an optimal solution for the linear program, denoted by $\ud$.

\begin{claim}\label{cl:ud=od}
For $x,y\in [n]$, if\/ $\ud(x,y)< h(x)+\tfrac12$ then $\ud(x,y)=\od(x,y)\in \mathbb N$.
\end{claim}
\begin{proof}[Proof of \autoref{cl:ud=od}]
For $(x,y)\in E$, it follows from~\eqref{eq:LP-boundary-cond}, \autoref{cl:w in [k]} and \autoref{cl:d_i=w} that 
$\ud(x,y)=\od(x,y)=w(x,y)$.
For $(x,y)\notin E$,
The equality $\ud(x,y)=\od(x,y)$ follows directly
from~\eqref{eq:LP-ud-lb} and~\eqref{eq:LP-ud-ub} and the assumption of the claim.
The integrality of $\ud(x,y)$ follows from the integrality of $\od$
(which follows since $\od=\min\{d_G,k+1\}$, and $d_G$ is integral by \autoref{cl:w in [k]}).
\end{proof}

Define
\[
U= \bigcup\bigl\{B_{\ud}(z,h(z)-1): \; z\in [n],\;h(z)\geq 1\bigr\}.
\]

\begin{claim} \label{cl:|U|=o(n)}
    $|U|=o(n)$.
\end{claim}
\begin{proof}[Proof of \autoref{cl:|U|=o(n)}]
Let $z\in [n]$.
From~\eqref{eq:LP-ud-ub} and \autoref{cl:ud=od} it follows that $B_{\ud}(z,r)=B_{\od}(z,r)$ for any $r< h(z)+1/2$, and in particular
$B_{\ud}(z,h(z)-1)=B_{\od}(z,h(z)-1)$ when $h(z)\geq 1$.
Hence, since $h(z)-1\leq k-1$, 
by \autoref{cl:ball-size},
\[
|B_{\ud}(z,h(z)-1)|\lesssim 2^{k-1} \theta_n^{1/2} n^{\frac{h(z)-1}{k}}
.
\]
So, 
\begin{multline*}
|U|\leq \sum_{z:\ h(z)\geq 1} |B_{\ud}(z,h(z)-1)|
 \lesssim
 2^{k-1} \sum_{z:\ h(z)\geq 1} \theta_n^{1/2} n^{\frac{h(z)-1}{k}}
\\ \leq 
2^{k-1} \theta_n^{1/2}{n^{1-\frac{1}{k}}}
+2^{k-1} n^{\frac{-1}{k}}
\sum_{z:\ h(z)\geq 1} \theta_n^{\frac{1}{2}} (n^{\frac{h(z)}{k}}-1)
\stackrel{\eqref{eq:h(x)}}
\leq o(n)+ 2^{k-1} n^{\frac{-1}{k}} \sum_{z:\ h(z)\geq 1} \deg_G(z)\\
\lesssim o(n)+2^{k} n^{\frac{-1}{k}}\theta_n n^{\frac{k+1}{k}} 
 \stackrel{\eqref{eq:theta-eta}}=o(n).
\quad \qed
\end{multline*}
\renewcommand{\qedsymbol}{}
\end{proof}

Define 
\begin{equation} \label{eq:def W}
    W=\bigcup\Bigl\{ (x,v):\ \exists u\in[n],\; (u,v)\in E,\; w(u,v)=h(u)+1,\;\ud(u,x)\leq h(u)\Bigr\}.
\end{equation}

\begin{claim}\label{cl:|W|=o(n^2)}
    $|W|=o(n^2)$.
\end{claim}
\begin{proof}[Proof of \autoref{cl:|W|=o(n^2)}]
Clearly,
\( W\subseteq \bigcup_{u\in [n]} \bigl(B_{\ud}(u,h(u))\times \{v: (u,v)\in E\}\bigr).\)
Therefore,
\[
|W|\leq \sum_{\substack{u\in[n]}} |B_{\ud}(u,h(u))|\cdot \deg_G(u)
\]
Since $h(u)\leq k-1$, by \autoref{cl:ball-size},
$|B_{\ud}(u,h(u))|\lesssim 2^{k-1} \theta_n^{1/2} n^{\frac{k-1}{k}}$, for any $u\in[n]$. So, 
\begin{equation*}
|W|
 \lesssim
2^{k-1} \theta_n^{1/2}n^{\frac{k-1}{k}}|E| 
\leq
2^{k-1}\theta_n^{1/2}n^{\frac{k-1}{k}} \theta_n n^{\frac{k+1}{k}}  
\stackrel{\eqref{eq:theta-eta}}=o(n^2).
\qedhere
\end{equation*}
\end{proof}

Denote $\bar W=\{(v,x): (x,v)\in W\}$
\begin{claim} \label{cl:ud(x,y)=1}
    For every non-edge pair $(x,y)\notin  W\cup \bar W\cup (U\times[n])\cup ([n]\times U)$, \   $\ud(x,y)\in \bigl\{0,\tfrac12\bigr\}$.
\end{claim}
\begin{proof}[Proof of \autoref{cl:ud(x,y)=1}]
Let  $(x,y)\in [n]\times[n] \setminus \bigl( W\cup (U\times[n])\cup ([n]\times U)  \bigr)$ be a non-edge and
assume without loss of generality that $x\neq y$.
From~\eqref{eq:LP-ud-lb}, $\ud(x,y)\geq \tfrac12$.
Suppose for the sake of contradiction that $\ud(x,y)>\tfrac12$.
Moreover, assume that the pair $x,y$ maximizes the distance $\ud(x,y)$ under these assumptions.
Since $\ud$ is optimal for the LP, one of the lower bound conditions on the distance $\ud(x,y)$ must be tight,
since otherwise it would have been possible to decrease $\ud(x,y)$, contradicting the optimality of $\ud$.
As $x,y\notin U$, 
$h(x)=h(y)=0$ and since
$\ud(x,y)>\tfrac12$, this implies that the constraint \eqref{eq:LP-ud-lb} is not tight.
Condition~\eqref{eq:LP-boundary-cond} is not relevant as $(x,y)\notin E$.
Hence one of the LP triangle inequalities~\eqref{eq:LP-triangle-ineq} involving $(x,y)$ on the left-hand side must be tight.
Therefore, there exists $z\in [n]\setminus \{x,y\}$ such that (without loss of generality), the shortest path from $x$ to $z$ passes via $y$, namely we have
\begin{equation} \label{eq:tight-xyz}
\ud(x,y)+\ud(y,z)=\ud(x,z).
\end{equation}
We can now iterate this process, asking whether there exists a point $w$ such that the shortest path from $x$ to $w$ passes via $z$ or the shortest path from $w$ to $z$ passes via $x$.
Since at each step the length of the larger path increases by at least $\tfrac12$ and it is upper bounded by $k$, this process must end at some point. That is, we eventually reach a pair $(u,v)\in[n]\times[n]$ 
which is not part of any larger shortest path with respect to $\ud$.
The pair $(u,v)$  then
satisfies
\begin{equation} \label{eq:tight-uxyv}
\ud(u,x)+\ud(x,y)+\ud(y,v)=\ud(u,v).
\end{equation}
In particular, observe that since either $u\neq x$ or $v\neq y$,
\begin{equation} \label{eq:ud(u,v)>1}
    \ud(u,v)\geq \ud(x,y)+\max\{\ud(u,x),\ud(v,y)\}\geq
    \ud(x,y)+\tfrac12>1.
\end{equation}
As before, $\ud(u,v)$ must have a tight  LP lower bound but
it cannot be the triangle inequality~\eqref{eq:LP-triangle-ineq}.
Hence, either~\eqref{eq:LP-boundary-cond} or~\eqref{eq:LP-ud-lb} are tight for $\ud(u,v)$.
\begin{enumerate}[(a)]
    \item \label{it:2k-lb-(a)}
    Suppose first that~\eqref{eq:LP-ud-lb} is tight for
    $(u,v)$. In particular $\ud(u,v)\leq \max\{h(u),h(v)\}+\tfrac12$, so $\max\{h(u),h(v)\}\geq 1$ by \eqref{eq:ud(u,v)>1}.
    Suppose without loss of generality that $h(u)=\max\{h(u),h(v)\}$.
    By~\eqref{eq:tight-uxyv} 
    we have 
    \[
    \ud(u,x)\leq \ud(u,v)-\ud(x,y)< h(u)+\tfrac12-\tfrac12=h(u).
    \]
By~\autoref{cl:ud=od}, $\ud(u,x)$ is integral, which
means that $\ud(u,x)\leq h(u)-1$.
Hence $x\in B_{\ud}(u,h(u)-1)\subseteq U$, which contradicts the assumption of \autoref{cl:ud(x,y)=1} that $x\notin U$.

\item Next, suppose~\eqref{eq:LP-boundary-cond} is tight for $(u,v)$, which means that $(u,v)\in E$.
We claim that in this case,
\begin{equation}\label{eq:ud-simple-formula}
\ud(u,v)\leq \max\{h(u),h(v)\}+1.
\end{equation}
Indeed, suppose $(u,v)$ was queried at the $i$-th query, i.e.
  $(u,v)=(x_i,y_i)$ in~\eqref{eq:def-w}.
Suppose toward a contradiction that $w(x_i,y_i)> \max\{h(x_i),h(y_i)\}+1$.
Since $h_i$ is non-decreasing in $i$, we also have
$w(x_i,y_i)> \max\{h_{i-1}(x_i),h_{i-1}(y_i)\}+1$.
But looking at~\eqref{eq:def-w} this can only happen when there exists $(u',v')\in E_{i-1}\subseteq E$ such that
\[
w(u,v)=w(u',v')-d_{i-1}(u',u)-d_{i-1}(v',v)\leq
w(u',v')-\ud(u',u)-\ud(v',v),
\]
since $\ud \leq \od \leq d_{i-1}$. As $(u,v), (u',v')\in E$, the triangle inequality in $\ud$ and the constraint \eqref{eq:LP-boundary-cond} show that the above inequality is equality, i.e.
\begin{equation}\label{eq:tight-ti}
\ud(u,v)+\ud(u',u)+\ud(v',v) = w(u,v)+\ud(u',u)+\ud(v',v)=w(u',v')=\ud(u',v').
\end{equation}
But~\eqref{eq:tight-ti} is a tight triangle inequality~\eqref{eq:LP-triangle-ineq} of the LP.
This contradicts the property of the pair $(u,v)$ (spelled out few lines above~\eqref{eq:tight-uxyv}) not being part of any larger shortest path.
This concludes the proof of~\eqref{eq:ud-simple-formula}. We now consider subcases of this case:

\begin{itemize}
    \item If $\ud(u,v)< \max\{h(u),h(v)\}+1$,
    then (since $\ud(u,v)=w(u,v)$ is integral) we have $\ud(u,v)\leq \max\{h(u),h(v)\}$.   
    Suppose without loss of generality that $h(u)=\max\{h(u),h(v)\}$.
    As in \autoref{it:2k-lb-(a)} above, 
    by~\autoref{cl:ud=od}, $\ud(u,x)<\ud(u,v)\leq h(u)$ is integral, which
means that $\ud(u,x)\leq h(u)-1$.
Hence $x\in B_{\ud}(u,h(u)-1)\subseteq U$, which contradicts the assumption of \autoref{cl:ud(x,y)=1} that $x\notin U$. 

\item
    Otherwise 
     $\ud(u,v)= \max\{h(u),h(v)\}+1$.
         Suppose without loss of generality that we have $h(u)=\max\{h(u),h(v)\}$.
\begin{itemize}
\item
    If $v\neq y$, then $\ud(v,y)\geq \tfrac12$,
    and
    \[\ud(u,x)=\ud(u,v)-\ud(v,y)-\ud(x,y)<h(u)+1-\tfrac12-\tfrac12= h(u).\]
    As in \autoref{it:2k-lb-(a)} above, 
    by~\autoref{cl:ud=od}, $\ud(u,x)<h(u)$ is integral, which
means that $\ud(u,x)\leq h(u)-1$.
Hence $x\in B_{\ud}(u,h(u)-1)\subseteq U$, which contradicts the assumption of \autoref{cl:ud(x,y)=1} that $x\notin U$. 
\item We are left with the case $\ud(u,v)=h(u)+1$,
and $v=y$.
    By~\eqref{eq:tight-uxyv} 
    we have 
    \[
    \ud(u,x)= \ud(u,v)-\ud(x,y)< h(u)+1-\tfrac12=h(u)+\tfrac12.
    \]
By~\autoref{cl:ud=od}, $\ud(u,x)$ is integral, which
means that $\ud(u,x)\leq h(u)$. That is
\[(x,y)=(x,v)\in 
W,\]
which contradicts the assumption of \autoref{cl:ud(x,y)=1} that $(x,y)\notin W$. 
\end{itemize}
\end{itemize}
\end{enumerate}
In all cases of the analysis above we reached a contradiction, so the proof of \autoref{cl:ud(x,y)=1} is complete.
\end{proof}

It follows immediately from \autoref{cl:|U|=o(n)}, \autoref{cl:|W|=o(n^2)}, and \autoref{cl:ud(x,y)=1} that the set
\[
Y=(U\times[n]) \cup ([n]\times U) \cup W \cup \bar W \cup \{(x,x):x\in[n]\} \cup \{(x,y):\{x,y\}\in E\}
\]
satisfies \autoref{it:(c)} of \autoref{thm:det-adaptive-lb}.
This concludes the proof of \autoref{thm:det-adaptive-lb}.
\end{proof}

\begin{corollary}[Generalization and rephrasing of \autoref{thm:impossibility-DA}]
\label{cor:impossibility-DA-power-p}
Fix $k\in \mathbb N$ and $p\in(0,\infty)$.
Suppose that $A$ is a deterministic adaptive $\alpha$-approximator of the average of the $p$-th power of the distances.
If\/ $A$ uses $o(n^{\frac{k+1}{k}})$ queries for $n$ points, then $\alpha\geq (1-o_n(1))\cdot (2k+2)^p$.
\end{corollary}
\begin{proof}
We use the terminology of \autoref{thm:det-adaptive-lb}.
Since $\ud$ and $\od$ agree on the queried pairs $E$, the approximator cannot distinguish between them,
and therefore its approximation factor cannot be better than
$\frac{\avg({\smash{\od}}^{\,p})}{\avg(\ud^p)}$.

\begin{align*}
    \avg\bigl({\smash{\od}}^{\,p}\bigr)&\geq \frac{(n^2-o(n^2))(k+1)^p}{n^2}\geq (1-o_n(1))(k+1)^p,\\
    \avg\bigl(\ud^p\bigr)& \leq \frac{n^2
    \cdot \bigl(\tfrac12\bigr)^p
    + o(n^2)+ o(n^2)(k+1)^p}{n^2}\leq (1+o_n(1))
    \cdot \bigr(\tfrac12\bigr)^p
    .
\end{align*}

Hence, $\frac{\avg({\smash{\od}}^{\,p})}{\avg(\ud^p)}\geq (1-o_n(1))\cdot (2(k+1))^p$.
\end{proof}

\autoref{cor:impossibility-DA-power-p} (with $p=1$) applies only for approximations at least $4$. 
The range $[1,4)$ is handled below. 
It is much simpler to analyze because
the triangle inequality has no effect. 

\begin{proposition} \label{prop:UA-lb14}
Fix  $\e\in(0,1]$. Suppose that $A$ is a deterministic adaptive $\alpha$-approximator of the average distance.
If\/ $A$ uses at most $\e\binom n2$ queries for $n$ points,  
then $\alpha \geq 4 (1-\e/2) / (1+\e)$.
\end{proposition}
\begin{proof}
 The adversary answers '$w(x,y)=1$' for all queried pairs
 $(x,y)\in \binom{[n]}{2}$.
 Let $E\subseteq \binom{[n]}2$ be the queried pairs.
Define the metrics
\begin{align*}
    \od(x,y)&=\begin{cases}
        0 & x=y\\ 1& (x,y)\in E\\ 2&\text{o/w}.
    \end{cases}
    &
    \ud(x,y)&=\begin{cases}
        0 & x=y\\ 1& (x,y)\in E\\ \tfrac12&\text{o/w}.
    \end{cases}   
\end{align*}
Both are obviously metrics compatible with $G=(V,E)$. 
 In this case we have
\begin{align*}
\avg(\od)&=\frac{2\e\binom{n}{2}\cdot 1+2(1-\e)\binom{n}{2}\cdot 2}{n^2}=2-\e\pm o_n(1),\\
\avg(\ud)&=\frac{2\e\binom{n}{2}\cdot 1+2(1-\e)\binom{n}{2}\cdot \frac12}{n^2}=\frac12+\frac{\e}{2}\pm o_n(1).
\end{align*} 
Hence,
\[
\frac{\avg(\od)}{\avg(\ud)}
\geq \frac{2-\e}{\frac12+\frac{\e}2} -o_n(1)
= 4\frac{1-\frac{\e}2}{1+{\e}} -o_n(1). \qedhere
\]

\end{proof}

\begin{question}
Is the lower bound in \autoref{prop:UA-lb14} is tight for adaptive approximators? 
\end{question}

\section{Distortion growth of embedding into \texorpdfstring{$\ell_2$}{l\_2} and \texorpdfstring{$\ell_1$}{l\_1}}
\label{sec:no-dichotomy}

In this section, we prove \autoref{thm:no-embedding-dichotomy-ell_2} and \autoref{thm:no-embedding-dichotomy-ell_1} about the distortion growth of embedding CAT(0) subsets in $\ell_2$ and $\ell_1$.
The spaces $X_\varphi$ in \autoref{thm:no-embedding-dichotomy-ell_1} are constructed using
concepts similar to those used in \autoref{sec:extrapolation-rrg}.
Specifically,  $X_\varphi$ would be essentially the metric $\mathcal X_\delta$ defined in~\eqref{eq:def:X-delta},
but with a ``non-constant''
 $\delta$.
Other differences are as follows. The $L_1$ embeddability condition~\eqref{eq:F-embed-L1} is replaced with a stronger requirement (\eqref{eq:almost-cover-tree} below) required to prove $L_2$ embeddability, and an added
expansion requirement needed for $L_2$ and $L_1$ non-embeddability.
These additional requirements were already proved to exist in~\cites{ALNRRV}{MN-expanders2}.

\begin{definition}\label{def:L}
Fix two integers $n, d > 3$ and $K\in(1,\infty)$.
Denote by ${L}_K^{n,d}$ the set of $n$-vertex graphs $L=(V_L,E_L)$ whose maximum degree is at most $d$
with the following properties.

\noindent $\bullet$ First,
\begin{equation}\label{eq:hi-girth-log-diam}
\diam(L)\le K\cdot \min\{\girth(L), \log_d n\}.
\end{equation}
$\bullet$ For every $S\subseteq V_L$ with $|S|\le n/2$ we have
\begin{equation}\label{eq:cheeger}
|E_L(S,V_L\setminus S)|\ge |S|/K.
\end{equation}
$\bullet$ For every $S\subset \Sigma(L)$ with $|S|\le \sqrt{|V_L|}$,
there exists $U\subset V_L$ with the following properties.
Denote by $H=(U,E_L(U))$ the induced subgraph of $U$. Then,
\begin{equation} \label{eq:S usbset U}
S\subseteq \Sigma(H)
\end{equation}
and for every $x,y\in S$,
\begin{equation} \label{eq:bi-Lipschitz-U}
d_{\Sigma(L)}(x,y)\le d_{\Sigma(H)}(x,y)\le K\cdot d_{\Sigma(L)}(x,y);
\end{equation}
Moreover, let $\mathcal T(\Sigma(H))$ be the set of spanning trees
of $\Sigma(H)$.
Then, there exists a probability measure $\mu$ over  $\mathcal T(\Sigma(H))$  such that for every non-self-loop edge $e\in E(H)=E_L(U)$, and every $x,y\in [e]$,
\begin{equation}\label{eq:almost-cover-tree}
\mu(\{T \in \mathcal T(\Sigma(H)):\; [x,y]\subseteq T\})\ge 1-\frac{C_d \, d_{\Sigma(H)}(x,y)}{\log n},
\end{equation}
where $[x,y]\subseteq [e]$ is the closed interval whose endpoints are $x$ and $y$, and $[e]$ is the unit interval $\Sigma(H)$ corresponding to the edge $e$.
\end{definition}

\begin{lemma}\label{lem:L_K^{n,d}}
There exists a universal constant $K\in [1,\infty)$ such that for every even%
\footnote{The proof of \autoref{lem:L_K^{n,d}} is based on~\cite[Lemma~3.12]{MN-expanders2}.
In the statement of \cite[Lemma~3.12]{MN-expanders2}, $d$ is not required to be even. This is clearly an oversight, since $nd$ is always even in simple $d$-regular graphs on $n$ vertices.}
integer $d > 3$ there exists $n_0\in \mathbb N$ such that for all
$n>n_0$, ${L}_K^{n,d}\ne \emptyset$.
\end{lemma}

Before proving \autoref{lem:L_K^{n,d}} we recall the notion of
$(1+\delta)$-sparse graphs defined in~\cite{ALNRRV,MN-expanders2}.
\begin{definition}
A graph $G=(V,E)$ is called $(1+\delta)$-sparse if
\begin{equation}
\forall S\subseteq V, \ |E_G(S)|\le (1+\delta)\cdot |S|.
\end{equation}
\end{definition}

\begin{lemma}[{\cite[Lemma~6.3]{MN-expanders2}}]
\label{lem:Lemma6.3-MN}
Fix $\delta\in (0,1)$, and integer $g\ge 3$. Then a $(1+\delta)$-sparse graph $G=(V,E)$ with girth at least $g$ satisfies
\begin{equation}
\forall S\subseteq V, \ |E_G(S)|\le (1+3\max\{\delta,1/(g-1)\})\cdot (|S|-1).
\end{equation}

\end{lemma}

\begin{proof}[Proof of \autoref{lem:L_K^{n,d}}]
The proof of this lemma is the same as the proof of~\cite[Lemma~3.12]{MN-expanders2},
which is mainly based on arguments from~\cites{ALNRRV}{ABLT}.
We next describe the outline of that proof and the changes that are made. The description here is not complete and should be read together with the proof of~\cite[Lemma~3.12]{MN-expanders2}.

Fix an even integer $d>3$. Lemma~3.12 from~\cite{MN-expanders2} asserts that with probability $1-C_d/n^{1/3}$ random $d$-regular graphs have the property $\mathcal{L}_K^{n,d}$ as defined in Definition~3.11 in~\cite{MN-expanders2}.
Hence, $\mathcal{L}_K^{n,d}\ne\emptyset$ for sufficiently large $n$.
Per that definition, for each $G\in \mathcal L_K^{n,d}$, there exists a subset of edges $I_G\subset E_G$
(of size at most $\sqrt{n}$, but this is immaterial in the current context) with the following property:
any graph in the class $\{L=(V_G,E_G\setminus I_G):\; G\in\mathcal L_K^{n,d}\}$ has degree at most $d$ and satisfies~\eqref{eq:hi-girth-log-diam} and~\eqref{eq:cheeger}.

We are left to prove that such a graph will also admit an induced subgraph $H=(U,E_L(U))$ with properties~\eqref{eq:S usbset U}, \eqref{eq:bi-Lipschitz-U}, and~\eqref{eq:almost-cover-tree}.
Lemma~3.12 of~\cite{MN-expanders2} does not state those properties, but its proof in~\cite[§7.1]{MN-expanders2} contains them, as we will detail next.

As stated in~\cite{MN-expanders2} before Inequality~(199) there,
a subset $U\subset V_L$ satisfying~\eqref{eq:S usbset U} and~\eqref{eq:bi-Lipschitz-U} is
guaranteed to exist by \cite[Lemma~7.5]{MN-expanders2}
(Inequality~\eqref{eq:bi-Lipschitz-U} is presented in~\cite[Inequality~(201)]{MN-expanders2}).
As explained after~\cite[Formula~(201)]{MN-expanders2}, $|U|\le n^{2/3}$,
which means, as noted there, that $H=(U,E_L(U))$ is $(1+K/\log_d n)$-sparse.
According to \autoref{lem:Lemma6.3-MN} and the bound
$\girth(H)\gtrsim \log_dn$ (see~\eqref{eq:hi-girth-log-diam}), $H$ satisfies
\[\forall T\subset U,\ |E_L(T)| \le (1+K/\log_d n) (|T|-1).
\]
Hence, by \cite[Formula~(170)]{MN-expanders2}, which is based on  \cite[Claim~3.14]{ALNRRV} (see also~\cite[Claim~6.4]{MN-expanders2}), there exists%
\footnote{Formula~(170) in~\cite{MN-expanders2} is stated for pair of points on an edge ``not in $E(\Gamma)$''. In the context of~\cite{MN-expanders2}, $E(\Gamma)$ is the set of edges in ``short'' cycles. In the context of the current paper, since the underline graph $G$ has a high-girth, $E(\Gamma)=\emptyset$.}
 a probability measure $\mu$ on the spanning trees of $H$ satisfying~\eqref{eq:almost-cover-tree}.
\end{proof}

\begin{lemma}\label{lem:subset-L2-embedding}
Let $L=(V,E)\in L_K^{n,d}$. Then
for every $S\subset \Sigma(L)$ with $|S|\le \sqrt{n}$ and every $\alpha>0$,
\begin{align}
c_1(S,\min\{d_{\Sigma(L)}, \alpha\}) &\lesssim_d 1;
\label{eq:ALNRRV-L1-embedding}
\\
\label{eq:ALNRRV-L2-embedding}
c_2(S,\min\{d_{\Sigma(L)}, \alpha\}) &\lesssim_d
\min \Bigl\{
\sqrt{\log |S|} \ , \
\max\Bigl\{ \sqrt{\log\log |S|}, \alpha\big/\sqrt{\log n}\Bigr\}
\Bigr\}.
\end{align}
\end{lemma}
\begin{proof}
We may assume without loss of generality that
\begin{equation} \label{eq:alpha-bounded}
\alpha\lesssim \log_d n.
\end{equation}
In fact, by~\eqref{eq:hi-girth-log-diam},
$\diam(L)\lesssim \log_dn$, so for every $\alpha\gtrsim \log_d n$,
$\min\{d_{\Sigma(L)}, \alpha\}\asymp d_{\Sigma(L)}
\asymp \min\{d_{\Sigma(L)}, \log_d n\}$.

We begin with poving~\eqref{eq:ALNRRV-L1-embedding}.
Let $U\subset V$, and $H=(U,E_L(U))$ and $\mu$ be a probability measure over $\mathcal T(\Sigma(H))$ as in \autoref{def:L}.
By~\eqref{eq:S usbset U} and \eqref{eq:bi-Lipschitz-U} it is sufficient to prove
$c_1(\Sigma(H),\min\{d_{\Sigma(H)}, \alpha\}) \lesssim 1$.
To do that, we first argue that
the metric $\int_{\mathcal T(\Sigma(H))}\min\{d_T,\alpha\}\,\mathsf{d}\mu(T)$ is a constant approximation of $\min\{d_{\Sigma(H)},\alpha\}$.
Indeed, for every $T\in\mathcal T(\Sigma(H))$, $d_T\ge d_{\Sigma(H)}$, and therefore
\[
\min\{d_{\Sigma(H)}(x,y),\alpha\}\le \int_{\mathcal T(\Sigma(H))}\min\{d_T(x,y),\alpha\}\,\mathsf d\mu(T), \quad \forall x,y\in\Sigma(H).
\]
In the reverse direction, we prove that
\begin{equation} \label{eq:stam1}
    \min\{d_{\Sigma(H)}(x,y),\alpha\}\gtrsim \int_{\mathcal T(\Sigma(H))}\min\{d_T(x,y),\alpha\}\,\mathsf d\mu(T), \quad \forall x,y\in\Sigma(H).
\end{equation}
If $d_{\Sigma(H)}(x,y)\ge \alpha$, then~\eqref{eq:stam1} is vacuously true.
Thus, we assume  $d_{\Sigma(H)}(x,y)< \alpha$, and in this case
it is sufficient to prove it for $x,y\in\Sigma(H)$ that belongs to the same closed unit interval that corresponds to an edge of $H$.
Let $e\in E(H)$ be such that
$x,y\in [e]$.
Obviously, for $T\in\mathcal T(\Sigma(H))$, if $[x,y]\subset T$, then $d_T(x,y)=d_{\Sigma(H)}(x,y)$.
Recalling~\eqref{eq:almost-cover-tree},
\[\mu(T\in\mathcal T(\Sigma(H)):\; [x,y]\not\subset T)\lesssim C_d d_{\Sigma(H)}(x,y)/\log n,
\]
we have
\begin{equation} \label{eq:aa}
\begin{aligned}
\int_{\mathcal T(\Sigma(H))} \min&\{d_T(x,y),\alpha\}\,\mathsf d\mu(T)
\\ & \le \mu([x,y]\subset T)\cdot \min\{d_{\Sigma(H)}(x,y),\alpha\}  +  \mu([x,y]\not\subset T) \cdot \alpha
\\ & \le d_{\Sigma(H)}(x,y)  +
\frac{C_d d_{\Sigma(H)}(x,y)}{\log n} \alpha
\\ & \stackrel{\eqref{eq:alpha-bounded}}{\lesssim_d} d_{\Sigma(H)}(x,y).
\end{aligned}
\end{equation}
Since tree metrics are $L_1$ metrics, truncations of $L_1$ metrics embed with constant distortion in $L_1$~\cite[Lemma~5.4]{MN-expanders2},
and convex combination of $L_1$ metrics are $L_1$ metrics, we conclude that
\[
c_1\bigl(S,\min\{d_{\Sigma(L)},\alpha\}\bigr) \leq
c_1\bigl(\Sigma(H),\min\{d_{\Sigma(H)},\alpha\}\bigr)\lesssim_d 1.
\]

Next, we prove the first upper bound in~\eqref{eq:ALNRRV-L2-embedding}.
Chang, Naor, and Ren~\cite{CNR}
(improving on~\cite{ALN}) proved that any $k$-point subset of $L_1$ embed in $L_2$ with $O(\sqrt{\log k})$ distortion and therefore, together with~\eqref{eq:ALNRRV-L1-embedding},
\[
c_2\bigl(S,\min\{d_{\Sigma(L)},\alpha\}\bigr)\lesssim_d \sqrt{\log |S|}.
\]

Lastly, we prove the second upper bound
in~\eqref{eq:ALNRRV-L2-embedding}.
Let $W\subset U$ be a minimal subset (under containment) of $U$ such that \(S\subseteq \Sigma(W).\)
Obviously $|W|\le 2|S|$, and $S\subseteq \Sigma(W)\subseteq \Sigma(H)$.
The bound is achieved by embedding only $\Sigma(W)$ (instead of $\Sigma(H)$) directly in $L_2$, avoiding intermediate embedding in $L_1$.
However, in this case, we will first embed $W$
and then extend that embedding to an embedding of $\Sigma(W)$
using \autoref{lem:biLip-extension}.

Fix $T\in\mathcal T(H)$ to be a spanning tree of $H$.
By Matou\v{s}ek tree embedding~\cite{Mat-trees} and the rigidity of Euclidean metrics under truncations~\cite[Lemma~5.2]{MN-quotients}
there exists%
\footnote{The embedding $\Phi_{T,W}$ embeds only the (tree) metric induced on $W$ by $T$, and not $T$ itself.
Otherwise, Matou\v{s}ek's distortion bound would only give $O(\sqrt{\log\log |T|})$, and $|T|$ might be too large for our purpose.}
$\phi_{T,W}:W\to L_2$
such that for any $u,v\in W$
\begin{equation} \label{eq:Mat-tree-embed-1}
\min\{d_{T}(u,v),\alpha\}\le \|\phi_{T,W}(u)-\phi_{T,W}(v)\|_2 \lesssim
\min\Bigl\{\sqrt{\log\log |W|}\, d_{T}(x,y),\alpha\Bigr\}.
\end{equation}
Recall that
$\mu$ is a probability measure over $\mathcal T(\Sigma(H))$ as in \autoref{def:L}.
The Euclidean embedding $\phi:(W,d_H) \to L_2(\mu;L_2)$ of $H$ is defined
for $u\in W$ as
\[ \big[\phi(u)\big](T)=\phi_{T,W}(u).\]

Fix $x,y\in W$. In one direction,
\[
\min\{d_{H}(x,y),\alpha\}^2
\stackrel{\eqref{eq:Mat-tree-embed-1}}\leq
\int_{\mathcal T(H)} \|\phi_{T,W}(x)-\phi_{T,W}(y)\|_2^2
\,\mathsf d\mu(T) = \|\phi(x)-\phi(y)\|^2_{L_2(\mu;L_2)}.
\]
In the other direction, we bound the Lipschitz constant of $\phi$ similarly to~\eqref{eq:aa}. It is sufficient to bound the Lipschitz constant for edges $(x,y)\in E_L(U)$. So
\begin{equation} \label{eq:ab}
\begin{aligned}
\|\phi(x)-\phi(y)\|_{L_2(\mu;L_2)}^2 &=
\int_{\mathcal T(H)} \|\phi_{T,W}(x)-\phi_{T,W}(y)\|_2^2
\,\mathsf d\mu(T)
\\ & \stackrel{\eqref{eq:Mat-tree-embed-1}}\lesssim
\int_{\mathcal T(H)} \min\{(\log\log |W|)\cdot  d^2_T(x,y),\alpha^2\}\,\mathsf d\mu(T)
\\ & \leq
\mu(\{T:(x,y)\subseteq T\})\cdot \min\{\log\log |W|\,,\,\alpha^2\}  +  \mu(\{T:(x,y)\not\subseteq T\}) \cdot \alpha^2
\\ &  \stackrel{\eqref{eq:almost-cover-tree}}\leq
\min\{\log\log |W|\,,\,\alpha^2\}
+  {C_d \alpha^2}/{\log n} .
\\ & \lesssim \max\{\log\log |S|\,,\, C_d\alpha^2/\log n \}.
\end{aligned}
\end{equation}
Applying \autoref{lem:biLip-extension} to the graph $H=(U,E_L(U))$, the subset $W\subseteq U$, and the mapping $\phi$, we deduce that there exists a bi-Lipschitz embedding $\Phi:\Sigma(W)\to L_2$, concluding the proof.
\end{proof}

\begin{remark}
The proof of the Euclidean embedding~\eqref{eq:ALNRRV-L2-embedding} in
\autoref{lem:subset-L2-embedding}
used the special structure of the representation of the graph metric by
a convex combination of spanning trees: each edge of the graph appears in all the spanning trees except a subset whose total weight is at most $c/\log n$.
Far more popular is the more general concept of $C$-approximation of a metric space $(X,d_X)$ by a probability distribution $(\Omega,\mu)$ on dominating tree metrics~\cites{Bartal-trees}{Bartal-trees-2}{FRT}. I.e.,
\begin{inparaenum}[(i)]
    \item any $T\in\Omega$ is a tree metric on the point-set $X$;
    \item $d_X(x,y)\le d_T(x,y)$ for any $x,y\in X$ and $T\in\Omega$;
    \item $\int_{\Omega} d_T(x,y)\mathsf{d}\mu(T)\le C\cdot d_X(x,y)$
    for any $x,y\in X$.
\end{inparaenum}
Despite their popularity, not much is known structurally about this class of metric spaces.
Concretely, we ask:
\end{remark}

\begin{question}
Let $M$ be an $n$-point metric space $C$-approximated by a probability distribution on dominating tree metrics, where $C$ is constant independent of $n$.
Is $c_2(M)=o(\sqrt{\log n})$?
\end{question}

\begin{lemma}
Let $G\in L_K^{n,d}$, and $p\ge 1$. Then%
\footnote{Technically, the Poincar\'e inequality~\eqref{eq:p-poincare-metric} was defined for regular graphs, whereas here $G$ only has an upper bound of $d$ on the degrees. This can be fixed by adding self loops to $G$ that make it $d$-regular.
The additional self-loops do not affect the shortest-path metric in $G$ and do not affect~\eqref{eq:cheeger} and~\eqref{eq:almost-cover-tree}.
Therefore, they are immaterial to the proof.}
\begin{align} \label{eq:PIp-L_K^nd}
\gamma_p(G,\ell_p)&\lesssim_p (dK)^p.
\end{align}
\end{lemma}
\begin{proof}
The proof is a standard application of Matou\v{s}ek extrapolation theorem (or Cheeger inequality for graphs~\cites{Dodziuk}{AM}{Alon-cheeger} for the special case $p=2$).
We provide the standard details for completeness.
By setting $S$ to be the smaller (in cardinality) of the two subsets $f^{-1}(0)$ and $f^{-1}(1)$ for a non-constant $f:V\to \{0,1\}$,
we have
\begin{multline*}
\gamma_1(G,\{0,1\})= \sup_{\substack{f:V\to \{0,1\}\\
\text{non-const}}}
\frac
{\frac{1}{n^2} \sum_{u,v\in V} |f(u)-f(v)|}
{\frac{2}{nd}\sum_{\{u,v\}\in E} |f(u)-f(v)|}
\\
=\frac{d}{2n}\,\sup_{\substack {S\subset V\\ 0<|S|\le n/2}}
\frac{|S|\cdot |V\setminus S|}{|E(S,V\setminus S)|}
 \leq
\frac{d}{2}\,\sup_{\substack {S\subset V\\ 0<|S|\le n/2}}
\frac{|S|}{|E(S,V\setminus S)|}
\stackrel{\eqref{eq:cheeger}}{\le} \frac{dK}{2}.
\end{multline*}

Since for $f:V\to \mathbb R$, the metric $d(u,v)=|f(u)-f(v)|$
can be written as a conic combination of cut metrics on $V$, we
have
\begin{equation} \label{eq:PI1-G-R}
\gamma_1(G,\mathbb R)= \gamma_1(G,\{0,1\})\leq dK/2.
\end{equation}
Applying Matou\v{s}ek extrapolation theorem~\eqref{eq:mat-extrapolation}
on $\gamma_1(G,\mathbb R)$,
\[ \gamma_p(G,\mathbb R) \stackrel{\eqref{eq:mat-extrapolation}}\lesssim \gamma_1(G,\mathbb R)^p
\stackrel{\eqref{eq:PI1-G-R}}\leq (dK/2)^p.
\]
Using a standard argument of summation over coordinates,
it is straightforward to check that $\gamma_p(G,\mathbb R)=\gamma_p(G,\ell_p)$ for any regular graph $G$.
\end{proof}

\begin{proof}[Proof of\/ \autoref{thm:no-embedding-dichotomy-ell_2} and \autoref{thm:no-embedding-dichotomy-ell_1}]
We prove both theorems using the same construction.
Fix a metric transform $\varphi:[0,\infty) \to [0,\infty)$ and $d=4$.
Normalizing $\varphi$ by dividing it by $\varphi(1)$,
we may assume without loss of generality that $\varphi(1)=1$.
Let $K$, $n_0$ and $(G_n)_{n>n_0}$ where $G_n\in L_K^{n,d}$
from \autoref{lem:L_K^{n,d}}.
Fix $G=(V,E)=G_n$.
Let
\[
\hat X_G=\Bigl(\Sigma(G), \tfrac{2\pi}{\varphi(\girth(G))} d_{\Sigma(G)}\Bigr).
\]
Observe that
\[
\frac{2\pi}{\varphi(\girth(G))} = \delta \frac{2\pi}{\girth(G)}
\]
for the parameter $\delta = \girth(G)/\varphi(\girth(G))$ which satisfies $\delta\geq1$ since $\varphi$ is concave with $\varphi(1)=1$. Thus, similarly to \autoref{lem:X_delta-hadamard}, we deduce that
$\cone(\hat X_G)$ is a Hadamard space and
\begin{equation} \label{eq:expander-snowflake-in-X}
c_{\cone(\hat X_G)}((\Sigma(G),\min\{\varphi(\girth(G)),d_{\Sigma(G)}\}))\lesssim 1
\end{equation}
since $G\in L_K^{n,d}$.  The space $ X_\varphi$ is defined to be their $\ell_1$-union:
\begin{equation}\label{eq:defX-O-product}
X_\varphi=\Bigl(\biguplus_{n>n_0} \bigl(\cone(\hat X_{G_n}), 0_{\cone(\hat X_{G_n})}\bigr)\Bigr)_1,
\end{equation}
where $0_{\cone(\hat X_G)}$ is the cusp of
$\cone(\hat X_G)$.
By \autoref{lem:ell1-union}, $X$ is a Hadamard space.

We first prove that $D_k(X_\varphi,\ell_p)\gtrsim_p \varphi(\log k)$, $p\in[1,\infty)$.
From~\eqref{eq:expander-snowflake-in-X} it is sufficient
to show that
for any $G=G_n=(V,E)$, we have
\[
c_p((V,\min\{\varphi(\girth(G)),d_G\}))\gtrsim_p \varphi (\log n).
\]
Fix $f:V\to \ell_p$.
By scaling $f$,
assume without loss of generality that $f$ is $1$-Lipschitz, i.e.,
\begin{equation} \label{eq:dist-f}
\|f(u)-f(v)\|_p \le \min \{\varphi(\girth(G)),d_G(u,v)\}
\le \mathrm{dist}(f)\cdot \|f(u)-f(v)\|_p,
\end{equation}
where $\mathrm{dist}(f)$ is the distortion of $f$.
Applying the $p$-Poincar\'e inequality~\eqref{eq:p-poincare-metric} for $G$ on $f$,

\begin{multline}\label{eq:PI-for-snowflake}
\frac{1}{\mathrm{dist}(f)^p n^2} \sum_{u,v\in V} \min\{\varphi(\girth(G)),d_G(u,v)\}^p
\\ \stackrel{\eqref{eq:dist-f}}\le
\frac1{n^{2}} \sum_{u,v\in V}\|f(u)-f(v)\|_p^p
\stackrel{\eqref{eq:PIp-L_K^nd}}\lesssim
\frac{1}{nd}\sum_{(u,v)\in E} \|f(u)-f(v)\|_p^p
\stackrel{\eqref{eq:dist-f}}\le 1.
\end{multline}
Using a simple counting argument,
the set $\{(u,v):d_G(u,v)\ge \girth(G)/4\}$
contains $(1-o(1))n^2$ of the pairs in $V\times V$.
Furthermore, since $\varphi(0)=0$, $\varphi(1)=1$, and $\varphi$ is concave,
$\varphi(x)\le x$ for $x\ge 1$.
Therefore,
the left-hand side of~\eqref{eq:PI-for-snowflake} is at least
$\Omega\bigl(\frac{\varphi(\girth(G))^{p}}{\mathrm{dist}(f)^p}\bigr)$.
Hence,
\[ \mathrm{dist}(f)\gtrsim \varphi(\girth(G))
\stackrel{\eqref{eq:hi-girth-log-diam}}\gtrsim\varphi(\log n). \]
The last inequality used the fact that the graphs $G_n$ have a girth of the order of $\log n$ along with the fact that for a metric transform
$\varphi$ and $\alpha\in(0,1]$, $\varphi(\alpha x)\ge \alpha \varphi(x)$.

\medskip

We next prove that $D_k(X_\varphi,\ell_p)\lesssim\varphi(\log (k+1))$, when either $p=1$ and $\varphi$ is any metric transform satisfying $\varphi(1)=1$; or $p=2$ and $\varphi$ is any metric transform satisfying $\varphi(1)=1$ and $\varphi(t)\gtrsim \sqrt{\log t}$ for $t\geq 2$.
Since $X_\varphi$ is an $\ell_1$-union (\autoref{def:orthogonal-union}), by
\autoref{lem:ell1-union}
it is sufficient to prove
\[\sup_{n>n_0} D_k(\cone(\hat X_{G_n})\embed\ell_p)\lesssim\varphi(\log (k+1)).\]
Moreover, by~\eqref{eq:1-cos identity}, the cone of a metric space $(A,d_A)$ is isometric to the cone of $(A,\min\{d_A,\pi\})$ and so by~\eqref{eq:cone(Lp) in Lp} it further suffices to prove the distortion bound.

\[ \sup_{n>n_0} D_k\bigl ((\Sigma(G_n),\min\{\varphi(\girth(G_n)), d_{\Sigma(G_n)}\})\embed\ell_p\bigr)\lesssim\varphi(\log (k+1)).\]

We fix an $n$-vertex graph $G=G_n=(V,E)$,
and a $k$-point subset
$S\subset \Sigma(G)$, $|S|=k$.
We consider two cases, depending on the size of $k$.

\begin{asparadesc} 
\item[$k\ge \sqrt{n}$:]
In this case, we present a simple embedding of
\((\Sigma(G),\min\{d_{\Sigma(G)},\varphi(\girth(G))\})\)
in $\ell_p(V)$ with a distortion of at most ${2}^{1/p}\varphi(\girth(G))$.
The embedding
$f:\Sigma(G)\to \ell_p(V)$ is defined as follows
\[
f(x)_{u}=
\begin{cases}
2^{-1/p}(1-d_{\Sigma(G)}(x,u)) & \exists v\in V,\ \{u,v\}\in E,\ x\in[u,v]\subset\Sigma(G)\\
0 & \text{ otherwise.}
\end{cases}
\]
It is straightforward to check that this embedding is $1$-Lipschitz and the inverse Lipschitz constant is at most
$\max\{2,2^{1-1/p}\cdot \varphi(\girth(G))\}$.
Therefore,
\begin{multline} \label{eq:hat-X-cobe(L2)}
c_{p}((S,\min\{\varphi(\girth(G)), d_{\Sigma(G)}\}))
\le
c_{p}((\Sigma(G),\min\{\varphi(\girth(G)), d_{\Sigma(G)}\}))
\\ \leq 2\varphi( \girth(G))
\stackrel{\eqref{eq:hi-girth-log-diam}}{\lesssim}
\varphi(\log n) \le 2\varphi(\log k).
\end{multline}

\item[$k< \sqrt{n}$:]
Consider the embedding in $\ell_1$ (\autoref{thm:no-embedding-dichotomy-ell_1}) and in $\ell_2$ (\autoref{thm:no-embedding-dichotomy-ell_2}) separately.
\begin{itemize}
\item 
The embedding into $\ell_1$ is a straightforward application of \autoref{lem:subset-L2-embedding}
with $\alpha=\varphi(\girth(G))$, as
\[
c_1(S, \min\{\varphi(\girth(G)), d_{\Sigma(G)}\})
\stackrel{\eqref{eq:ALNRRV-L1-embedding}}\lesssim 1
= \varphi(1)\leq \varphi(\log(k+1)).
\]
\item
For embedding in $\ell_2$,
recall that we further assume that $\varphi(t)\gtrsim \sqrt{\log t}$.
Denote $\beta=\log (k+1)$. 
First assume that $\varphi(\beta)\lesssim \sqrt{\log n} $.
By~\eqref{eq:ALNRRV-L2-embedding},
\begin{equation*}
c_{2}((S,\min\{\varphi(\beta), d_{\Sigma(G)}\}))
\stackrel{\eqref{eq:ALNRRV-L2-embedding}}
\lesssim 
\max \Bigl\{ \sqrt{\log \beta}\,,\,   \frac{\varphi(\beta)}{\sqrt{\log n}} \Bigr\}
=\sqrt{\log \beta} \lesssim \varphi(\beta).
\end{equation*}
Next assume that $\varphi(\beta)\gtrsim \sqrt{\log n} $.
    In this case,
\begin{equation}\label{eq:theta>1/2}
c_{2}((S,\min\{\varphi(\beta), d_{\Sigma(G)}\}))
\stackrel{\eqref{eq:ALNRRV-L2-embedding}}
\lesssim \sqrt{\beta} \leq \sqrt{\log n} \lesssim \varphi(\beta). \qedhere
\end{equation}
\end{itemize}
\end{asparadesc}
 \end{proof}

\begin{remark}
\label{rem:D(X,ell_p)}
\autoref{thm:no-embedding-dichotomy-ell_2} can be generalized to embedding in $\ell_p$ for any fixed $p\in[1,\infty)$ in the following straightforward way. The space $X_\varphi$ is the same.
The lower bound $D_k(X_\varphi\embed \ell_p)\gtrsim_p \varphi(\log (k+1))$
is actually proved in the above proof for any $p\in[1,\infty)$.
The upper bound is a consequence of (for example) 
Dvoretzky theorem: 
\(
D_k(X_\varphi\embed \ell_p)\leq 
D_k(X_\varphi\embed \ell_2)\cdot D_k(\ell_2\embed \ell_p)
=D_k(X_\varphi\embed \ell_2) \asymp \varphi(\log (k+1)).
\)
\end{remark}

\newpage
\backgroundsetup{contents={}}
\printbibliography   

\end{document}